\documentclass[a4paper,UKenglish,cleveref, autoref, thm-restate]{lipics-v2021}
\usepackage{hyperref}

\usepackage{amsmath,amssymb,amsfonts}
\usepackage{graphicx}
\usepackage{textcomp}
\usepackage[binary-units=true]{siunitx}
\usepackage{comment}
\usepackage{amsmath}
\usepackage{etoolbox}
\usepackage{color,colortbl}
\usepackage{mathtools}
\usepackage{enumerate}
\usepackage{algorithm, algpseudocode}
\usepackage{booktabs}
\usepackage[normalem]{ulem}
\usepackage{pgf}
\usepackage{pgfplotstable}
\usepackage{tikz,pgfplots}
\usetikzlibrary{trees,decorations,arrows,arrows.meta,automata,shadows,positioning,plotmarks,backgrounds,shapes,shapes.misc}
\usetikzlibrary{calc,matrix,fit,petri,decorations.pathmorphing,patterns}
\usetikzlibrary{decorations.pathreplacing,decorations.markings,shapes.geometric,calc}
\usetikzlibrary{tikzmark}
\usepackage{paralist}
\usepackage{stmaryrd}
\usepackage{xspace}
\usepackage{graphicx}
\usepackage{float}
\usepackage[utf8]{inputenc} 
  \usepackage{csquotes} 
\usepackage{multirow}
\usepackage{array}
\usepackage{dsfont}
 \usepackage{array}
\usepackage{xifthen}
\usepackage{relsize}
\usepackage{xfrac}
\usepackage{wrapfig}
\usepackage{rotating}
\usepackage{paralist}

\usepackage{longtable}
\usepackage{caption}
\usepackage[position=b]{subcaption}

\colorlet{darkgreen}{green!80!black}
\colorlet{darkred}{red!80!black}
\usetikzlibrary{arrows, automata, shapes}
\tikzset{auto, >= stealth}
\tikzset{every edge/.append style={thick, shorten >= 1pt}}
\tikzset{initial/.style={draw, thick, <-, shorten <=1pt}}
\tikzset{player0/.style = {draw, thick, shape=circle, minimum size=5mm}}
\tikzset{player1/.style = {draw, thick, shape=rectangle, minimum size=5mm}}

\usepackage{marginnote}
\usepackage{schemata}

\usepackage[
colorinlistoftodos, %enable a coloured square in the list of todos
textwidth=\marginparwidth, %set the width of the todonotes
textsize=scriptsize, %size of the text in the todonotes
]{todonotes}

\newcommand{\gamegraph}{G}
\newcommand{\spec}{\Phi}

\makeatletter
\newsavebox{\@brx}
\newcommand{\llangle}[1][]{\savebox{\@brx}{\(\m@th{#1\langle}\)}%
  \mathopen{\copy\@brx\kern-0.5\wd\@brx\usebox{\@brx}}}
\newcommand{\rrangle}[1][]{\savebox{\@brx}{\(\m@th{#1\rangle}\)}%
  \mathclose{\copy\@brx\kern-0.5\wd\@brx\usebox{\@brx}}}

\makeatother

\newcommand{\set}[1]{\left\lbrace #1\right\rbrace}

\newcommand{\lang}{\mathcal{L}}

\newcommand{\A}{\mathcal{A}}
\newcommand{\Sc}{\mathcal{S}}
\newcommand{\B}{\mathcal{B}}

\newcommand{\fontsmall}{\fontsize{7pt}{6pt}\selectfont}

\newcommand{\ev}[2]{\llbracket {#1}, {#2} \rrbracket}
\newcommand{\1}{\mathbb{I}}
\newcommand{\te}{\mathbf{t}}

\newcommand{\Lpre}{\mathsf{Lpre}}
\newcommand{\Apre}{\mathsf{Apre}}
\newcommand{\Cpre}{\mathsf{Cpre}}
\newcommand{\Npre}{\mathsf{Npre}}
\newcommand{\Pre}{\mathsf{Pre}}
\newcommand{\SOLVE}{\mathsf{SOLVE}}
\newcommand{\SafeReach}{\mathsf{SafeReach}}
\newcommand{\e}{E'}

\newcommand{\Even}{\ensuremath{\textsf{Even}}\xspace}
\newcommand{\Odd}{\ensuremath{\textsf{Odd}}\xspace}
\newcommand{\Ve}{\ensuremath{V_{\Even}}\xspace}
\newcommand{\Vo}{\ensuremath{V_{\Odd}}\xspace}
\newcommand{\We}{\ensuremath{\mathcal{W}_{\Even}}\xspace}
\newcommand{\Wo}{\ensuremath{\mathcal{W}_{\Odd}}\xspace}
\newcommand{\bb}{\ensuremath{\Lambda}\xspace}
\newcommand{\nb}{\ensuremath{{\neg\bb}}\xspace}
\newcommand{\Xsr}{\ensuremath{{\mathcal{X}}}\xspace}

\newcommand{\Vl}{V^{\ell}}

\newcommand{\tup}[1]{\left( #1\right)}
\newcommand{\ltup}[1]{\langle #1\rangle}

\newcommand{\rank}[1]{\mathop{\mathsf{rank}\ifthenelse{\isempty{#1}}{}{(#1)}}}
\newcommand{\ranko}[1]{\mathop{\overline{\mathrm{rank}}\ifthenelse{\isempty{#1}}{}{(#1)}}}

\makeatletter 
\newif\ifFIRST
\newif\ifSECOND
\let\LISTOP\relax
\newcommand{\List}[4][\;]{#3#1%
        \FIRSTtrue
        \@for\i:=#2\do{%
        \ifFIRST\LISTOP{\i}\FIRSTfalse\else,\LISTOP{\i}\fi%
        }%
        #1#4%
        \let\LISTOP\relax
}
\makeatother

\makeatletter

\newcommand{\propNeg}{\@ifstar\propNegStar\propNegNoStar}
\newcommand{\propNegStar}[1]{\ensuremath{\left(\propNegNoStar{#1}\right)}}
\newcommand{\propNegNoStar}[2][\cdot]{\ensuremath{\neg\ifthenelse{\isempty{#2}}{#1}{#2}}}

\newcommand{\propConj}{\@ifstar\propConjStar\propConjNoStar}
\newcommand{\propConjStar}[2]{\ensuremath{\left(\propConjNoStar{#1}{#2}\right)}}
\newcommand{\propConjNoStar}[3][\cdot]{\ensuremath{\ifthenelse{\isempty{#2}}{#1}{#2}\wedge\ifthenelse{\isempty{#3}}{#1}{#3}}}

\newcommand{\propDisj}{\@ifstar\propDisjStar\propDisjNoStar}
\newcommand{\propDisjStar}[2]{\ensuremath{\left(\propDisjNoStar{#1}{#2}\right)}}
\newcommand{\propDisjNoStar}[3][\cdot]{\ensuremath{\ifthenelse{\isempty{#2}}{#1}{#2}\vee\ifthenelse{\isempty{#3}}{#1}{#3}}}

\newcommand{\propImp}{\@ifstar\propImpStar\propImpNoStar}
\newcommand{\propImpStar}[2]{\ensuremath{\left(\propImpNoStar{#1}{#2}\right)}}
\newcommand{\propImpNoStar}[3][\cdot]{\ensuremath{\ifthenelse{\isempty{#2}}{#1}{#2}\Rightarrow\ifthenelse{\isempty{#3}}{#1}{#3}}}

\newcommand{\propAequ}{\@ifstar\propAequStar\propAequNoStar}
\newcommand{\propAequStar}[2]{\ensuremath{\left(\propAequNoStar{#1}{#2}\right)}}
\newcommand{\propAequNoStar}[3][\cdot]{\ensuremath{\ifthenelse{\isempty{#2}}{#1}{#2}\Leftrightarrow\ifthenelse{\isempty{#3}}{#1}{#3}}}

\newcommand{\AllQ}{\@ifstar\AllQStar\AllQNoStar}
\newcommand{\AllQStar}[3][\;]{\ensuremath{\left(\forall #2#1.#1#3\right)}}
\newcommand{\AllQNoStar}[3][\;]{\ensuremath{\forall #2#1.#1#3}}
\newcommand{\AllQu}{\@ifstar\AllQuStar\AllQuNoStar}
\newcommand{\AllQuStar}[3][\;]{\ensuremath{\left(\forall^{\infty} #2#1.#1#3\right)}}
\newcommand{\AllQuNoStar}[3][\;]{\ensuremath{\forall^{\infty} #2#1.#1#3}}

\newcommand{\ExQ}{\@ifstar\ExQStar\ExQNoStar}
\newcommand{\ExQStar}[3][\;]{\ensuremath{\left(\exists #2#1.#1#3\right)}}
\newcommand{\ExQNoStar}[3][\;]{\ensuremath{\exists #2#1.#1#3}}

\newcommand{\NExQ}{\@ifstar\NExQStar\NExQNoStar}
\newcommand{\NExQStar}[3][\;]{\ensuremath{\left(\nexists #2#1.#1#3\right)}}
\newcommand{\NExQNoStar}[3][\;]{\ensuremath{\nexists #2#1.#1#3}}

\newcommand{\UniqueQ}{\@ifstar\UniqueQStar\UniqueQNoStar}
\newcommand{\UniqueQStar}[3][\;]{\ensuremath{\left(\exists! #2#1.#1#3\right)}}
\newcommand{\UniqueQNoStar}[3][\;]{\ensuremath{\exists! #2#1.#1#3}}

  \newlength{\SFS@HEIGHT}
  \newlength{\SFS@WIDTH}
  \newcommand{\SplitX}[2]{
            \settoheight{\SFS@HEIGHT}{$#2$}
            \settowidth{\SFS@WIDTH}{$#2$}
            \mbox{\begin{tikzpicture}[baseline=(current bounding box.center)]
            \node[] (E) at (0,0) {$#1$};
            \node[inner sep=0pt] (F) at ($(E.south west)+(1ex,-1ex)+(3ex+.5\SFS@WIDTH,-\SFS@HEIGHT)$) {$#2$};
            \node[] (E) at (0,0) {\phantom{$#1$}};
            \draw[fill] ($(E.east)+(1ex,0ex)$) circle (.2ex);
            \draw[-] ($(E.east)+(1ex,0ex)$) -- ($(E.south east)+(1ex,-0.5ex)$) -- ($(E.south west)+(1ex,-0.5ex)$) -- ($(E.south west)+(1ex,-1ex)-(0,\SFS@HEIGHT)$) -- ($(E.south west)+(2.5ex,-1ex)-(0,\SFS@HEIGHT)$);
            \draw[fill] ($(E.south west)+(2.5ex,-1ex)-(0,\SFS@HEIGHT)$) circle (.2ex);
            \end{tikzpicture}}}
  \newcommand{\SplitS}[2]{
            \settoheight{\SFS@HEIGHT}{$#2$}
            \settowidth{\SFS@WIDTH}{$#2$}
            \mbox{\begin{tikzpicture}[baseline=(current bounding box.center)]
            \node[] (E) at (0,0) {$#1$};
            \node[inner sep=0pt] (F) at ($(E.south west)+(1ex,0.5ex)+(0ex+.5\SFS@WIDTH,-\SFS@HEIGHT)$) {$#2$};
            \end{tikzpicture}}}

\newcommand{\VSet}[2][]{\let\LISTOP\val\List[#1]{#2}{\{}{\}}}

\newcommand{\VTuple}[2][]{\let\LISTOP\val\List[#1]{#2}{(}{)}}
\newcommand{\UNION}{\@ifstar\UNIONStar\UNIONNoStar}
\newcommand{\UNIONStar}[2]{\ensuremath{\left(\UNIONNoStar{#1}{#2}\right)}}
\newcommand{\UNIONNoStar}[2]{\ensuremath{\ifthenelse{\isempty{#1}}{\cdot}{#1}\cup\ifthenelse{\isempty{#2}}{\cdot}{#2}}}

\newcommand{\UNIOND}{\@ifstar\UNIONDStar\UNIONDNoStar}
\newcommand{\UNIONDStar}[2]{\ensuremath{\left(\UNIONDNoStar{#1}{#2}\right)}}
\newcommand{\UNIONDNoStar}[2]{\ensuremath{\ifthenelse{\isempty{#1}}{\cdot}{#1}\uplus\ifthenelse{\isempty{#2}}{\cdot}{#2}}}

\newcommand{\SETMINUS}{\@ifstar\SETMINUSStar\SETMINUSNoStar}
\newcommand{\SETMINUSStar}[2]{\ensuremath{\left(\SETMINUSNoStar{#1}{#2}\right)}}
\newcommand{\SETMINUSNoStar}[2]{\ensuremath{\ifthenelse{\isempty{#1}}{\cdot}{#1}\setminus\ifthenelse{\isempty{#2}}{\cdot}{#2}}}

\newcommand{\INTERSECT}{\@ifstar\INTERSECTStar\INTERSECTNoStar}
\newcommand{\INTERSECTStar}[2]{\ensuremath{\left(\INTERSECTNoStar{#1}{#2}\right)}}
\newcommand{\INTERSECTNoStar}[2]{\ensuremath{\ifthenelse{\isempty{#1}}{\cdot}{#1}\cap\ifthenelse{\isempty{#2}}{\cdot}{#2}}}

\newcommand{\CARTPROD}{\@ifstar\CARTPRODStar\CARTPRODNoStar}
\newcommand{\CARTPRODStar}[2]{\ensuremath{\left(\CARTPRODNoStar{#1}{#2}\right)}}
\newcommand{\CARTPRODNoStar}[2]{\ensuremath{\ifthenelse{\isempty{#1}}{\cdot}{#1}\times\ifthenelse{\isempty{#2}}{\cdot}{#2}}}

\newcommand{\FINCOUNT}{\@ifstar\FinCountStar\FinCountNoStar}
\newcommand{\FinCountStar}[1]{\ensuremath{\#(\ifthenelse{\isempty{#1}}{\cdot}{#1})}}
\newcommand{\FinCountNoStar}[1]{\ensuremath{\#\left(\ifthenelse{\isempty{#1}}{\cdot}{#1}\right)}}

\makeatother

\title{Solving Odd-Fair Parity Games}%TODO Please add

\author{Irmak Sa\u{g}lam}{Max Planck Institute for Software Systems (MPI-SWS), Kaiserslautern, Germany}{isaglam@mpi-sws.org}{[orcid]}{}

\author{Anne-Kathrin Schmuck}{Max Planck Institute for Software Systems (MPI-SWS), Kaiserslautern, Germany}{akschmuck@mpi-sws.org}{[orcid]}{}

\authorrunning{I. Sa\u{g}lam, A.-K. Schmuck} %TODO mandatory. First: Use abbreviated first/middle names. Second (only in severe cases): Use first author plus 'et al.'

\Copyright{Irmak Sa\u{g}lam and Anne-Kathrin Schmuck} %TODO mandatory, please use full first names. LIPIcs license is "CC-BY";  http://creativecommons.org/licenses/by/3.0/

 %TODO mandatory; please add comma-separated list of keywords

\funding{This work was partially supported by the DFG projects SCHM 3541/1-1 and 389792660 TRR 248–CPEC.}

\acknowledgements{We are grateful for the immense support provided by Munko Tsyrempilon for the experimental validation.}

\nolinenumbers %uncomment to disable line numbering

\hideLIPIcs  %uncomment to remove references to LIPIcs series (logo, DOI, ...), e.g. when preparing a pre-final version to be uploaded to arXiv or another public repository

\begin{document}

\maketitle

\begin{abstract}
  This paper discusses the problem of efficiently solving parity games where player \Odd~has to obey an additional \emph{strong transition fairness constraint} on its vertices -- given that a player \Odd~vertex $v$ is visited infinitely often, a particular subset of the outgoing edges (called \emph{live edges}) of $v$ has to be taken infinitely often. Such games, which we call \emph{\Odd-fair parity games}, naturally arise from abstractions of cyber-physical systems for planning and control. 
%   and can currently not satisfactory been dealt with.}
  %for planning and control. %as well as in planning and in resource management.\todo{todo: if we don't add the explanation about resource planning, we should remove it from here}

  In this paper, we present a new Zielonka-type algorithm for solving \Odd-fair parity games. This algorithm not only shares \emph{the same worst-case time complexity} as Zielonka's algorithm for (normal) parity games but also preserves the algorithmic advantage Zielonka's algorithm possesses over other parity solvers with exponential time complexity. %, such as the recently introduced \emph{fixed-point algorithm} for \Odd-fair parity games by Banerjee et. al.~\cite{banerjee2022fast}.
 
 We additionally introduce a formalization of \Odd~player winning strategies in such games, which were unexplored previous to this work. %To the best of our knowledge, this is the first time such a formalization has been introduced.  %We represent these strategies as \enquote{almost positional} strategy templates, and prove their existence from all player \Odd~winning vertices. 
 This formalization serves dual purposes: firstly, it enables us to prove our Zielonka-type algorithm; secondly, it stands as a noteworthy contribution in its own right, augmenting our understanding of additional fairness assumptions in two-player games.\end{abstract}

\vspace*{-0.3cm}
\section{Introduction}
\vspace{-1mm}

\emph{Parity games} are a canonical representation of $\omega$-regular two-player games over finite graphs, which arise from many core computational problems in the context of correct-by-construction synthesis of reactive software or hardware. In particular, two player games on graphs have been extensively used in the context of cyber-physical system design \cite{Tabuada2009,belta2017formal}, showing their practical importance. 
% 
% On the other hand, parity games are interesting from a theoretical point of view as they lie in the complexity class  $UP \cap coUP$, which is contained in $NP \cap coNP$~\cite{JurdzinskiUPcoUP,EmersonJutlaModelCheckingMuCalculus}. %It is widely believed that there exists a polynomial time algorithm to solve them, but non is known so far. 
% Further, parity games are polynomial-time equivalent to the model checking problem of modal $\mu-$calculus \cite{EmersonJutlaModelCheckingMuCalculus,ColinModelChecking,Walukiewicz_muCalculusMSOL}. 
% \vspace{-1mm}
% 
\emph{Fairness}, on the other hand, is a property that widely occurs in this context - both as a desired property to be enforced (e.g., requiring a synthesized scheduler to fairly serve its clients), as well as a common assumption on the behavior of other components (i.e., assuming the network to always eventually deliver a data packet). 
While \emph{strong fairness} encoded by a Streett condition necessarily incurs a high additional cost in synthesis \cite{EJ99}, it is known that the general reactivity(1) (GR(1)) fragment of linear temporal logic (LTL)~\cite{BJPPS2006} allows for efficient synthesis in the presence of very restricted fairness conditions. Due to its efficiency, it is extensively used in the context of cyber-physical system design, e.g.~\cite{WEK18,AMT2013,MR2015,KFP2007,KFP2009,SKCCB2015}. %The GR(1) fragment was particularly well received in the robotics and cyber-physical systems community and was extensively used to synthesize controllers for physical systems \cite{}.
Despite the omnipresence of fairness in such synthesis problems and the success of the GR(1) fragment, not much else is known about tractable fairness constraints in synthesis via two player games on graphs. A notable exception is the recent work by Banerjee et.\ al.~\cite{banerjee2022fast} which considers the sub-class of \emph{strong \textbf{transition} fairness assumptions}~\cite{QS83,Francez,baierbook} which require that whenever the environment player vertex $v$ is visited infinitely often, a particular subset of the outgoing edges (called \emph{live edges}) of $v$ has to be taken infinitely often. In other words, \emph{strong \textbf{transition} fairness assumptions} limit \emph{strong fairness assumptions} to individual transitions.
% The work by Banerjee et. al.~\cite{banerjee2022fast} shows that this limitation allows to handle them efficiently in synthesis.
% While this result makes these fairness conditions theoretically interesting, they are also practically important as they are known to
Despite their limited expressive power, such restricted fairness constrains do naturally arise in resource management \cite{CAFMR13}, in abstractions of continuous-time physical processes for planning \cite{CPRT03, DTV99, PT01, DIRS18, RS14, AGR20} and controller synthesis \cite{thistle1998control, NOL17, MMSS2021}, which makes them interesting to study.
% \todo{I filled in citations and cited that work you recommended. In planning we have plenty papers, hoerver on resource management I only have Rupak's paper. And about that I have the concerns I texted you about from mattermost. should we still cite it? }
%In particular, Banerjee et. al.~\cite{banerjee2022fast} consider $\omega$-regular games where the environment player has to obey additional \emph{strong transition fairness constraints} on its vertices  -- i.e., if the environment player vertex $v$ is visited infinitely often, a particular subset of the outgoing edges (called \emph{live edges}) of $v$ has to be taken infinitely often.
 
Concretely, Banerjee et.\ al.~\cite{banerjee2022fast} show that \emph{parity games} with strong transition fairness assumptions on player \Odd\ -- which we call \emph{\Odd-fair parity games} -- can be solved via a symbolic fixed-point algorithm in the $\mu$-calculus with almost the same computational worst case complexity as the algorithm for the \enquote{normal} version of the same game. %This makes strong transition fairness a promising candidate for a tractable class of fairness constraints in synthesis. 
% In particular, this provides a nested fixed-point equation for \Odd-fair parity games. 
The existence of quasi-polynomial time solution algorithms for \Odd-fair parity games then follows as a corollary of their nested fixed-point characterization~\cite{HS21, ANP21, JMT22}. %\todo{IS : Here, I have removed the intuition behind Banerjee et al.'s construction (syntactic and local change)}
%\end{comment}
% 
%Here, player \Odd has to obey the additional strong transition fairness constraint. We call such games \emph{\Odd-fair Parity games}. While Banerjee et. al.'s work allows to compute the winning region and a winning strategy of player \Even in such games, 
Unfortunately, it is well known that symbolic fixed-point computations become cumbersome very fast for parity games, as the number of priorities in the game graph increases, leading to high computation times in practice. 
Given the known inefficiency of existing quasi-polynomial algorithms for parity games \cite{oink, Parys19}, despite their theoretical advantages, they are not viable candidates for adoption in the development of efficient solution algorithms for \Odd-fair parity games either.
%The failure of existing quasi-polynomial algorithms for parity games to provide efficiency, despite their theoretical superiority remove them from the candidacy  %fail to meet the efficiency of the mentioned exponential algorithms in practice\todo{citation}, the most celebrated one being Zielonka's algorithm.
% IS: I removed the sentence below, because I din't see how it contributes to our point
%While Banerjee et. al. \cite{banerjee2022fast} provide a tool called \texttt{FairSyn} which exploits multi-threaded executions to allow for efficient synthesis over Rabin games with fairness constrains, this multi-threading has no effect on Parity games (which could be interpreted as special Rabin games), as the nesting of their $\mu$-calculus formula only allows for a single thread in \texttt{FairSyn}. 
% 
For (normal) parity games, computational tractability can be achieved by other algorithms, such as Zielonka's algorithm \cite{Zielonka98}, tangle learning \cite{van-Dijk-tangle-learning} or strategy-improvement \cite{Schewe-strategy-improvement}, implemented in the state-of-the-art tool \texttt{oink}~\cite{oink}, with Zielonka's algorithm being widely recognized as the most prominent approach.

% \todo{IS: I changed the following paragraph, can you check? I cemmented out the previous version below}
The \emph{\textbf{main contribution}} of this paper is a Zielonka-type algorithm, referred to as \enquote{\emph{\Odd-fair Zielonka's algorithm}}, for solving \Odd-fair parity games. This novel algorithm meets the efficiency of Zielonka's algorithm while maintaining the same computational worst-case complexity (which is exponential just like the worst-case complexity of the fixed-point algorithm from \cite{banerjee2022fast}).
%Using a prototype implementation, we experimentally verify that it retains the algorithmic advantage Zielonka's algorithm holds over fixed-point algorithms for classical parity games. %exhibits similar advantages to Zielonka's algorithm over fixed-point algorithms for classical parity games,
Using a prototype implementation, we experimentally verify its efficiency, demonstrating that it matches Zielonka’s algorithm in speed, thereby highlighting its comparable performance to fixed-point algorithms for classical parity games.

In contrast to the work by Banerjee et. al.~\cite{banerjee2022fast}, the adaptation and the correctness proof of \Odd-fair \emph{Zielonka's algorithm} requires the understanding of \Odd player strategies, while \cite{banerjee2022fast} studies the solution of such games solely from the \Even player's perspective. Unfortunately, \Odd strategies are substantially more complex than \Even strategies in such games, as they are not positional -- while player \Even strategies still are (see \cite[Thm.3.10]{banerjee2022fast}). %Positionality of \Even strategies is due to live edges only originating from \Odd vertices, and can be achieved as... . 
The \emph{\textbf{second contribution}} of this paper is therefore the formalization of \Odd player strategies in \Odd-fair parity games, via so called \emph{strategy templates}, which was unexplored prior to this work.
%We note that this is the first time \Odd player strategies are explored in these games since in~\cite{banerjee2022fast} the game is studied only from \Even player's perspective. 
We give a constructive proof for the existence of strategy templates winning for \Odd from all vertices in the winning region of \Odd.
This serves dual purposes: firstly, it enables us to prove the correctness of the \Odd-fair Zielonka's algorithm; secondly, it stands as a noteworthy contribution in its own right, augmenting our understanding of additional fairness assumptions in two-player games which are currently only unsatisfactorily adressed in various practically motivated synthesis problems.

\vspace*{-0.3cm}
\section{Preliminaries}
\noindent\textbf{Notation.}
We use $\mathbb{N}$ to denote the set of natural numbers including zero and $\mathbb{N}^+$ to denote positive integers. Let $\Sigma$ be a finite set. Then  $\Sigma^*$ and $\Sigma^\omega$ denote the sets of finite and infinite words over $\Sigma$, respectively. %, and $\Sigma^\infty$ is equal to $\Sigma^*\cup \Sigma^\omega$.
% Given two natural numbers $a,b\in\mathbb{N}$ with $a<b$, we use $[a;b]$ to denote the set $\set{n\in\mathbb{N} \mid a\leq n\leq b}$.
% For any given set $[a;b]$, we write $i\ineven [a;b]$ and $i\inodd [a;b]$ as short hand for $i\in [a;b]\cap \set{0,2,4,\ldots}$ and $i\in [a;b]\cap \set{1,3,5,\ldots}$ respectively.
% Given two sets $A$ and $B$, a relation $R\subseteq A\times B$, and an element $a\in A$, we write $R(a)$ to denote the set $\set{b\in B\mid (a,b)\in R}$.

\smallskip
\noindent\textbf{Game graphs.}
A \emph{game graph} is a tuple $\gamegraph= \tup{V,V^0,V^1,E}$ where $(V,E)$ is a finite directed graph with  \emph{edges} $ E $ and \emph{vertices} $ V $ partioned into player $0$ and player $1$ vertices, $V^0$ and $V^1$, respectively. 
Without loss of generality, we can assume
that all nodes in $V$ have at least one outgoing edge. Under this assumption, there exist plays from each vertex.
A \emph{play} originating at a vertex $v_0$ is an infinite sequence of vertices $\pi=v_0v_1\ldots \in V^\omega$. 
For $v \in V$, $E(v)$ denotes its successor set $\{w \mid (v, w) \in E\}$. 

\smallskip
\noindent\textbf{LTL winning conditions.}
Given a game graph $\gamegraph$, we consider winning conditions specified using a formula $\spec$ in \emph{linear temporal logic} (LTL) over the vertex set $V$, that is, we consider LTL formulas whose atomic propositions are sets of vertices. 
In this case the set of desired infinite plays is given by the semantics of $\spec$ which is an $\omega$-regular language $\lang(\spec)\subseteq V^\omega$. 
%In this case, an $\omega$-regular winning condition expressed as the LTL formula $\spec$ can be interpreted as the $\omega$-regular language $\lang(\spec)\subseteq V^\omega$ containing all desired infinite sequences of vertices. 
The standard definitions of $\omega$-regular languages and LTL are omitted for brevity and can be found in standard textbooks \cite{baierbook}. A game graph $\gamegraph$ under the winning condition $\spec$ is written as $\ltup{\gamegraph, \spec}$. A play $\pi$ is winning for player $0$ in $\ltup{\gamegraph, \spec}$ if $\pi\in\lang(\spec)$, i.e. $\pi\models\spec$.

\smallskip
\noindent \textbf{Strategies.}
A \emph{strategy} for player $j$ over the game graph $\gamegraph$ is a function $\rho^j : V^* \cdot V^j \to V$ with the constraint that for all $u\cdot v \in V^* \cdot V^j$ it holds that $\rho^j(u \cdot v) \in E(v)$. A play $\pi=v_0v_1\ldots \in V^\omega$ is compliant with $\rho^j$ if for all $i\in \mathbb{N}$ holds that $v_i\in V^j$ implies $v_{i+1}=\rho^j(v_0 \ldots v_i)$. A strategy $\rho^j$ is winning from a subset $V'$ of vertices of the game $\ltup{\gamegraph,\Psi}$ if all plays $\pi$ in $G$ that start at a vertex in $V'$ and are compliant with $\rho^j$ are winning w.r.t.\ $\Psi$. 
A strategy $\rho$ is called \emph{positional} iff for all $w_1, w_2 \in V^*$, $\rho(w_1 \cdot v) = \rho(w_2 \cdot v)$. %That is, a positional strategy $\rho_0$ for
 %\Even can be defined from $\Ve \to V$, instead of $V^* \cdot \Ve \to V$. 
%\AKS{misses definition of compliant plays}

\smallskip
\noindent\textbf{Parity Games.}
Parity games are particular two player games over a game graph $\gamegraph$ where the winning condition is given by a particular mapping of vertices. Formally, a parity game is a tuple $\mathcal{G} = \ltup{ V, \Ve, \Vo, E, \chi}$, where $\tup{V, \Ve, \Vo, E}$ is a game graph and $\chi : V \to \mathbb{N}^+$ is a function which labels each vertex with an integer value, called a \emph{priority}. The players $0$ and $1$ are called $\Even$ and $\Odd$ in a parity game and 
% 
% between player \Even and player \Odd, where
% \begin{inparaenum}[(i)]
%     \item $V$ is the vertex set, which is a disjunction of $\Ve$, the vertices owned by player \Even and $\Vo$, the vertices owned by player \Odd;
%     \item $E \subseteq V \times V $ is the edge set; and
%     \item $\chi : V \to \mathbb{N}^+$ is the priority function, which labels each vertex with an integer value.
% \end{inparaenum}
% % 
% The two players -- \Even and \Odd -- get to chose the next edge to take from their respective vertices $\Ve$ and $\Vo$. A play is
% an infinite sequence $v_1 v_2 v_3 \ldots \in V^\omega$ 
% \AKS{Usually we start counting at $0$, should we change it?}
% of vertices, where for all $i\in\mathbb{N}$, $(v_i, v_{i+1}) \in E$. Without loss of generality, we can assume
% that all nodes in $V$ have at least one outgoing edge. Under this assumption, there exists plays from each vertex.
a play $\pi = v_1 v_2 \ldots$ is winning for \Even iff $\max\{\inf(\pi)\}$ is \emph{even}, where $\inf(\pi)$ is the set of vertices visited infinitely often in $\pi$. Otherwise the play is winning for \Odd. %We define the sets $C_i := \{ v \in V \mid \chi(v) = i\}$ and $\overline{C_i} := V \setminus C_i$ to ease notation.

% \smallskip
% \noindent \textbf{Winning Parity Games.}
A node $v\in V$ is said to be won by \Even, if \Even has a (winning) strategy $\rho$ such that 
all plays $\pi = v\cdot \pi'$ that are compliant with $\rho$ are won by \Even. %\IS{Q: Is there a difference between saying (i) 'all plays that are compliant wih $\rho_0$ are won by... and (ii) for all player \Odd strategies $\rho_1$, all plays that are compliant with $\rho_0$ and $\rho_1$ are won by...}
The winning region of \Even is the set of all nodes won by \Even and is denoted by $\We$. The winning region of \Odd, $\Wo$, is defined similarly. 
It is well-known that parity games are determined, that is, all nodes are either in $\mathcal{W}_{Even}$ or in $\mathcal{W}_{Odd}$; and that both players have positional winning strategies from their respective winning regions \cite{EJ89}.

\smallskip
\noindent\textbf{\Odd-Fair Parity Games.} %\todo{I removed the 'fair parity' part upon the correct complaints of a review that they are not well-defined.}
An \emph{\Odd-fair parity game} 
$\mathcal{G}^\ell$ is a tuple $\ltup{\mathcal{G},E^\ell}$, where $\mathcal{G} = \langle V, \Ve, \Vo, E, \chi\rangle$ is a parity game, $E^\ell \subseteq E $ is a set of \emph{live edges} that originate from \Odd player vertices %inducing a \emph{strong transition fairness constraint}. 
and $V^\ell \subseteq \Vo$, the domain of the relation $E^\ell$, is the set of \emph{live vertices}.
The live edges induce a \emph{strong transition fairness constraint} -- whenever a live vertex $v$ is visited infinitely often, every outgoing live edge $(v, w') \in E^\ell$ needs to be taken infinitely often.
Formally, a play $\pi$ in $\mathcal{G}$ \emph{complies} with $E^\ell$ if the LTL formula%
\footnote{Here, $\square$, $\diamondsuit$ and $\bigcirc$ stand for the LTL operators 'always', 'eventually' and 'next'.}
\begin{equation}\label{eq:fairness-ltl}
    \textstyle\alpha := \bigwedge_{(v, w)\in E^\ell} (\,\square\, \diamondsuit \,v \implies \square\, \diamondsuit\, (v \wedge \bigcirc w))\vspace{-1mm}
\end{equation} 
holds along $\pi$, i.e.\ $\pi\models\alpha$. %\footnote{$\square$ stands for "always", $\diamondsuit$ for "eventually" and $\bigcirc$ for "next". For more detailed information on LTL \cite{}}
A play $\pi$ is winning for \Even in $\mathcal{G}^\ell$ if and only if $\pi \models \neg\alpha$ or $\max\{\inf(\pi)\}$ is \emph{even}. Dually, $\pi$ is winning for \Odd iff $\pi \models \alpha$ and $\max\{\inf(\pi)\}$ is odd. A strategy $\rho$ over $\gamegraph$ is therefore winning for \Even (resp.\ \Odd) in $\mathcal{G}^\ell$ if all plays compliant with $\rho$ are winning for \Even (resp. \Odd) in $\mathcal{G}^\ell$.

As the winning condition of a parity game can be equivalently modeled by a suitably defined LTL winning condition, we see that \Odd-fair parity games are a special $\omega$-regular game with perfect information. This implies that \Odd-fair parity games are determined (by the Borel determinacy theorem \cite{Martin75}) and whenever there exists a winning strategy for \Even/\Odd in such a game, then there also exists one with \emph{finite} memory \cite{GH82}.

\vspace{-0.15cm}

\section{Strategy Templates}\label{sec:templates}
% \vspace{-0.5em}

%Motivated by the known practical importance of exponential algorithms, such as Zielonka's algorithm, for solving Parity games fast in practice, the goal of this paper is to extend Zielonka's algorithm to \Odd-fair Parity games while retaining its efficiency. This, however, requires the formalization of player \Odd wining strategies in \Odd-fair parity games.
In this section, we introduce a formalization of player \Odd strategies in \Odd-fair parity games via \emph{strategy templates}.
% 
% under the name \Odd \emph{strategy templates}. This formalization will be needed in Sec. 4 while proving the correctness of Zielonka's algorithm for \Odd-fair parity games. 
% We state the existence of winning \Odd-strategy templates on \Wo in Prop~\ref{prop:existence-maximaloddstrategytemplates} but do not provide its proof until Sec. 5. 
% This is to avoid dividing the reader's attention from the extension of Zielonka's algorithm to Odd-fair parity games, which is the main emphasis of this paper and is introduced in Section 4.
% 
In contrast to player \Even, player \Odd winning strategies are no longer positional in \Odd-fair parity games, as illustrated by the following example. %that requires the same number of symbolic steps as the algorithm computing winning strategies for \Even in \enquote{normal} parity games.
% \vspace{-0.5em}
\begin{example}\label{ex:strategytemplates}
Consider the three different parity games depicted in Fig.~\ref{fig:Oddstrategies1}. %, three \Odd-fair parity games are depicted, with circles indicating \Ve and squares indicating \Vo. Edges in $E^\ell$ are shown by dashed lines. All nodes are labeled with their priorities.
   In all three games, \Odd has a winning strategy from all vertices, i.e., $\mathcal{W}_{Odd}=V$. %The red-colored edges indicate \Odd's strategy: if \Odd takes the red edges alternatingly from the source nodes, it wins from all nodes. 
  However, in order to win, the vertex $3$ has to be seen infinitely often in game (a) and (b), which forces \Odd to use its live edge\textbackslash s infinitely often. This prevents the existence of a positional strategy for \Odd in games (a) and (b): In (a) it needs to somehow alternate between (it's only) live edge to $4$ and a \enquote{normal} edge to $7$ (both indicated in red) in order to win, and in (b) it needs to somehow alternate between all its live edges (also indicated in red). In the game (c), \Odd can win by 'escaping' its live vertex $3$ to a \enquote{normal} vertex $5$, and thereby has a positional strategy. % (again indicated in red).
   
  Now consider the subgraph of each game formed by all colored edges (red and blue), which include the strategy choices from \Vo and \emph{all} outgoing edges from \Ve. As we have seen that \Odd needs to play all red edges repeatably, this subgraph represents the paths that \emph{can} be seen in the game depending on the \Even strategy. Hence, a node $v\in\Vl\subseteq\Vo$ can be seen infinitely often in a play (compliant with \Odd's strategy), if it lies on a cycle in this subgraph. We observe that, in games (a) and (b), node $3$ lies on cycles in this subgraph, whereas in game (c), it does not. 
  We further see that whenever a vertex  $v\in\Vl$ lies on a cycle, \Odd needs to take all its outgoing live edges (as for vertex $3$ in example (b)) and possibly one more edge (as for vertex $3$ in example (a)), for all other vertices in $\Vo$ a positional strategy suffices (as for vertex $5$ in all examples, and for vertex $3$ in example (c)). This shows that \Odd strategies are intuitively still \enquote{almost positional}.
% %   
%   
% Another intuition we gather is that, in all of these examples for any node in \Wo it is sufficient for \Odd to take either one outgoing edge (i.e. \Odd had a positional strategy on the node) or, to take all its outgoing live edges and possibly one more edge. We express this feature of a strategy by \enquote{almost positionality}.
   
%    The difference between the games (a)-(b) and (c) is that, in the first two node $3$ has to be seen infinitely often whereas in the later this is not the case. 
%    In game (a), \Odd does not have a positional winning strategy from node $3$; however, it has a winning strategy that allows it to take its live edge to $4$ together with its edge to $7$ infinitely often to win the game.
% In (b), once more \Odd does not have a positional winning strategy from node $3$, since in all winning plays it needs to take both its live outgoing edges infinitely often.
% On the other hand game (c) is an example where \Odd has a positional winning strategy that ignores the live outgoing edge of $3$.
\end{example}%\vspace{-2em}

\begin{figure}
\begin{center}
    \includegraphics[width=0.8\textwidth]{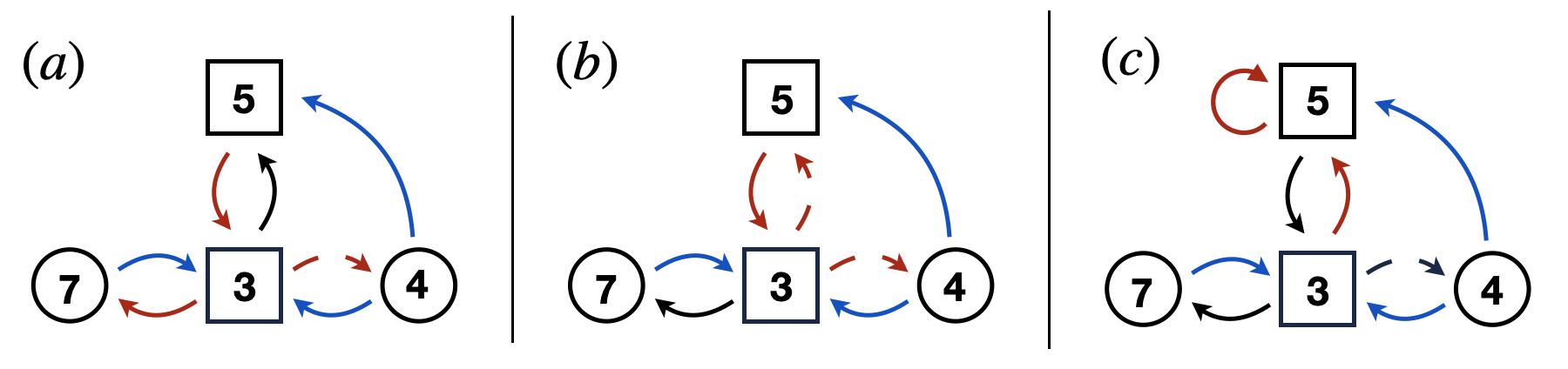}
    \end{center}
    \vspace{-0.8cm}
     \caption{\Odd-fair games with player even \Ve (circles) and player odd \Vo (squares) vertices (labeled with their priorities). Live edges $E^\ell$ (dashed) originate from $\Vo$. Colored player \Odd (red) and player \Even (blue) edges belong to player \Odd's strategy template.}\label{fig:Oddstrategies1}%\vspace{-2em}
     \vspace*{-0.5cm}
    \end{figure}

% The difference between the games (a)-(b) and (c) is that, in the first two node $3$ has to be seen infinitely often whereas in the later this is not the case. Observe that, for each game in Fig.~\ref{fig:Oddstrategies1} if we mark all outgoing edges of \Even nodes with red in addition to the red \Odd edges, we get a subgraph of the game.

\vspace*{-0.2cm}

The intuitions conveyed by Ex.~\ref{ex:strategytemplates} are formalized by the following definitions. % for \Odd strategy templates.

% are captured by so called \emph{strategy templates}, data structures that allow us to define an infinite number of winning player \Odd strategies for an \Odd-fair parity game in a finitary manner that is \emph{almost positional}.

\begin{definition}[\Odd Strategy Template]\label{def:Oddstrategytemplate}
 Given an \Odd-fair parity game $\mathcal{G}^\ell = \ltup{\mathcal{G}, E^\ell}$ with \newline $\mathcal{G} = \langle V, \Ve, \Vo, E, \chi\rangle$, an \Odd \emph{strategy template} $\mathcal{S}$ over $\mathcal{G}^\ell$ is a subgraph of $\mathcal{G}$ given as follows: $\mathcal{S}:=\tup{V',E'}$ where $V'\subseteq V$ and $E'\subseteq E \cap (V' \times V')$ such that the following hold,
\begin{compactitem}\label{item:Oddstrtemprules}
 \item if $v \in \Vo \cap V'$ does not lie on a cycle in $(V',E')$, then $|E'(v)|=1$,
 \item if $v \in \Vo \cap V'$ lies on a cycle in $(V',E')$ then $E^\ell(v) \subseteq E'(v)$ and  $1\leq |E'(v)|\leq |E^\ell(v)| + 1$,
 \item if $v \in \Ve \cap V'$, then  $E'(v) = E(v)$.
\end{compactitem}
\end{definition}
%Intuitively, whenever an \Odd vertex lays on a cycle in $(V', E')$, we expect it
%to have all its outgoing live edges, and possibly one more edge. Whenever it does not lay on a cycle, it has exactly one of its outgoing edges. All vertices $v\in\Ve$ have all of their outgoing edges. Moreover, all of the vertices in $V'$ have at least one outgoing edge in $E'$. \AKS{maybe move to the example above}\IS{I am not sure, we can remove if space doesn't suffice.}
% We call a tuple $(V', V'_\Even, V'_\Odd, E')$ a \textit{plausible} strategy template to indicate that the sets $V', V'_\Even, V'_\Odd $ and $E'$ satisfy the restrictions given in the definition above. \IS{??}
\begin{definition}\label{def:compliantstrat}
 Let  $\mathcal{G}^\ell = \ltup{\mathcal{G}, E^\ell}$ be an \Odd-fair parity game with \Odd strategy template $\mathcal{S}=\tup{V',E'}$, and $V'_\Odd := V' \cap V_\Odd$. Then an
\Odd strategy $\rho$ is said to be \textbf{compliant} with $\mathcal{S}$ if % the restriction of $\rho$ to $V'$ is a winning strategy in the game $\ltup{\mathcal{S},\alpha'}$ where 
it is a winning strategy in the game $\ltup{\gamegraph,\alpha'}$ where $\gamegraph= \tup{V,\Ve,\Vo,E}$ and 
\begin{subequations}
 \begin{align}
 \alpha':= &\textstyle\bigwedge_{v\in\Vo'}(\,\square\, (\,v \implies \bigvee_{(v,w)\in E'} \bigcirc\, w\,))\,\label{equ:alpha:a}\\
 & \textstyle\wedge \bigwedge_{v\in\Vo'} (\,\square \,\diamondsuit\, v \implies \bigwedge_{(v,w)\in E'}\square\, \diamondsuit\, (\,v \wedge \bigcirc \,w\,)).\label{equ:alpha:b}
\end{align}
\end{subequations}
\end{definition}

Intuitively, for all \Odd vertices in $\mathcal{S}$, the strategy $\rho$ compliant with $\mathcal{S}$ takes only their outgoing edges in $\mathcal{S}$ \eqref{equ:alpha:a}, and if a play visits an \Odd node $v$ infinitely often, then $\rho$ takes each of $v$'s outgoing edges in $\mathcal{S}$ infinitely often \eqref{equ:alpha:b}.
For an \Odd strategy template $\mathcal{S}$, if $v \in V'_\Odd$ lies on a cycle in $\mathcal{S}$, then by Def. \ref{def:Oddstrategytemplate}, $\mathcal{S}$ contains all live outgoing edges of $v$. By \eqref{equ:alpha:b} any \Odd strategy $\rho$ compliant with $\mathcal{S}$ satisfies the fairness condition in \eqref{eq:fairness-ltl} for $v$. 
On the other hand, if $v \in V'_\Odd$ does not lie on a cycle in $\mathcal{S}$, then by \eqref{equ:alpha:a} any such $\rho$ sees $v$ at most once. Thus $\rho$ trivially satisfies \eqref{eq:fairness-ltl} for $v$. 
% That is, any \Odd strategy $\rho$ compliant with an \Odd strategy template satisfy the fairness condition \eqref{eq:fairness-ltl}. 
This observation is stated in the following proposition.

%With this, it becomes clear that any \Odd strategy $\rho$ compliant with any strategy template $\mathcal{S}$, obeys the fairness condition.
% \vspace{-1mm}
\begin{proposition}%\footnote{The proof of the proposition can be found in the appendix, Sec.~\ref{app:S-proof}. }\label{prop:S:fair}
 Given the premisses of Def.~\ref{def:compliantstrat} let $\pi$ be a play starting from a node in $V'$ that complies with $\rho$. Then $\pi \models \alpha$ where $\alpha$ is the LTL formula in~\eqref{eq:fairness-ltl}.%\vspace{-2mm}
\end{proposition}

Next, we define \Even strategy templates. Each \Even strategy template encodes a unique \Even positional strategy, which is known to exist in \Odd-fair parity games~\cite{Klarlund94}, due to the lack of fair edges defined on \Even vertices. %, \Even strategy templates are very simple\footnote{In fact, \Even strategy templates simply encode a positional strategy and are only re-defined to make further arguments more symmetric for both players.}.
\begin{definition}\label{def:Evenstrategytemplate}
    Given an \Odd-fair parity game $\mathcal{G}^\ell = \ltup{\mathcal{G}, E^\ell}$ with \newline $\mathcal{G} = \langle V, \Ve, \Vo, E, \chi\rangle$, an \Even \emph{strategy template} $\mathcal{S}$ over $\mathcal{G}^\ell$ is a subgraph of $\mathcal{G}$ given as $\mathcal{S}:=\tup{V', E'}$ where $V'\subseteq V$ and $E'\subseteq E \cap (V' \times V')$ such that,    \begin{compactitem}\label{item:Evenstrtemprules}
     \item if $v \in \Ve \cap V'$, then $|E'(v)|=1$,
     \item if $v \in \Vo \cap V'$, then  $E'(v) = E(v)$.
    \end{compactitem}
\end{definition}

\vspace*{-0.1cm}

An \Even strategy $\rho$ is compliant with the \Even strategy template $\mathcal{S} = \tup{V', E'}$ if for all $v \in V'_\Even$, $\rho(v) = E'(v)$. In other words, $\rho$ is the positional strategy defined by $\mathcal{S}$.

Let $\rho$ be an \Odd (\Even) strategy, compliant with the \Odd (\Even) strategy template $\mathcal{S}$ and let $\pi$ be a play compliant with $\rho$. Then we call $\pi$ a play \emph{compliant with $\mathcal{S}$}.

\vspace*{-0.1cm}

\begin{definition}
An \Odd (\Even) strategy template $\mathcal{S}=\ltup{V', E'}$ is \emph{winning} in the \Odd-fair parity game $\mathcal{G}^\ell$ if all \Odd (\Even) strategies $\rho$ compliant with $\mathcal{S}$ are winning for player \Odd (\Even) in $\mathcal{G}^\ell$ from $V'$. A winning \Odd (\Even) strategy template $\mathcal{S}$ is called \emph{maximal} if $V'=\Wo$ ($\We$).%\vspace{-2mm}
\end{definition}
%\IS{in the previous defn, we used \Even strategy templates without defining them first. Either define them beforehand or remove them from this definition.}

\vspace*{-0.2cm}
We note that maximal winning \Odd (\Even) strategy templates $\mathcal{S}$ immediately imply that for every vertex $v\in \Wo$ ($\We$) there exists a winning strategy for player \Odd (\Even) from $v$ that is compliant with $\mathcal{S}$.
The existence of maximal winning \Even strategy templates follows from the existence of positional \Even strategies~\cite{Klarlund94}. 
The first main contribution of this paper is a constructive proof showing the existence of maximal winning \Odd strategy templates given in the next section. 
This result is then used in Sec.~\ref{sec:zielonka} to prove the correctness of \Odd-fair Zielonka's algorithm, which is introduced there.

\vspace*{-0.25cm}
\section{Existence of Maximal Winning \Odd Strategy Templates}\label{sec:strat-templates}

This section proves the existence of maximal winning \Odd strategy templates\footnote{In the rest of this section, we will sometimes call \Odd strategy templates simply, \emph{strategy templates}, since these are the only strategy templates we will be dealing with.} in \Odd-fair parity games, formalized in the following theorem.

\begin{theorem}\label{thm:existence-maximaloddstrategytemplates}
    Given an \Odd-fair parity game $\mathcal{G}^\ell$, there exists a maximal winning \Odd strategy template. 
\end{theorem}

% Since the existence of maximal winning \Even strategy templates follow from previous work (as mentioned in Sec. 3), the proof of existence of maximal winning \Odd strategy templates
% allow us to use strategy templates as a formalization of strategies of both player in \Odd-fair parity games. This is the first formalization of both players' strategies in \Odd-fair games known to the authors.
% Furthermore, we use this formalization in Sec. 4 in the correctness proof of Zielonka's algorithm for \Odd-fair parity games.
We prove Thm.~\ref{thm:existence-maximaloddstrategytemplates} by giving an algorithm which constructs $\mathcal{S}$ from a ranking function induced by a fixed-point algorithm in the $\mu$-calculus which computes \Wo. Towards this goal, Sec.~\ref{sec:assump:prelim} first introduces necessary preliminaries, Sec.~\ref{sec:templates:solving} gives the fixed-point algorithm to compute \Wo and Sec.~\ref{sec:templates:ranking} formalizes how to extract a strategy template $\mathcal{S}$ from the ranking induced by this fixed-point and proves that $\mathcal{S}$ is indeed maximal and winning. % from this computation and Sec.~\ref{sec:templates:result} finally shows how this ranking can be used to construct maximal winning strategy templates and proves their correctness. 

While this section uses fixed-point algorithms extensively to \emph{construct} a maximal winning \Odd strategy template towards a \emph{proof} of Thm.~\ref{thm:existence-maximaloddstrategytemplates}, we note again that the proof of the new Zielonka's algorithm given in Sec.~\ref{sec:zielonka} only uses the \emph{existence} of templates (i.e., the fact that Thm.~\ref{thm:existence-maximaloddstrategytemplates} holds) and does not utilize their \emph{construction} via the algorithm presented here. %.\todo{But in the proofs we sometimes construct a strategy template, no? Maybe we can say "does not construct one" instead of "does not require the construction"}

\subsection{Preliminaries on Fixed-Point Algorithms}\label{sec:assump:prelim}

This subsection contains the basic notation used in this section. % to define and evaluate the fixed-point algorithms for computing $\Wo$.

\smallskip
\noindent\textbf{Set Transformers.}  Let $ \gamegraph=(V,\Ve, \Vo, E) $ be a game graph, $ S,T\subseteq V $ and $\bb$ be the player index.\footnote{$\bb \in \{\Even,\Odd\} $ where $\bb=\Even$ implies $\neg \bb=\Odd$, and vice versa.} Then we define the following predecessor operators: 
\begin{subequations}\label{equ:Pres}
 \begin{align*}    
    \Pre_\bb^\exists(S) &:= \{ v \in V_\bb \mid E(v) \cap S \neq \emptyset \} && 
        \Lpre^\exists(S) := \{v \in \Vo \mid E^\ell(v) \cap S \neq \emptyset\} \notag \\ 
        \Pre_\bb^\forall(S) &:= \{ v \in V_\bb \mid E(v) \subseteq S  \} &&
    \Lpre^\forall(S) := \{v \in \Vo \mid E^\ell(v) \subseteq S \}\quad  (3)
     \end{align*}
\end{subequations}

The predecessor operators $\Pre_\bb^\exists(S) $ and $\Pre_\bb^\forall(S)$ compute the sets of vertices with \emph{at least one} successor and with \emph{all} successors in $ S $, respectively. The live predecessor operators  $ \Lpre^\exists(S) $ and $\Lpre^\forall(S)$ restrict this analysis to live edges.
We see that 
% \vspace{-2mm}
% \begin{subequations}
 \begin{align}    \label{equ:Preseq}
   \neg \Pre_\bb^{\exists}(\neg S)&= V_{\neg \Lambda} \cup \Pre_{\neg \bb}^{\forall}(S)&&\text{and}&&
   \neg \Lpre^{\exists}(\neg S)= \Ve \cup \Lpre^{\forall}(S)%\vspace{-2mm}
 \end{align}
% \end{subequations}
% \vspace{-1.5mm}
where for a set $X \subseteq V$, $\neg X$ stands for $V \setminus X$. We combine the pre-operators from \eqref{equ:Pres} into the combined set:\footnote{Note that $\Apre(S,T)$ and $\Npre(S,T)$ are meaningful only when $T \subseteq S$ and $S \subseteq T$, respectively. Otherwise they are equivalent to $\Cpre_\Even(T)$ and $\Cpre_\Odd(T)$. We note that these preconditions will always be satisfied in our calculations due to the monotonicity of fixed-point computations.} 
\begin{subequations}\label{equ:combindedPres}     
     \begin{align}
    \Cpre_\bb(S) &:= \Pre_\bb^\exists(S) \cup \Pre_{\nb}^\forall(S)\label{equ:cpre}\\
    \Apre(S, T) &:= \Cpre_\Even(T) \cup (\Lpre^{\exists}(T) \cap \Pre_\Odd^{\forall}(S))\label{equ:apre}\\
    \Npre(S,T) &:= \Cpre_\Odd(T) \cap (\Ve \cup \Lpre^\forall(T) \cup \Pre_\Odd^{\exists}(S))\label{equ:npre}       
    \end{align}
\end{subequations}
% \vspace{-1mm}
The \emph{controllable predecessor operator} $\Cpre_\bb(S)$ computes the set of vertices from which player $\bb$ can force visiting $ S $ in \emph{one} step. It immediately follows that 
%\begin{subequations}\label{equ:combindedPreseqal}   
%\vspace{-0.5cm}
\begin{align}
\neg \Cpre_\Even(\neg S)&:= \Cpre_\Odd(S)\label{equ:cpre_equal}.
\end{align}
% 
% $\neg \\Cpre_\Even(\neg S):= \Cpre_\Odd(S)~\refstepcounter{equation}(\theequation)\label{equ:cpre_equal}$.
% 
%$~\refstepcounter{\neg \Cpre_\Even(\neg S)&:= \Cpre_\Odd(S)}(\theequation)\label{equ:cpre_equal}$
The \emph{almost-sure controllable predecessor} operator $\Apre(S,T)$ computes the set of states that can be controlled by Player \Even to stay in $T$ (via $\Cpre_\Even(T ))$ as well as all Player \Odd states in $V^\ell$ that
(a) will eventually make progress towards $T$ if Player \Odd obeys its fairness-assumptions (via $\Lpre^{\exists}$) and (b) will never leave $S$ in the \enquote{meantime} (via $\Pre_\Odd^{\forall}(S))$). Using \eqref{equ:Preseq} and \eqref{equ:cpre_equal} we have 
  $\Npre(S,T):= \neg \Apre(\neg S, \neg T)$.

\smallskip
\noindent\textbf{Fixed-point Algorithms in the $ \mu $-calculus.} 
The $ \mu $-calculus offers a succinct representation of symbolic algorithms (i.e., algorithms manipulating sets of vertices instead of individual vertices) over a game graph $ \gamegraph $. 
% The formulas of the $ \mu $-calculus, interpreted over a 2-player game graph $ \gamegraph $, are given by the grammar %\vspace{-2mm}
% % \[ 
% $\phi\coloneqq p \mid X \mid \phi\cup\phi \mid \phi\cap\phi \mid \mathit{pre}(\phi) \mid \mu X.\phi \mid \nu X.\phi %\vspace{-1mm}
% $
% % \]
% where $ p $ ranges over subsets of $ V $, $ X $ ranges over a set of formal variables, $ pre $ ranges over the monotone set transformers in  \eqref{equ:Pres} and \eqref{equ:combindedPres}, %$ \{\textsf{pre}, \textsf{cpre}^a, \textsf{attr}^a  \} $, 
% and $ \mu $ and $ \nu $ denote, respectively, the least and the greatest fixed-point of the functional defined as $ X\mapsto \phi(X) $. 
% Since the operations $ \cup, \cap $, and the set transformers are all monotonic, the fixed-points are guaranteed to exist, due to Knaster-Tarski Theorem \cite{KnasterTarski:TraskiKnasterTheorem}.
We omit the (standard) syntax and semantics of $ \mu $-calculus formulas (see \cite{Kozen:muCalculus}) and only discuss their evaluation
%  
% \smallskip
% \noindent\textbf{Evaluating Fixed-point Algorithms.} 
% A $ \mu $-calculus formula evaluates to a set of vertices over $ \gamegraph $, and the set can be computed by induction over the structure of the formula, where the fixed-points are evaluated by iteration. 
% The reader may note that $ \textsf{pre} $ and $ \textsf{cpre} $ can be computed in time polynomial in number of vertices, and since the game graph is finite, $ \textsf{attr} $ is also computable in polynomial time. %: by applying the respective predecessor operators repeatedly until the set stabilizes. %Hence, $ \textsf{tpre} $ is computable in polynomial time as well.
% 
on an example fixed-point algorithm given by a 2-nested $ \mu $-calculus formula of the form $Z=\mu Y.~\nu X.~\phi(X,Y)$, where  $ X,Y \subseteq V$ are subsets of vertices
 and $ \mu $ and $ \nu $ denote, respectively, the least and the greatest fixed-point. $\phi$ is a formula composed from the \emph{monotone set transformers} in  \eqref{equ:Pres} and \eqref{equ:combindedPres}. % of the functional defined as $ X\mapsto \phi(X) $. 
 
 Given this formula, first, both formal variables $X$ and $Y$ are initialized. As $Y$ (resp. $X$) is preceded by $\mu$ (resp. $\nu$) it is initialized with $Y^0:=\emptyset$ (resp. $X^0:=V$). Now we first keep $Y$ at its initial value and iteratively compute $X^k=\phi(X^{k-1},Y^0) $ until $X^{k+1}=X^k$. At this point $X$ saturates, denoted by $X^\infty$. We then \enquote{copy} $X^\infty$, to $Y$, i.e., have $Y^1:=X^\infty$, reinitialize $X^0:=\emptyset$, and re-evaluate $X^k=\phi(X^{k-1},Y^1) $ with the new value of $Y$. This calculation terminates if $Y$ saturates, i.e.,  $Y^\infty=Y^{l+1}=X^l$ for some $l\geq 0$, and outputs $Z=Y^\infty$. In order to remember all intermediate values of $X$ we use $X^{l,k}$ to denote the set computed in the $k$-th iteration over $X$ during the computation of $Y^l$. I.e., $Y^l=X^{l,\infty}$.

\smallskip
\noindent\textbf{Additional Notation.} 
We will use the letters $l,m$ and $n$ exclusively to denote \emph{even} positive integers. For $a \leq b \in \mathbb{N}$, we will use the regular set symbol $[a,b]$ to denote the set of all integers between $a$ and $b$, i.e., $[a,b]:=\{a, a+1 , \ldots , b\}$; and $\ev{a}{b}$ to denote all the \emph{even} integers between $a$ and $b$. %, including $a$ or $b$ as well given that it is even,
E.g. $\ev{2}{7} = \{2, 4,  6\}$.
In addition, given an \Odd-fair parity game $\mathcal{G}^\ell$, we define the sets $C_i := \{ v \in V \mid \chi(v) = i\}$ and $\overline{C_i} := V \setminus C_i$ to ease notation. We say $\mathcal{G}^\ell$ has 
the least even upper bound $l$ if $C_l \cup C_{l-1}\neq\emptyset$ and $C_i=\emptyset$ for all $i>l$.

\vspace{-0.1cm}
\subsection{A Fixed-Point Algorithm for $\mathcal{W}_{\Odd}$}\label{sec:templates:solving}

%\begin{align*}
%    \mathcal{W}_{Even}= \nu {Y_l} \mu {X_{l-1}} \ldots \nu {Y_2} \mu {X_1}. A_i\quad &(C_l \cap \Cpre_\Even(Y_l)) \cup \\
%                                                                        &(\bigcup_{i \in [1, l-1]} C_i \cap \Apre(Y_l, X_{l-1})) \cup \nonumber \\
%                                                                        & \ldots\nonumber \\
%                                                                        &(C_2 \cap \Cpre_\Even(Y_2)) \cup \nonumber \\
%                                                                        & (C_1 \cap \Apre(Y_2, X_1))\nonumber
%\end{align*}

Given an  \Odd-fair parity game $\mathcal{G}^\ell = \ltup{\langle V, \Ve, \Vo, E, \chi \rangle, E^\ell}$ this section presents a fixed-point algorithm in the $\mu$-calculus which computes the winning region $\Wo$ of player $\Odd$ in \Odd-fair parity games. It is obtained by negating the fixed-point formula computing \We \,in~\cite{banerjee2022fast}, formalized in the following proposition and proven in App.~\ref{app:fp-proof}.

%\vspace*{-0.15cm}
%It was recently shown in \cite{banerjee2022fast} that the winning region $\We$ for \Even in an \Odd-fair parity game $\mathcal{G}^\ell$ with least even upperbound color $l\geq 0$ can be computed by a fixed-point formula that preserves the complexity of parity fixed-point formula.
 %As \Odd-fair parity games are determined, we can simply compute the winning region for player $\Odd$ by negating this fixed-point formula, and obtain the formula in Prop.~\ref{prop: W_Odd}.
\begin{proposition}\label{prop: W_Odd}
Given an \Odd-fair parity game $\mathcal{G}^\ell = (\ltup{V, \Ve, \Vo, E, \chi}, E^\ell)$ with least even upper bound $l\geq 0$ it holds that $Z=\Wo$, where
\begin{small}
\begin{align}\label{eq:fp-odd}
    Z &:=\textstyle \mu {Y_l}.~  \nu {X_{l-1}}.~  \ldots \mu{Y_2}.~  \nu{X_1}.~  \bigcap_{j \in \ev{2}{l}} \B_j[Y_j, X_{j-1}], \\ \vspace{0.1cm}
    &\text{ where} \quad
    \B_j[\mathbf{Y}, \mathbf{X}] := \left(\textstyle\bigcup_{i \in [j+1,l]} C_i\right) \cup \left(\overline{C_j} \cap \Npre(\mathbf{Y}, \mathbf{X}) \right) \cup \left(C_j \cap \Cpre_\Odd(\mathbf{Y})\right).\nonumber
\end{align}
\end{small}
% then $Z=\Wo$.
% Further, it takes $\mathcal{O}(n^{l+1})$ symbolic steps to compute $Z$.
\end{proposition}

%\vspace*{-0.15cm}

Before utilizing \eqref{eq:fp-odd} we illustrate its computations via an example. %IRMAAAAK IRMAK

% In order to ease the understanding of the subsequent use of \eqref{eq:fp-odd} for the construction of strategy templates, let us consider the following example. 
%\vspace{-1cm}
%\vspace*{-0.15cm}

\begin{example}\label{ex:1}
Consider the \Odd-fair parity game $\mathcal{G}^\ell $ depicted in Fig.~\ref{fig:ex1} (left). Here, the name of the vertices coincide with their priorities, e.g., $C_2=\set{2a, 2b, 2c}$. $\Ve$ and $\Vo$ are indicated by circles and squares, respectively. Edges in $E^\ell$ are shown by dashed lines. 
% Let $\langle \mathcal{G}, E^\ell \rangle $ for $ \mathcal{G} = \langle V, \Ve, \Vo, E, \chi \rangle$ with $V = \{1a, 2a, 2b, 2c, 3a, 3b, 4a\}$, $\Ve=\{1a,2a,3a,3b,2c\}$, $\Vo=\{2b,4a\}$, $E^\ell = \{(4a, 2a), (4a,3a), (2b,2c)\}$, $E = E^\ell \cup \{(2a,2a),(2a,4a),(3a,4a),(1a,4a),(1a,2c),(2c,2b),(2b,3b),(3b,2b) \}$ and $C_1 = \{1a\}$, $C_2=\{2a,2b,2c\}$, $C_3 = \{3a,3b\}$, $C_4=\{4a\}$. The game graph is depicted in figure \ref{gamegraph} where the live edges are shown by dashed lines. 
% 
As the least even upper bound in this example is $l=4$, 

\vspace*{-0.3cm}
\begin{small}
\begin{align}\label{equ:fpexample}
    &Z = \mu Y_4.~ \nu X_3.~ \mu Y_2.~ \nu X_1.~ \Phi^{Y_4, X_3, Y_2, X_1}~\quad  \text{where}\\
    &\Phi^{Y_4, X_3, Y_2, X_1}:= (\overline{C_4} \cap \Npre(Y_4, X_3)) \cup (C_4 \cap \Cpre_\Odd(Y_4)))\nonumber\\
    & \hspace{2.03cm}\cap (\overline{C_2} \cap \Npre(Y_2, X_1)) \cup (C_2 \cap \Cpre_\Odd(Y_2)) \cup C_4 \cup C_3)\nonumber.
\end{align}
\end{small}

\vspace{-0.2cm}

\begin{figure}%\label{fig:ex1}
\begin{center}
  \includegraphics[width=0.7\textwidth]{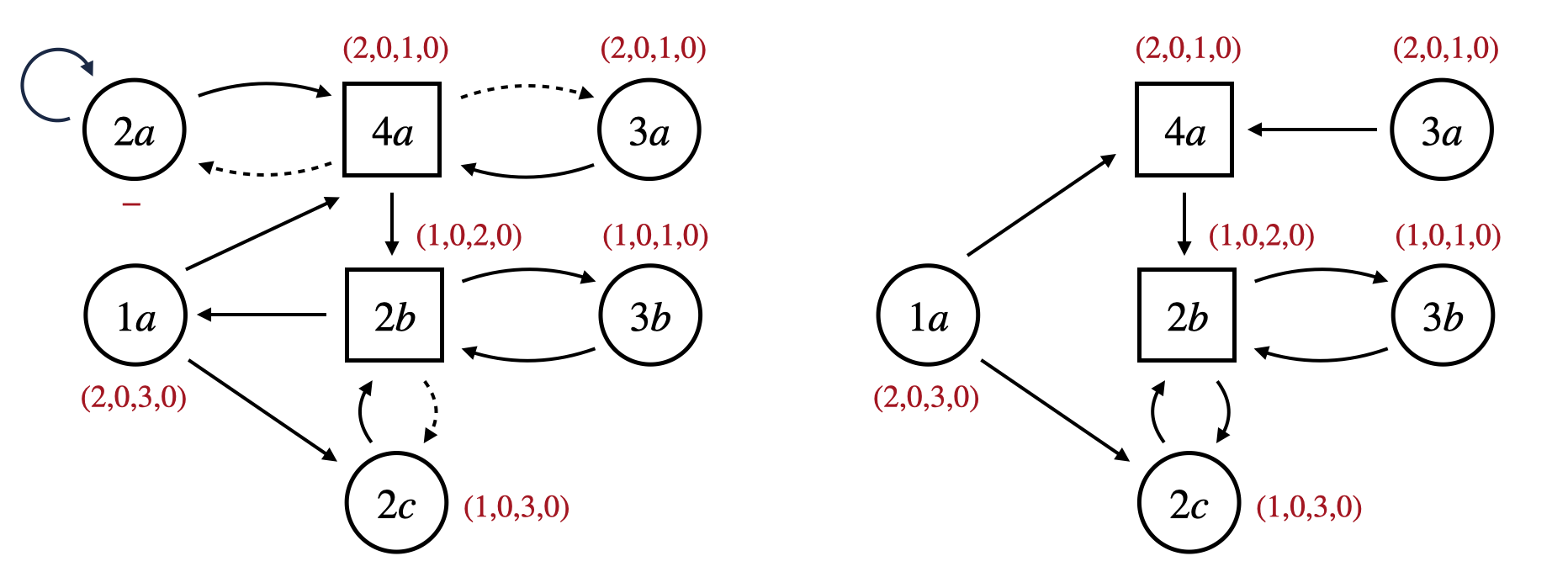}
\end{center}
\vspace{-0.7cm}
       \caption{ \Odd-fair parity game $\mathcal{G}^\ell$ discussed in Ex.~\ref{ex:1}, \ref{ex:2}, and \ref{ex:3} (left) and its corresponding minimum rank based maximal \Odd strategy template $\Sc^{\mathcal{G}^\ell}$ as defined in Def.~\ref{def:S} (right).}\label{fig:ex1} 
       \vspace{-1em}
\end{figure}

\vspace{-0.2cm}

Using the notation defined in Sec.~\ref{sec:assump:prelim}, we initialize  \eqref{equ:fpexample} by $Y_4^{0} = \emptyset$, $X_3^{0, 0} = V$, $Y_2^{0,0,0} = \emptyset$ and $X_1^{0,0,0,0} = V$ and observe from \eqref{equ:combindedPres} that  $\Cpre_\Odd(\emptyset)=\emptyset$ and $\Npre(\emptyset, V)=V$. We obtain 
\begin{small}
\begin{align*}
 X_1^{0,0,0,1} &= \Phi^{Y_4^{0}, X_3^{0, 0}, Y_2^{0,0,0}, X_1^{0,0,0,0} }
 =((\overline{C_4} \cap \Npre(\emptyset, V)) \cup (C_4 \cap \Cpre_\Odd(\emptyset)))\cap ((\overline{C_2} \cap \Npre(\emptyset, V)) \\ 
 & \quad \,\, \cup (C_2 \cap \Cpre_\Odd(\emptyset)) \cup C_4 \cup C_3) =(\overline{C_4} ) \cap (\overline{C_2} \cup C_4 \cup C_3) =C_3 \cup C_1\\
%   \end{align*}
%   \begin{align*}
  X_1^{0,0,0,2} &= \Phi^{Y_4^{0}, X_3^{0, 0}, Y_2^{0,0,0}, X_1^{0,0,0,1} }\\
  &= C_3 \cup (C_1 \cap \Npre(Y_2^{0,0,0}, X_1^{0,0,0,1})) = C_3 \cup (C_1 \cap \Npre(\emptyset, C_3\cup C_1))=C_3
%   \quad
%  &=X_1^{0,0,0,3}
\end{align*}
\end{small}

\vspace{-0.1cm}

where $ \Npre(\emptyset, C_3\cup C_1)=\emptyset$ as $v \in \Npre(\emptyset, C_3 \cup C_1)$ implies $v\in \Cpre_\Odd(C_3 \cup C_1) = \{2b,4a\}$ and $v\in \Ve \cup \Lpre^\forall(C_3 \cup C_1)$. However, $2b, 4a$ are \Odd vertices with live outgoing edges to $2a,2c\in (V \setminus (C_3 \cup C_1))$.
In the next iteration, we again get $X_1^{0,0,0,3} = C_3$ and thus $X_1$ saturates with $C_3$. Therefore, $Y_2^{0,0,1}=C_3$. Now the next round of computations of $\Phi$ results in 
\begin{small}
\begin{align*}
   X_1^{0,0,1,1} &= \Phi^{Y_4^{0}, X_3^{0, 0}, Y_2^{0,0,1}, X_1^{0,0,1,0} } =  C_3 \cup (C_1 \cap \Npre(Y_2^{0,0,1}, X_1^{0,0,1,0})) \cup (C_2 \cap \Cpre_\Odd(Y_2^{0,0,1}))\\
 & =C_3 \cup (C_1 \cap \Npre(C_3, V)) \cup (C_2 \cap \Cpre_\Odd(C_3))=C_3 \cup C_1 \cup \{2b\}\\
    X_1^{0,0,1,2} &= \Phi^{Y_4^{0}, X_3^{0, 0}, Y_2^{0,0,1}, X_1^{0,0,1,1} } = C_3 \cup \{2b\}=X_1^{0,0,1,3} 
% = \Phi^{Y_4^{0}, X_3^{0, 0}, Y_2^{0,0,1}, X_1^{0,0,0,2} } = C_3 \cup \{2b\}\\
\end{align*}
\end{small}
Here $C_1$ and $\{2b\}$ get added in $X_1^{0,0,1,1}$ as $1a \in \Npre(C_3, V)$ trivially and $2b \in \Cpre_\Odd (C_3)$ due to the edge $(2b,3b)$. $C_1$ is removed from $X_1^{0,0,1,2}$ since
$1a$ cannot be forced by \Odd to $C_1 \cup C_3 \cup \{2b\}$ in the next step.
%$1a \not \in ( C_1 \cap \Npre(Y_2^{0,0,1}, X_1^{0,0,1,1} )) = (C_1 \cap \Npre(C_3, C_1 \cup C_3 \cup \{2b\}))$ since $1a \not \in \Cpre_\Odd(C_1\cup C_3 \cup \{2b\})$.
The fixed-point calculation proceeds in a similar fashion, until $Y_4$ reaches its saturation value $V \setminus \{2a\}$. 
The full computation of $Z$ is given in App.~\cite{app:example}. %\vspace{-2mm}
\end{example}

%\begin{align}\label{eq:fp-odd2}
%    \Wo = \mu {Y_l} \nu {X_{l-1}} \ldots \mu{Y_2} \nu{X_1}. \quad & (\overline{C_l} \cap \Npre(Y_l, X_{l-1})) \cup (C_l \cap \Cpre_\Odd(Y_l))\cap \\
%%                                                                                & \ldots \nonumber \\
 %%                                                                               & ((\bigcup_{i \in [j+1, l]} C_i) \cup (\overline{C_j} \cap \Npre(Y_j, X_{j-1})) \cup (C_j \cap \Cpre_\Odd(Y_j))) \cap \nonumber \\
 %                                                                               & \ldots \nonumber \\
 %                                                                               & ((\bigcup_{i \in [3, l]} C_i) \cup (\overline{C_2} \cap \Npre(Y_2, X_{1})) \cup (C_2 \cap \Cpre_\Odd(Y_2))) \nonumber
 %   \end{align}

\subsection{Construction of a Rank-based Strategy Template}\label{sec:templates:ranking}
Given an \Odd-fair parity game $\mathcal{G}^\ell$ with the least even priority upper bound $l\geq 0$, we define a ranking function $\rank{}: \Wo \to \mathbb{N}^{l}$ first introduced in~\cite{SE84} and highly related to \enquote{progress measures}~\cite{KlarlundKozen91,Klarlund94,Klarlund90,Jurdzinski00}. Intuitively, $\rank{v}$ indicates in which iteration $v$ was added to $Z$ in \eqref{eq:fp-odd} and  never got removed from $Z$ again, as illustrated by the following example. %We show that there exists a strategy template $\Sc=(V',\Ve',\Vo', E')$ of $\mathcal{G}^\ell$, constructed according to $r$, for which all compliant player \Odd strategies are winning in $\mathcal{G}^\ell$.

%Let us show the saturation values of a variable $Z$ with $Z^\infty$. That is, $Z^\infty = Z^n = Z^{n+1}$. %Also for all even $j \in \ev{2}{l}$, let us fix the notation

\begin{example}\label{ex:2}
 Consider again the \Odd-fair parity game depicted in Fig.~\ref{fig:ex1}. Here, $\rank{v}$ of each $v \in \Wo = V \setminus \{2a\}$ is shown in red next to the node in the figure. Intuitively, the $4-$tuple is associated with the subscript $Y_4,Y_3,Y_2,Y_1$ of $\Phi$ in \eqref{equ:fpexample}. For instance $\rank{3a}=(2,0,1,0)$ indicates that $3a$ was added to $Z$ 
 during the first iteration of $Y_2$ inside the second iteration of $Y_4$.
 More concretely, $3a \not \in Y_4^{0}, 3a \not \in Y_4^1, 3a \in Y_4^2$. So $2$ is the first iteration of the $Y_4$ variable in which $3a$ got included in the variable. For $Y_2$, $3a \not \in Y_2^{2,0, 0}$ and $3a \in Y_2^{2,0,1}$, and therefore $\rank{3a} =  (2,0,1,0)$.
\end{example}

The intuition of Ex.~\ref{ex:2} is formalized in the following definition.

% Formally, ranks are defined as follows:
\begin{definition}[rank]\label{def:rank}
Given an \Odd-fair parity game $\mathcal{G}^\ell = (\ltup{V, \Ve, \Vo, E, \chi}, E^\ell)$ with least even upper bound $l\geq 0$ and winning region $\Wo\subseteq V$, we define the ranking function $\rank{}: \Wo \to \mathbb{N}^{l}$ for $v\in \Wo$ such that 
 \begin{equation}\label{eq:rank}
  \textstyle\rank{v}=(r_l,0,r_{l-1},0\hdots r_2, 0) \quad\text{if}\quad v\in \bigcap_{j\in\ev{2}{l}}Y_j^{r_l,0,\hdots ,r_j}\setminus Y_j^{r_l,0,\hdots ,r_j-1}.
 \end{equation}
where the valuations of the variables $Y_j$ are obtained from the iterations of the fixed-point calculation in~\eqref{eq:fp-odd} as illustrated in Ex.~\ref{ex:1}.
\end{definition}

% The formal definition of the rank of $v \in \Wo$ is the following,
% \begin{definition}
% \[ \rank{v} = (j_l+1, 0, j_{l-2}+1, 0, \ldots, j_4+1, 0, j_2+1, 0) \]
% where $(j_l, 0, j_{l-2}, 0, \ldots, j_2, 0)= min\{ (a_l, 0, a_{l-2}, 0, \ldots, a_2, 0) \mid \forall (t_l, 0, t_{l-2}, 0, \ldots, t_2, 0 )  \in \mathbb{N}^l \geq (a_l, 0, a_{l-2}, 0, \ldots, a_2, 0), \, v \in X_1^{t_l, 0, \ldots, t_2, 0} \}$.
% \end{definition}
% Note that the even indices of a rank are always $0$. Even though this introduces a redundancy in the representation, keeping the ranks in this forms helps to provide intuition in the proofs. 
% Also note that, the reason why we seem to consider only the $Y_j$'s iteration counts comes from the fact that $X_j$ variables are least fixpoint variables; i.e. $ X_j^ 0 \supseteq X_j^1 \supseteq \ldots$ is a decreasing sequence, and $X_j^0$ is initialized as $V$, 
% so the smallest iteration count that contains $v$ is always $0$ for all $X_j$. On the other hand, $Y_j$s are least fixpoint variables, are initialized as $\emptyset$ and $Y_j^0 \subseteq Y_j^1 \subseteq \ldots $ is an increasing sequence. Thus, there exist a unique iteration count $j_j$ for all $Y_j$ such that $Y_j^{\ldots, j_j}$ contains $v$ for the first time.
% It was not needed to add $+1$ to $j_l, j_{l+1}, \ldots$ but the authors found it easier to work with this definition, since $\rank{v} = (j_l, 0, j_{l-2}, \ldots, j_2, 0)$ implies that for all $Y_m$, $v \in Y_m^{j_l, 0, \ldots, j_m, 0, 0, \ldots, 0}$ and $v \not \in Y_m^{j_l, 0, \ldots, j_m-1, 0, 0, \ldots, 0}$.

% \section{Strategy Templates}
A ranking function obtained from a fixed-point computation as in \eqref{eq:rank} naturally gives rise to a positional winning strategy for the respective player in (normal) $\omega$-regular games that allow for positional strategies. The corresponding positional strategy is obtained by always choosing a \emph{minimum ranked successor} in the winning region.\footnote{See \cite{banerjee2022fast} for a similar construction of the positional winning strategy of \Even in \Odd-fair parity games} 
% 
% acquired in this exact manner from the $\mu-$calculus formula that solves parity games \IS{ref}, gives a positional winning \Even, or \Odd strategy. The only thing \Even, or \Odd needs
% to do to win, is to take its \emph{minimum ranked successor} from each vertex that is hers. A similar strategy would work for the \Even player in \Odd-fair parity games, with the formula that is the negation of ~\eqref{eq:fp-odd}. This is because \Even has positional strategies in \Odd-fair parity games.
We use this insight to obtain a \emph{candidate} maximal strategy template for player \Odd (which we prove to be also \emph{winning} in Prop.~\ref{prop:mainresult}) as follows.
% We simply start with the outlined minimum ranked positional strategy induced by the ranking in \eqref{} and 
% 
% 
% Our contribution in this section is to show that a minimum ranked successor strategy based on the ranking function we get from \eqref{eq:fp-odd} gives us a maximal winning \Odd strategy template if we make sure all \Odd nodes that are seen infinitely often, see both their minimum ranked successors and all their live edges infinitely often. This is what we called the \emph{"almost positionality"} of \Odd strategies in \Odd-fair parity games.
% 
% 
% According to the definition of an \Odd strategy template, all \Odd nodes that lie on a cycle in the template, needs to have all their live outgoing edges in the template. 
% And all \Odd strategies compliant with the template, sees each edge in the template infinitely often if its source node is seen infinitely often. 
% Therefore, all we need to do to get the above-mentioned \Odd strategy template, is to
We start with a subgraph on \Wo defining the minimum ranked successor strategy for \Odd induced by the ranking in \eqref{eq:rank}, and then iteratively add all live edges of nodes that lie on a cycle in the subgraph, to the subgraph. The saturated subgraph then defines a strategy template for \Odd, as formalized next. 

\begin{definition}[Rank-based Strategy Template]\label{def:S}
    Given an \Odd-fair parity game $\mathcal{G}^\ell = (\ltup{V, \Ve, \Vo, E, \chi}, E^\ell)$ with least even upper bound $l\geq 0$ on the priorities of nodes, winning region $\Wo\subseteq V$ and the ranking function $\rank{}: \Wo \to \mathbb{N}^{l}$ from Defn.~\ref{def:rank}, we define a strategy template $\Sc^{\mathcal{G}^\ell}=(\Wo,E')$ where $E'$ is constructed as follows:
   \begin{enumerate}\label{const:S}
   \item[(S1)] for all $v \in \Ve \cap \Wo$, add all $(v, w)\in E$ to $E'$;
   \item[(S2)] for all $v \in \Vo \cap \Wo$, add $(v,w)\in E$ to $E'$ for a $w$ with %$w$ is the successor of $v$ with minimum rank, i.e., 
   $w=argmin_{w'\in E(v)}\rank{w'}$ ($w$ is arbitrarily picked amongst the successors with the mimimum ranking);
   \item[(S3)] for all $v \in V^\ell\cap \Wo$, add all $(v,w)\in E^\ell$ to $E'$ if $v$ lays on a cycle in $\mathcal{S}^{\mathcal{G}^\ell}$;
   \item[(S4)] repeat item (S3) until no new edges are added.
   \end{enumerate}
   We call $\Sc^{\mathcal{G}^\ell}$ the \emph{minimum rank based maximal \Odd strategy template of $\mathcal{G}^\ell$}.
   \end{definition}
   
      \begin{example}\label{ex:3}
    $\Sc^{\mathcal{G}^\ell}$ for $\mathcal{G}^\ell$ from Ex.~\ref{ex:1} is depicted in Fig.~\ref{fig:ex1} (right). %We see that the live edges originating from vertex $4a$ are not contained in the template as $4a$ cannot be seen infinitely often if player \Odd chooses the minimal rank successor (i.e., moves to $2b$) upon the first visit to $4a$. After that, $4a$ cannot be visited again if player \Odd plays a strategy compliant with the strategy template.
   \end{example}
   
   It is clear from the definition that $\Sc^{\mathcal{G}^\ell}$ is an \Odd strategy template in $\mathcal{G}^\ell$. It is also maximal since each $v \in \Wo$ is assigned a rank. %The next subsection proves that it is also \emph{winning}.
   It remains to show that it is winning:

%    \subsection{Soundness of the Candidate Template}\label{sec:soundness_candidate}
%   It is easy to see that Thm.~\ref{thm:existence-maximaloddstrategytemplates} is proven, if we can show that $\Sc^{\mathcal{G}^\ell}$ is indeed winning. This is the main result of this section and formalized next.
%     In the main theorem of this section Thm. \ref{thm:mainresult} we claim that $\Sc^{\mathcal{G}^\ell}$ is winning.
    \begin{proposition}\label{prop:mainresult}
        Every player \Odd strategy compliant with $\Sc^{\mathcal{G}^\ell}$ is winning for \Odd in $\mathcal{G}^\ell$.
    \end{proposition}
    
    \noindent The full proof of Prop.~\ref{prop:mainresult} can be found in App.~\ref{app:counter-strategy-templates} and we only give a proof-sketch here.
    
    First, recall that $\Sc^{\mathcal{G}^\ell}$ is obtained by extending a minimum-rank based strategy as formalized in Def.~\ref{def:S}. Based on this we call a play $v_1 v_2 \ldots$ in $\Sc^{\mathcal{G}^\ell}$ \emph{minimal} if for all $v_i \in V_\Odd$, $v_{i+1}$ is the minimum ranked successor of $v_i$. We further call a cycle minimal, if it is a section of a minimal play.
    Now consider a play $\pi= v_0v_1\ldots$ which is compliant with $\Sc^{\mathcal{G}^\ell}$ and $v_0 \in \Wo$.  Since $\pi$ is compliant with an \Odd strategy template, it obeys the fairness condition. It is left to show that $\pi$ is \Odd winning.
    %This gives by induction that, a minimal cycle that passes through $v$ should be visited infinitely often. 
    %The main idea for proving Prop.~\ref{prop:mainresult} is now to show that any play that embeds an infinite minimal play is winning. 
    %The main idea for proving Prop.~\ref{prop:mainresult} is now to show that these minimal cycles are actually winning. 
    We do this by a chain of three observations,% formally proven in App.~\ref{app:counter-strategy-templates}:
    \begin{enumerate}
     \item If $\Wo \neq \emptyset$, there exists a non empty set $M := \{ v \in \Wo \mid \rank{v} = (1, 0, 1, 0, \ldots, 1, 0)\}$ (see Prop.~\ref{app-prop:Mexists}).
     \item All cycles in $\Sc^{\mathcal{G}^\ell}$ that pass through a vertex in $M$ are \Odd winning (see Prop.~\ref{app-prop:cycle-through-M}).
     \item All infinite minimal plays in $\Sc^{\mathcal{G}^\ell}$ visit $M$ infinitely often (see Prop.~\ref{app-prop:minimal-play-visits-M}).
    \end{enumerate}
    
    While item 1 simply follows from the observation that $(1,0,1,0 ,\ldots, 1, 0)$ is the minimum rank the ranking function assigns to a vertex and the set of nodes with this rank cannot be empty due to the monotonicity of \eqref{eq:fp-odd}, the proofs for item 2 and 3 are rather technical. %They require a careful analysis of the fixed-point algorithm in \eqref{eq:fp-odd} w.r.t.\ ranks over cycles within $\Sc^{\mathcal{G}^\ell}$ and are given in full detail in App.~\ref{app:counter-strategy-templates}.
    
    With the observations in item 1-3 being proven, we are ready to show that $\pi$ is \Odd winning. 
    Observe that $\pi = v_1 v_2 \ldots$ \enquote{embeds} an infinite minimal play, that is, there exists a subsequence $\pi' = v_{j_1} v_{j_2} \ldots$ of $\pi$ where $j_1 < j_2 < \ldots$ that is a minimal play. This is because whenever a $v \in V_\Odd \cap \Wo$ is seen infinitely often in $\pi$, $(v, v_{\min})$ is seen infinitely often as well, where $v_{\min}$ is the minimum-rank successor of $v$ in $\Sc^{\mathcal{G}^\ell}$.
    Since $\pi'$ visits $M$ infinitely often (from item 3), $\pi$ does so too.
    %This is because it embeds $\pi'$, which visists $M$ infinitely often (from item 3). %for any $v$ that is visited infinitely often in $\pi$, it's minimum-rank successor $v_{\min}$ is also visited infinitely often. 
    %This gives us an infinite minimal subsequence of $\pi$ and we know that all infinite minimal plays visit $M$ infinitely often (by item 3). %. Since all minimal plays visit $M$ infinitely often (from item 3) $\pi$ visits $M$ infinitely often. 
    Then due to pigeonhole principle, there exists an $x\in M$ that is visited infinitely often by $\pi$. Thus, a tail of $\pi$ can be seen as consecutive cycles over $x$. Since all cycles that pass through $M$ are \Odd winning (from item 2), we conclude that $\pi$ is \Odd winning.  
    
    Thm.~\ref{thm:existence-maximaloddstrategytemplates} now follows as a corollary of Prop.~\ref{prop:mainresult}.

\section{Zielonka's Algorithm for \Odd-Fair Parity Games}\label{sec:zielonka}
In this section, we construct a Zielonka-like algorithm that solves \Odd-fair parity games. We call this algorithm \emph{\Odd-fair Zielonka's algorithm}. We first recall Zielonka's original algorithm in Sec.~\ref{sec:zielonka:orig} and outline the changes imposed for our new \Odd-fair version in Sec.~\ref{sec:zielonka:fair}. We then discuss the correctness of this new algorithm in Sec.~\ref{sec:zielonka:correct}.

% and preview it for player \Even and \Odd, respectively, in Alg.~\ref{algo:fair-zielonka-odd} and Alg.~\ref{algo:fair-zielonka-even}.
% The differences between Zielonka's algorithm and \Odd-fair Zielonka's algorithm very minor and will be illustrated in Sec.~\ref{}.

% are small in sight. Namely, if one slightly changes the definition of $\SafeReach^f_\Even$ in Alg.~\ref{alg:fair-zielonka} and 
% replace the $\SafeReach^f_\Odd(V, G, \mathcal{G}^\ell)$ with $G$ on the last line of procedure $\SOLVE_\Odd(l, \mathcal{G}^\ell)$, one gets Zielonka's original algorithm.
% 
% The changes in the algorithm are actually small enough not to alter its worst-case computational complexity at all. However, the proof of one recursive call  of $\SOLVE_\Odd$ becomes substantially more complex with the involvement of live edges. 
% Due to space concerns, we only try to convey the main idea of the proof here and give the full proof in the appendix (App.~\ref{app:zielonka-proof}).

%  \vspace{0.5cm}
From now on we take $\mathcal{G}^\ell = \ltup{(V, V_\Even, V_\Odd, E, \chi), E^\ell}$ to be an \Odd-fair parity game. 

\subsection{Zielonka's Original Algorithm}\label{sec:zielonka:orig}
Intuitively, Zielonka's algorithm consists of two nested recursive functions, 
$\SOLVE_{\Even}(n,\mathcal{G})$ and $\SOLVE_{\Odd}(n,\mathcal{G})$ which compute \We and \Wo in a given parity game $\mathcal{G}$ with, respectively, even or odd upper bound priority $n$. Both functions recursively call each other on a sequence of sub-games that is constructed during the run of the algorithm. 

The main difference between Zielonka's original algorithm \cite{Zielonka98} and our new \Odd-fair version in Alg.~\ref{algo:fair-zielonka-bb} %(note that they are symmetrical except for the different \textbf{return} conditions) 
is the computation of the safe reachability set, denoted by $\SafeReach^f_\bb$ within the algorithms. 
%\footnote{Note that Alg.~\ref{algo:fair-zielonka-odd} and Alg.~\ref{algo:fair-zielonka-even} are symmetrical except for the different \textbf{return} conditions. In Alg.~\ref{algo:fair-zielonka-even} the last line can be changed to `\textbf{return} $\SafeReach_\Even^f(V, X, \mathcal{G}^\ell)$' to make the algorithms symmetrical, since it is equal to $X$ in this algorithm. However, we did not want to make this visual change to emphasize that the algorithm for \Even is slightly simpler than the one for \Odd, and to make it more visible that it mimics the original Zielonka's algorithm except for the definition of safe reachability sets.}. 
Intuitively, the safe reachability set of player \bb %in a (normal) game $G$ with safety set $S\subseteq V$ and reach set $R\subseteq V$, 
is the set of vertices from which \bb has a strategy to force the game into the reach set $R\subseteq V$, while staying in the safety set $S\subseteq V$. 
In a (normal) parity game $\mathcal{G}$ (without live edges), this set %the safe reachability set $\Xsr_\bb$ for player \bb 
can be computed via the single-nested fixed-point formula
\begin{equation}\label{equ:Xsr1}
 \Xsr_\bb:=\mu X~.~(S \cap (R \cup \Cpre_\bb(X))).
\end{equation}
If one interpretes Alg.~\ref{algo:fair-zielonka-bb} over (normal) parity games $\mathcal{G}$, defines $\SafeReach^f_\bb$ via \eqref{equ:Xsr1} for the respective player, and replaces $\SafeReach^f_\Odd(\cdot,X,\cdot)$ in the last return statement with $X$ (so, the algorithm returns $X$ for any $\Lambda$), one gets exactly Zielonka's algorithm for parity games.

\vspace*{-0.5cm}

\begin{figure}[htbp]
    \centering
    \begin{minipage}[b]{0.43\textwidth}
        \begin{algorithm}[H]
            \caption{\Odd-Fair Zielonka's Algo.}
            \label{algo:fair-zielonka-bb}
            \begin{algorithmic}\fontsmall
                \Procedure{$\SOLVE_{\bb}$}{$n$, $\mathcal{G}^\ell$} \vspace{0.05cm} 
                \State $X \gets V$ \vspace{0.1cm} 
                \State $Z_{\neg \bb} \gets G$ \label{lineo:Z_bb_initialize} \vspace{0.1cm} 
                \While{$Z_{\neg \bb} \neq \emptyset$}\label{lineo:while_start} \vspace{0.1cm} 
                    \State $N \gets \{ v \mid v \in X \text{ with } \chi(v) = n\}$ \label{lineo:N}
                    \State $Z \gets X \setminus \SafeReach^f_\bb(X, N, \mathcal{G}^\ell)$ \label{lineo:Z}
                    \State $Z_{\neg \bb} \gets \SOLVE_{\neg \bb}(n-1,\mathcal{G}^\ell[Z])$ \label{lineo:Z_bb}
                    \State $X \gets X \setminus \SafeReach^f_{\neg \bb}(X, Z_{\neg \bb}, \mathcal{G}^\ell)$\label{lineo:G_update}
                \EndWhile \label{lineo:while_end} \vspace{0.1cm}
                \If{\bb = \Even} \Return $X$  \vspace{0.1cm} \Else \,
                 \Return $\SafeReach_\Odd^f(V, X, \mathcal{G}^\ell)$ \EndIf \vspace{0.1cm} 
                \EndProcedure
            \end{algorithmic}
        \end{algorithm}
    \end{minipage}
    \hfill
    \begin{minipage}[b]{0.48\textwidth}
      \centering
      \includegraphics[width=1\textwidth]{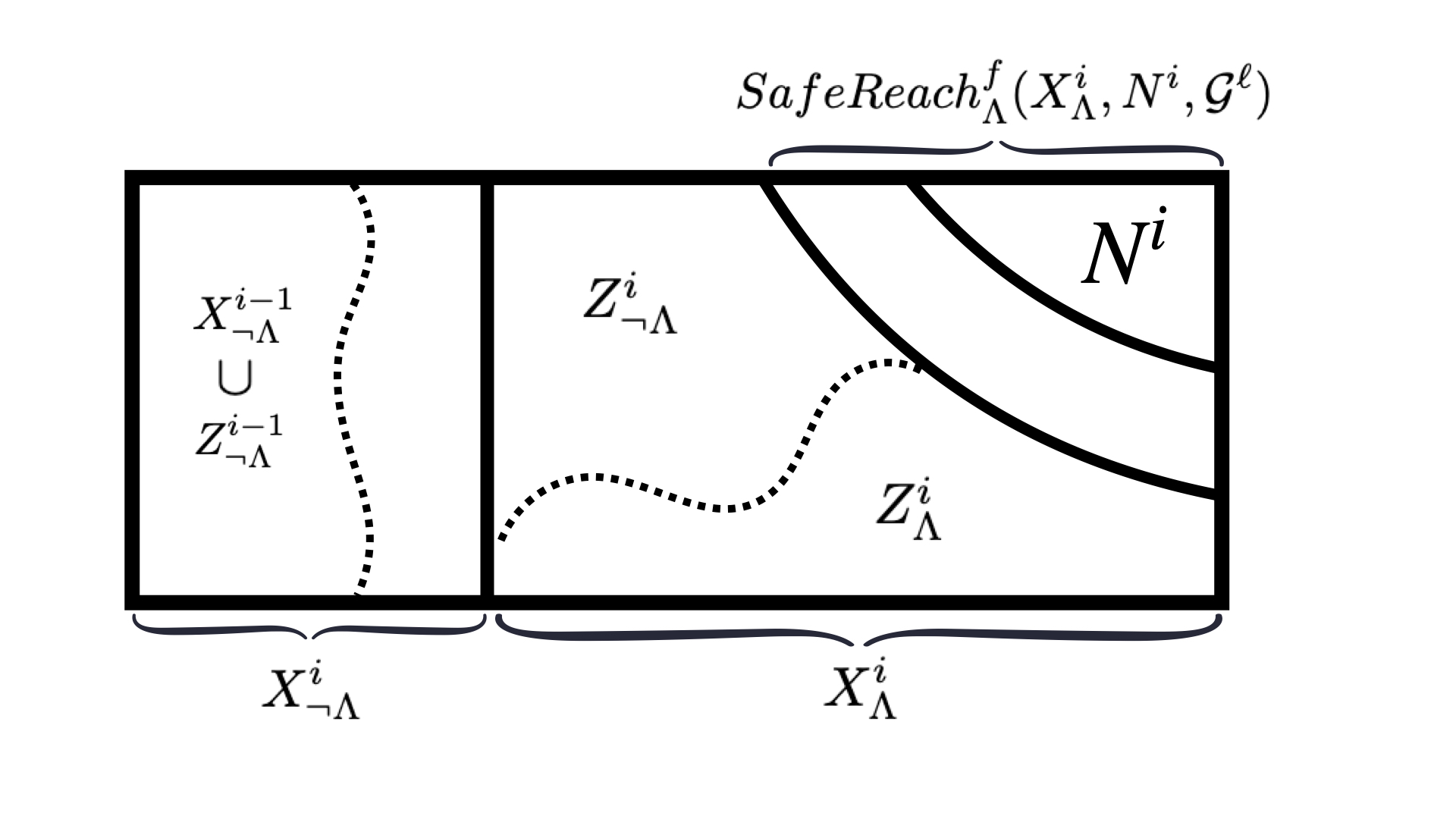}
      \caption{Visualization of the sets in Alg.~\ref{algo:fair-zielonka-bb}}
      \label{fig:kuesters-figure-extended}
    \end{minipage}
  \end{figure}

 %  \vspace*{-0.5cm}

\subsection{The \Odd-fair Zielonka's Algorithm}\label{sec:zielonka:fair}
We are now considering an \Odd-fair parity game $\mathcal{G}^\ell$. % with live edges on \Odd player vertices.
As discussed before, the main difference of the \Odd-fair Zielonka's algorithm from the original one lies in the construction of the safe reachability sets denoted by $\SafeReach^f_\bb$ in Alg.~\ref{algo:fair-zielonka-bb}. We therefore start by discussing its computation for both players.

\smallskip
\noindent\textbf{The \Odd Player.}
The first, somehow surprising, observation is that for player \Odd in \Odd-fair parity game $\mathcal{G}^\ell$, the safe reachability set $\Xsr_\Odd$ can still be computed via \eqref{equ:Xsr1}. This is due to the fact that $R$ only needs to be visited once, and 
\Even vertices do not have live outgoing edges that might prevent player \Odd from forcing a visit to $R$. 

In addition, we can extract a \emph{partial strategy template} for player \Odd from the iterative computation of \eqref{equ:Xsr1} via a similar, but much simpler ranking argument as used in Sec.~\ref{sec:strat-templates}. Here, $\rank{v} = 1$ for $v \in R$ and for the remaining vertices, 
$\rank{v}$ is the minimum integer $j$ for which $v \in X^j:=(S \cap (R \cup \Cpre_\Odd(X^{j-1})))$ where $X^0=\emptyset$. The positional strategy of \bb is then to take the minimum ranked successor from each \Odd node. 

Another way to think about this strategy is in the form of an acyclic subgraph of $\mathcal{G}^\ell$ on $\Xsr_\Odd$, where nodes in $R$ have no outgoing edges,
and for the remaining nodes, \Odd nodes have one outgoing edge and \Even nodes have all their outgoing edges. This is because if $v \in X^j\cap \Ve$, all outgoing edges achieve positive progress towards $R$, i.e. for all $(v, w) \in E$, $w \in X^{j-1}$.
Now it is easy to see that this subgraph almost defines a strategy template, i.e., on $\Xsr_\Odd\setminus R$, \Even nodes have all their outgoing edges in the subgraph, no \Odd node lies on a cycle and all of them have one outgoing edge. However, vertices in $R$ are dead-ends. We therefore call the strategy template induced by  \eqref{equ:Xsr1} \emph{partial} and denote it by $sr$. %\AKS{I think we need to properly formalize this template to use it in the next section}

\smallskip
\noindent\textbf{The \Even Player.}
It follows from the results of Banerjee et. al.~\cite{banerjee2022fast} that the safe reachability set $\Xsr_\Even$ of player \Even in \Odd-fair parity games requires the 2-nested fixed-point formula $\nu Y.\mu X.S \cap (R \cup \Apre(Y,X))$, which (via the operators defined in Sec.~\ref{sec:assump:prelim}) equals
\begin{equation}\label{equ:Xsr2}
 \Xsr_\Even: =~\nu Y~.~\mu X~.~S \cap (R \cup (\Cpre_\Even(X) \cup (\Lpre^{\exists}(X) \cap \Pre_\Odd^{\forall}(Y))))
\end{equation}%=&~\nu Y~.~\mu X~.~S \cap (R \cup \Apre(Y,X))\\
% \vspace{-0.1cm}
%with all predecessor operators defined in Sec.~\ref{sec:assump:prelim}.
% We denote this formula by $\SafeReach^l(S, R, \mathcal{G}^\ell)$. 
Intuitively, the necessity of a 2-nested formula arises from the following lack of information: we do not know in advance, which \Odd nodes need to lie on a cycle on a strategy template required for \Odd to win. If any positional strategy that lets \Odd win (i.e., to avoid $R$ or leave $S$) from a $v\in V^\ell$
requires $v$ to lie on a cycle, then \Odd has to take $v$'s live outgoing edges as well, and thus, it can enter $\Xsr_\Even$ and lose.
The calculation of \eqref{equ:Xsr2} starts with $Y^0 := V$, resulting in $\Pre_\Odd^{\forall}(V)=V$, hence
% 
% \vspace{-0.5cm}
\begin{equation}\label{equ:Xsr2a}
 Y^{1}:=\mu X~.~S\cap (R \cup \Cpre_\Even(X) \cup \Lpre^\exists(X)).
\end{equation}
% 
% \vspace{-0.1cm}
Due to the disappearence of $\Pre^\forall_\Odd(Y)$ in this iteration, intuitively all $v \in V^\ell$ are treated as if they do not have any positional winning \Odd strategy on them, so as if all \Odd strategies have to take all the live edges in the game. 
%This is due to the triviality of $\Pre^\forall_\Odd(Y) = V$ in this iteration. 
%When  $\Pre^\forall_\Odd(Y)$ vanishes from the equation, 
$Y^1$ includes any \Odd vertex that progresses towards $R$ while staying in $S$ with using either all its edges (due to $\Cpre_\Even(X)$) or through one live edge (due to $\Lpre^\exists(X)$). Thus, any vertex that manages to stay in $V \setminus Y^1$ does so due to being won by \Odd even if \Even could force all the live outgoing edges to be taken. 
Note that due to the monotonicity of fixed-point operators, for all $j$, $V \setminus Y^1 \subseteq V \setminus Y^j$.

Throughout the calculation, $ V \setminus Y^j$ keeps track of the nodes that have managed to escape $S$ or avoid $R$ in the previous iteration, so are `already' won by \Odd in the first $j$ iterations. The inner fixed-point calculation in the $(j+1)^{th}$ iteration treats $V \setminus Y^j$ as a subset of \Odd's winning region and it deems any node that can be forced by \Odd to reach $V \setminus Y^j$, lost by \Even.
When the algorithm saturates, $Y^\infty$ contains only those \Odd nodes that cannot be forced by \Odd to reach $V \setminus Y^\infty$, i.e., are won by \Even. Here it is important to observe that, $V \setminus Y^\infty$ contains some \Odd nodes that are not $V \setminus Y^1$. Since they are in $Y^1$, these nodes inductively %\todo{IS:iteratively? any other word?} 
reach \Even winning vertices through live edges. %(or, reach some vertices that have this property\IS{how to put this?}). 
This reveals that, all nodes in $V \setminus Y^j$ but not in $V \setminus Y^1$ win due to a positional \Odd strategy that reaches $V \setminus Y^{j-1}$. 
Iteratively, this reveals that all such nodes have positional \Odd strategies that make them reach $V \setminus Y^1$.%, and this is the key observation we will explot in the next section. 

%The key observation from the previous discussion is the following: if $v$ is won by \Odd even under the assumption that all the live edges in the game have to be taken by \Odd (which is an assumption that increases \Even's power), then $v$, or any node that can be forced by \Odd to reach $v$, is won by \Odd. Also, note that these nodes are exactly the nodes in $V \setminus Y^1$ as explained above.
The above alternative interpretation of the computation of $\Xsr_\Even$ in \eqref{equ:Xsr2} is the key insight that we utilize to define our new \Odd-fair Zielonka's algorithm, as discussed next.
% This is the key observation we will utilize within our new \Odd-fair Zielonka's algorithm as discussed next. %, and show that it carries over from \Odd-fair safe-reachability games to \Odd-fair parity games.
% \todo{IS: maybe we can get rid of the 'and show that it carries over...' part.}
%We utilize this key observation within our new \Odd-fair Zielonka's algorithm as discussed next, and thereby show that it carries over from \Odd-fair safe-reachability games to \Odd-fair parity games.

%As Banerjee et. al.~\cite{banerjee2022fast} have shown that this intuition carries over from safe-reachability games to different $\omega$-regular games under strong transition fairness, including Rabin and Parity games, we can utilize this observation for our \Odd-fair Zielonka's algorithm as follows.

\smallskip
\noindent\textbf{The \Odd-fair Zielonka's Algorithm.} %\todo{IS:I made some changes starting from here - until the end of this subsection}
% It remains to formally define the safe-reachability procedures used in Alg.~\ref{algo:fair-zielonka-odd} and Alg.~\ref{algo:fair-zielonka-even}. 
% 
Following up on the previous discussion, we use the following insight within the construction of the \Odd-fair Zielonka's algorithm. We assume the existence of a core subset $\Wo' \subseteq \Wo$ %\todo{this was in the text, I am not sure if it should be..: (actually corresponding to $ V \setminus Y^j$ in the above argumentation for \Odd-fair safe-reachability games)} 
that player \Odd can force all nodes in $\Wo$ %(i.e., the winning region of \Odd in the \Odd-fair parity game $\mathcal{G}^\ell$) 
to, that is winning for \Odd even under the assumption that \Even can force all the live edges in the game to be taken. %\footnote{We note that this is a similar insight used in the proof of winning strategy templates discussed in Sec.~\ref{sec:strat-templates}.}.  
% 
% We exploit this observation as follows: 
% \emph{Any $v \in \Wo$ can be made by \Odd to reach a core subset $\Wo' \subseteq \Wo$, that is won by \Odd even if all $v \in V^\ell \cap \Wo'$ have to take all their live edges.} 
% In $\SafeReach^l(S, R, \mathcal{G}^\ell)$, the first iteration of the $Y$ variable calculates this $\Wo'$. 
Since Zielonka's algorithm solves parity games by a sequence of nested safe-reachability calculations for alternating players, we apply the following trick:
Instead of computing $\Xsr_\Even$ via \eqref{equ:Xsr2} in each recursive call of Alg.~\ref{algo:fair-zielonka-bb}, we only compute $Y^1$ via \eqref{equ:Xsr2a} and use it as an \emph{overapproximation} of $\Xsr_\Even$ (which is indeed the case due to the monotonicity of \eqref{equ:Xsr2} in $Y$). 
That is, while we take the \Odd safe reachability set $\SafeReach^f_\Odd$ as the original (linear) \Odd safe reachability computation known for these games (given in~\eqref{equ:Xsr1}), we do not take \Even safe reachability formula $\SafeReach^f_\Even$ to be the (quadratic) \Even safe reachability computation known for these games (given in~\eqref{equ:Xsr2}),
but we instead take it as its (linear) subformula given in~\eqref{equ:Xsr2a} and arrive at an overapproximation of the \Even safe reachability region at the end of each $\SafeReach^f_\Even$ calculation. We finalize the recursive call $\SOLVE_\Odd$ by an extra call of $\SafeReach^f_\Odd$ applied to the (thus) underapproximated \Odd winning region in the sub-game, therefore expanding the returned \Odd winning region of the sub-game.  

By this, it turns out that the recursive call of $\SOLVE_\Odd(n, \mathcal{G}^\ell)$ actually computes $\Wo'$ as the set $X$ and we ensure that $\Wo$ is returned by the additional (linear) computation of $\SafeReach^f_\Odd$ over $X$ in the last return statement of Alg.~\ref{algo:fair-zielonka-bb}.
This instantiation of the safe-reachability computations is formalized next.

\begin{definition}\label{def:safereach}
 Given an \Odd-fair parity game $\mathcal{G}^\ell=\ltup{(V, V_\Even, V_\Odd, E, \chi), E^\ell}$ the safe-reachability procedures $\SafeReach^f_\Odd(S, R, \mathcal{G}^\ell)$ and $\SafeReach^f_\Even(S, R, \mathcal{G}^\ell)$ in Alg.~\ref{algo:fair-zielonka-bb} denote the iterative fixed-point computations in \eqref{equ:Xsr1} for \Odd and \eqref{equ:Xsr2a} for \Even.
\end{definition}
% 
% \noindent\textbf{The reason behind the computational advantage of Alg~\ref{algo:fair-zielonka-bb}.}\todo{One of the reviewers suggested that we highlight this part, since we have this part as the 'main achievement' of this technique. I think maybe we can add a small title like this one, or maybe some other title like just "computational advantage", or "reason of efficiency" or "computational succinctness" or something like this. }

\subsection{Complexity of the \Odd-fair Zielonka's Algorithm}\label{sec:zielonka:complexity}
The safe-reachability computations defined in Def.~\ref{def:safereach} have the same complexity as their computations via \eqref{equ:Xsr1} in Zielonka's original algorithm. The only difference is in the number of calculated $\Pre$ operations: while $\SafeReach_\Even$ from Zielonka's original algorithm~\eqref{equ:Xsr1} require the calculation of only one $\Pre$ operator, $\SafeReach_\Even^f$ from~\eqref{equ:Xsr2a} requires the calculation of 2 $\Pre$ operators. The additional final call of $\SafeReach^f_\Odd$ in $\SOLVE_\Odd$ procedure also has linear complexity and requires one $\Pre$ calculation. 
Therefore, not only the worst-case time complexity of Alg.~\ref{algo:fair-zielonka-bb} is equivalent to that of Zielonka's original algorithm (which would be the case even if we used the quadratic safe reachability formula from~\eqref{equ:Xsr2} for \Even since the overall complexity of the algorithm is exponential) but we create almost no additional computational overhead in the algorithm by introducing the fairness assumptions.

We further remark that Alg.~\ref{algo:fair-zielonka-bb} is not a straight-forward interpretation of the nested fixed-point in~\eqref{eq:fp-odd}, and its negation (see (14) in App. A.1 of~\cite{extended}) in the form of Zielonka's algorithm. 
\begin{comment}
%Such a straight forward approach would increase the number of $\Pre$ calls in each recursive step polynomially in the number of priorities, whereas in Alg.~\ref{algo:fair-zielonka-bb} we have at most 2 times (plus 1 extra for $\SOLVE_\Odd$) as many $\Pre$ calls compared to Zielonka's original algorithm.
Such a conversion would require the nested fixed-point formula to be turned into a fixed-point parity game \cite{}, which would basically be equivalent to the gadget-enhanced parity game from \cite{} and the state space of the resulting parity game would be polynomially increased with respect to the state space of the \Odd-fair parity game. More precisely, an efficient fixed-point parity game would have $\frac{3}{2}\cdot l \cdot|V|$-many states where $l$ is the number of priorities and $|V|$ is the number of states of the original \Odd-fair parity game, while keeping the same number of priorities. This approach would not only require a precomputation, but would also require a polynomial increase in the total number of $\Pre$ calculations performed by the resulting Zielonka's algorithm. However, in Alg.~\ref{algo:fair-zielonka-bb} we require at most 2 times (plus 1 extra for $\SOLVE_\Odd$) as many $\Pre$ calls compared to Zielonka's original algorithm at each recursive call.
\end{comment}
%\begin{comment}
%Such a conversion would be non-trivial due to $\Apre$ and $\Npre$ taking 2 different variables from 2 different iterations of the fixed-point calculation. 
Firstly, such a straightforward approach is non-trivial due to $\Apre$ and $\Npre$ operators taking two variables from two different iterations of the fixed-point calculation. 
Furthermore, 
at each \Even safe-reachability call of Alg.~\ref{algo:fair-zielonka-bb}, as mentioned we compute 2 $\Pre$ operators (equation~\ref{equ:Xsr2a}), whereas in each such corresponding step in the fixed-point iteration, we would have to compute 3 $\Pre$ operators due to the expansion of $\Apre$~\eqref{equ:apre} and $\Npre$~\eqref{equ:npre}.
%\end{comment}

%This is slightly surprising given that even the algorithm by Banerjee et. al. in \cite{banerjee2022fast}  %introduces a small overhead from $|V|^n$ (the complexity of \enquote{normal} parity fixed-point) to $|V|^{n+1}$ (the complexity of \Odd-fair parity fixed-point) where $n$ is the number of colors.
%\AKS{we should give the actual complexity here.}
% 
% As in the original Zielonka's algorithm, the Alg.~\ref{algo:fair-zielonka-odd} and Alg.~\ref{algo:fair-zielonka-even}
% 
% We first recall that the determinacy of \Odd-fair parity games follows from the results of Banerjee et. al.~\cite{banerjee2022fast}. 
% 
% We use two recursive calls to calculate \We and \Wo in an \Odd-fair parity game.
% $\SOLVE_\Odd(n,\mathcal{G}^\ell)$ takes an \Odd-fair parity game $\mathcal{G}^\ell$ with an \emph{odd} upper bound $n$ on the priorities of $V$ and returns $\Wo$ of $\mathcal{G}^\ell$.
% $\SOLVE_\Even(n,\mathcal{G}^\ell)$ returns \We for a $\mathcal{G}^\ell$ and an \emph{even} priority upper bound $n$ on $V$. Two calls recursively call each other, from a subgame $\mathcal{G}^\ell[X]$ that has priorities at most $ n - 1$.
% Since in the base cases $n = 0$ and $G = \emptyset$, the calls correctly return $\emptyset$, in order to prove the correctness of the algorithm, it is enough to 
% prove the correctness of each procedure, assuming the correctness of the other. 

It remains to show that \Odd-fair Zielonka's algorithm solves \Odd-fair parity games.

\subsection{Correctness of the \Odd-fair Zielonka's Algorithm}\label{sec:zielonka:correct}
% 
% 
% Now we will try to convey the idea of the correctness proof of Alg.~\ref{alg:fair-zielonka}. 
We first recall that \Odd-fair parity games are determined. %the results of Banerjee et. al.~\cite{banerjee2022fast}. Given an \Odd-fair parity game $\mathcal{G}^\ell=\ltup{(V, V_\Even, V_\Odd, E, \chi), E^\ell}$, we therefore know that \We and \Wo partition $V$. Following the original Zieloka's algorithms proof, it therefore remains to show that $\SOLVE_{\Even}(n,\mathcal{G}^\ell)$ and $\SOLVE_{\Odd}(n,\mathcal{G}^\ell)$ as in Alg.~\ref{algo:fair-zielonka-odd} and Alg.~\ref{algo:fair-zielonka-even} actually compute \We and \Wo, respectively.
Next, we prove the correctness of the algorithm by induction on $n$. Since in the base case $n = 0$ the calls correctly return $\emptyset$, it suffices to prove the correctness of each function, assuming the correctness of the other. This is formalized next. %for $\SOLVE_{EVEN}$ (Alg.~\ref{algo:fair-zielonka-even}), where $\SOLVE_{ODD}$ (Alg.~\ref{algo:fair-zielonka-odd}) follows from a symmetrical argument with odd $n$.
\begin{comment}
\begin{theorem}[Correctness of $\SOLVE_{\bb}$, Alg.~\ref{algo:fair-zielonka-bb}]\label{thm:solvebb}
Let $\mathcal{G}^\ell=\ltup{(V, V_\Even, V_\Odd, E, \chi), E^\ell}$ be an \Odd-fair parity game with \textsf{parity(\bb)}\footnote{\textsf{parity(\Odd)} is odd and  \textsf{parity(\Even)} is even.} upper bound priority $n$. Further, assume that for any \Odd-fair parity game $\mathcal{G'}^\ell$ with  \textsf{parity($\bb$)} upper bound priority $n'<n$ holds that 
 $\mathcal{W}_\bb[\mathcal{G'}^\ell]=\SOLVE_{\bb}(n',\mathcal{G'}^\ell)$ and for any  \Odd-fair parity game $\mathcal{G''}^\ell$ with \textsf{parity($\neg \bb$)} upper bound priority $n''<n$ holds that 
 $\mathcal{W}_{\neg \bb}[\mathcal{G''}^\ell]=\SOLVE_{\neg \bb}(n'',\mathcal{G''}^\ell)$. Then, $\mathcal{W}_\bb[\mathcal{G}^\ell]=\SOLVE_{\bb}(n,\mathcal{G}^\ell)$.
\end{theorem}
\end{comment}

\begin{theorem}[Correctness of $\SOLVE_{\bb}$, Alg.~\ref{algo:fair-zielonka-bb}]\label{thm:solvebb}
Assume that for any \Odd-fair parity game $\mathcal{G}'^\ell$ where $n' < n$ is an odd (resp. even) upper bound on the priorities of the game, $SOLVE_\Odd(n', \mathcal{G}')^\ell$ correctly returns the \Odd winning region (resp. $\SOLVE_\Even(n', \mathcal{G}')^\ell$ correctly returns the \Even winning region) in $\mathcal{G}'^\ell$. Then $SOLVE_\bb(n, \mathcal{G}^\ell)$ correctly returns the winning region of player $\bb$ where $n$ is even if $\bb= \Even$ and odd if $\bb = \Odd$.
\end{theorem}

%While the proof of Alg.~\ref{algo:fair-zielonka-odd} follows essentially the proof by Ralf K{\"u}sters \cite{Kuesters2002} of the original Zielonka's algorithm \cite{zielonkas-alg}, the proof of Alg.~\ref{algo:fair-zielonka-even} formalized by Thm.~\ref{thm:solveodd} becomes substantially more complex. First, our instantiation of $\SafeReach^f_\Even(S, R, \mathcal{G}^\ell)$ via \eqref{equ:Xsr2} only computes an \emph{overapproximation} of the safe reachability set $\Xsr_\Even$, and second, we must use \Odd \emph{winning strategy templates} instead of positional winning strategies, to prove a vertex to be winning. While the complete correctness proofs of both algorithms can be found in App.~\ref{app:zielonka-proof}, we give the intuition of Thm.~\ref{thm:solveodd} here, as this is the main contribution of this section. % In order to do so, we first define some preliminaries in Sec.~\ref{sec:zielonka:correctness:prelim}.
% We follow the notation of Ralf K{\"u}sters proof \cite{Kuesters2002} of the original Zielonka's algorithm \cite{Zielonka98}.% Let us first set up some preliminaries.

% \vspace{0.2cm}

\noindent\textbf{Notation.}
We follow the notation of K{\"u}sters' proof \cite{Kuesters2002} of Zielonka's original algorithm \cite{Zielonka98}.
%For the remainer of this section, take $\mathcal{G}^\ell = \ltup{(V, V_\Even, V_\Odd, E, \chi), E^\ell}$. 
Recall that $\mathcal{G}^\ell$ has no dead-ends. For some $X \subseteq V$, we call $\mathcal{G}^\ell[X] = \ltup{(X, X \cap V_\Even, X \cap V_\Odd, X \times X \subseteq E, \chi \mid_X), X \times X \subseteq E^\ell }$ 
a \emph{subgame} of $\mathcal{G}^\ell$ if it has no dead-ends. Here, $\chi\mid_X$ is the priority function $\chi : V \to \mathbb{N}$ restricted to domain $X$. Let $n$ be an upper bound on the priorities in $V$. If the parity of $n$ is even, set $\bb$ to $\Even$; if it's odd, set $\bb$ to \Odd. 

\vspace{0.2cm}

\noindent\textbf{\bb-trap and \bb-paradise.}
A $\bb$-trap is a subset $T \subseteq V$ for $\bb \in \{\Even, \Odd\}$ such that,
$\forall v \in T \cap V_{\nb},\,\, \exists (v, w)\in E \,\,\text{ with } w \in T$ and $\forall v \in T \cap V_{\bb},\,\, (v, w) \in E \implies w \in T$. 
A $\bb$-paradise in $\mathcal{G}^\ell$ is a subset $T \subseteq V$ which is a $\nb$-trap in $V$ and there exists a winning $\bb$ strategy template $(T, \e)$ in $\mathcal{G}^\ell$. %\AKS{don't we need that the vertices contained in this template are precisely $T$.}

\vspace{0.2cm}

The recursive calls of $\SOLVE_\bb$ and $\SOLVE_{\neg \bb}$ on subgames within Alg.~\ref{algo:fair-zielonka-bb} induce a characteristic partition of the game graph. For the correctness proof, 
we need to remember a series of these subgames that are constructed through previous recursive calls. The partition of these subsets is illustrated in Fig.~\ref{fig:kuesters-figure-extended} and formalized as follows.
% 
% \begin{figure}\centering 
%     \vspace{-0.5cm}\label{fig:kuesters-figure-extended}
%     \includegraphics[scale=0.10]{figures/Kuesters-figure-extended.jpeg}
%      \raisebox{1.8cm}[0pt][0pt]{%
%      \hspace{-10cm}
%      \parbox{10cm}{\caption{}\label{fig:kuesters-figure-extended}}}
%      \vspace{-1cm}
% \end{figure}

\vspace{-0.4cm}

\begin{align}\label{equ:seriesZielonka}
    &X_\bb^i := V \setminus X_\nb^i \quad \quad \quad &&N^i:= \{v \in X^i_\bb \mid \chi(v) = n\}\\
    &Z^i:= X^i_\bb \setminus \SafeReach^f_\bb(X^i_\bb, N^i, \mathcal{G}^\ell) \quad &&X^{i+1}_\nb :=  \SafeReach^f_\nb(V, X_\nb^{i} \cup Z_\nb^{i}, \mathcal{G}^\ell)\nonumber % X_\Even^{i} \cup \SafeReach_\Even^f(X^{i}_\Odd, Z_\Even^{i}, \mathcal{G}^\ell) )%\text{\todo{IS: I know the equality is not completely justified. The first one is cheaper for an algorithm pov, whereas the second one is easier to justify that $X^i_\Odd$ is an \Even-trap.}} 
\end{align}

%\vspace{-0.3cm}

where, in addition $Z_\bb^i$ is the \bb winning region in the subgame $\mathcal{G}^\ell[Z^i]$. Intuitively, the sets constructed in \eqref{equ:seriesZielonka} correspond to the sets with the same name within Alg.~\ref{algo:fair-zielonka-bb}.
    
We collect the following observations on these sets, which are proven in App.~\ref{app:zielonka-proof}. %and mimic the corresponding properties in the proof of the original Zielonka's proof \cite{Kuesters2002}.
\begin{enumerate}\label{it:zlk-observations}
\item  \textbf{(App. - Obs.~\ref{app-obs:traps-subgames})} $X^i_\nb$ is an \bb-trap, $X^i_\bb$, $Z^i$ and $Z_\bb^i$ are \nb-traps in $V$. $Z^i$ is in \nb-trap in $X_\bb$ and $Z_\nb^i, Z_\bb^i$ are \bb- and \nb-traps in $Z^i$, respectively. Therefore, $\mathcal{G}^\ell[Y]$ is a subgame of $\mathcal{G}^\ell$ with $Y$ being any of these sets.\label{it:obs1} %(see Obs.~\ref{obs:traps-subgames} in App.~\ref{}).
 \item \textbf{(App. - Lem.~\ref{app-lem:X_nb-equivalence})} $X_\nb^{i} \cup \SafeReach_\nb^f(X^{i}_\bb, Z_\nb^{i}, \mathcal{G}^\ell) =  \SafeReach_\nb^f(V, X_\nb^{i} \cup Z_\nb^{i}, \mathcal{G}^\ell)$.\label{it:obs2}%(see Lem.~\ref{lem:X_nb-equivalence} in App.~\ref{}).
 \item \textbf{(App.  - Cor.~\ref{app-cor:increasing-decreasing-sequences})} As a consequence of the previous item, $\{X_\nb^{i}\}_{i\in \mathbb{N}}$ is an increasing sequence. Consequently, $\{X_\bb^{i}\}_{i\in \mathbb{N}}$ is a decreasing sequence. As $V$ is finite, this immediately implies that these sequences reach a saturation value for some, and in fact the same, $k$. \label{it:obs3}
 \item \textbf{(App.  - Lem.~\ref{app-lem:safe-reach-Odd-paradise})} If $R \subseteq V$ is an \Odd-paradise in $\mathcal{G}^\ell$, then $\SafeReach^f_\Odd(V, R, \mathcal{G}^\ell)$ is also an \Odd-paradise in $\mathcal{G}^\ell$.\label{it:obs4}
 \item \textbf{(App.  - Lem.~\ref{app-lem:safereacheven-noliveedges})} The set $U \setminus \SafeReach_\bb(U, R, \mathcal{G}^\ell)$ is a $\bb$-trap in $U$. \label{it:obs5}
\end{enumerate}

\vspace{0.1cm}

In contrast to Zielonka's original algorithm, the proof of the procedures $\SOLVE_\Even$ and $\SOLVE_\Odd$ is not identical in \Odd-fair Zielonka's algorithm. This is due to the different safe-reachability set constructions used. Next we sketch the correctness proof of Thm.~\ref{thm:solvebb} for $\bb:=\Odd$, corresponding to the correctness of procedure $\SOLVE_\Odd$. The proof for $\bb:=\Even$ is left to the appendix, as it resembles the proof Zielonka's original algorithm more.
% % \subsubsection{Correctness of $\SOLVE_\Odd$ --  Thm.~\ref{thm:solvebb}}\label{sec:zielonka:correctness:odd}

\begin{proposition}\label{prop:n-odd}
Given the premisses of Thm.~\ref{thm:solvebb} for $\bb = \Odd$, if $Z_\Even^k = \emptyset$ then $\SafeReach^f_\Odd(V, X^k_\Odd, \mathcal{G}^\ell)$ is an \Odd-paradise and $V \setminus \SafeReach^f_\Odd(V, X^k_\Odd, \mathcal{G}^\ell)$ is an \Even-paradise in $\mathcal{G}^\ell$.
\end{proposition}

Within Prop.~\ref{prop:n-odd}, the fact that $Z_\Even^k = \emptyset$ refers to the termination of the recursive call in Alg.~\ref{algo:fair-zielonka-bb} which results in the saturation of the sequence $\{X_\Odd^i\}_{i\in \mathbb{N}}$ with $X_\Odd^k$. This implies that $\SOLVE_\Odd$ returns 
$ T:=\SafeReach^f_\Odd(V, X^k_\Odd, \mathcal{G}^\ell) $, which is an \Odd-paradise and $V \setminus T$ an \Even-paradise. With this, Thm.~\ref{thm:solvebb} follows from Prop.~\ref{prop:n-odd} for $\bb=\Odd$. % Alg.~\ref{algo:fair-zielonka-bb} for $\bb = \Odd$.
We now give a proof sketch of Prop.~\ref{prop:n-odd}.

We first recall from observation~\ref{it:obs1} that $T$ and $V\setminus T$ are \Even- and \Odd-traps in $V$, respectively. In order to prove Prop.~\ref{prop:n-odd}, it remains to show that there exists an \Odd (resp. \Even) strategy template which is winning in $\mathcal{G}^\ell$ and maximal on $T$ (resp. $V\setminus T$). We next give the construction of these templates and a high-level intuition on why they are actually \emph{winning}. 

\smallskip
\noindent\textbf{Winning \Odd Strategy Templates.} 
As $X^k_\Odd$ is known to be an \Even-trap, it can be proven to be an \Odd-paradise by constructing a winning maximal strategy template on it. It then follows from observation~\ref{it:obs4} that $T$ is also an \Odd-paradise.

Towards a construction of a maximal winning \Odd strategy template on $X_\Odd$, we first observe that $X^k_\Odd=Z_\Odd^k\cup \SafeReach^f_\Odd(X^k_\Odd, N^k, \mathcal{G}^\ell)$ (as $Z_\Even^k=\emptyset$). % Now we first consider $Z^k = Z_\Odd^k$ (as $Z_\Even^k=\emptyset$)\todo{IS: why do we 'consider' this, isn't this given for $k$?}
 Then there exists a maximal winning \Odd strategy template $z$ on $Z^k = Z_\Odd^k$ in game $\mathcal{G}^\ell[Z^k]$. % and the definition of $Z_\Odd^k$. 
 Any play $\pi$ compliant with $z$ that starts and stays in $Z^k$ is clearly \Odd winning.
However, $z$ is not necessarily an \Odd strategy template in $\mathcal{G}^\ell$ since there are possibly some $(v,w) \in E$ with $v \in Z^k \cap V_\Even$  and $w \not \in Z^k$.
For all such edges, $w \in \SafeReach^f_\Odd(X^k_\Odd, N^k, \mathcal{G}^\ell)$ since $X^k_\Odd$ is an \Even-trap in $V$.
For the state set $\Xsr_\Odd:=\SafeReach^f_\Odd(X^k_\Odd, N^k, \mathcal{G}^\ell)$, recall from Sec.~\ref{sec:zielonka:fair} that there exists partial strategy template $sr$ defined on $\Xsr_\Odd$ with dead ends in $N^k$.

Using the templates $z$ and $sr$, we can construct a maximal candidate \Odd strategy template on $X^k_\Odd$. Following the intuition behind the construction of $\mathcal{S}^{\mathcal{G}^\ell}$ in Def.~\ref{def:S}, we first define a base subgraph $(X^k_\Odd,\e)$ with $\e\subseteq E$ s.t.\
 $(v,w) \in E $ is in $\e$ if either (i) $(v,w) \in z \cup sr$, (ii) $v \in V_\Even \cap X^k_\Odd$, or (iii) $v \in N^k \cap V_\Odd$ and $w = v_r$
where $v_r$ is a random fixed successor of $v$, that is in $X^k_\Odd$.
Such a successor is guaranteed to exist since $X^k_\Odd$ is an \Even-trap.
We now extend the subgraph $(X^k_\Odd,\e)$ to an \Odd strategy template by adding all live edges originating in vertices  $X^k_\Odd\cap V^\ell$ that lie on a cycle in $\e$, similar to Def.~\ref{const:S} (S3)-(S4). This results in a subgraph $\Sc=(X^k_\Odd,\overline{\e})$ %where $\overline{e}$ is defined to be the saturation value of the sequence $\overline{e}^j = \overline{e}^{j-1} \cup \{(v, w) \in V^\ell \mid v \text{ lies on a cycle in } (X_\Odd^k, \overline{e}^{j-1})\}$ where $\overline{e}^0 = e$.
that is a maximal \Odd strategy template. %\AKS{Why don't we need $\Sc$ to be maximal on $T$?}
The underlying idea behind $\mathcal{S}$ being winning %(formally proven in App.~\ref{app:zielonka-proof}) 
is the following: Any play that starts in $X_\Odd^k$ either stays in $Z^k$ after some point and is won by $\mathcal{S}$ collapsing to $z$, or sees a newly added cycle (one that is not in $z \cup sr$) infinitely often. All such cycles contain a newly added edge. An analysis of newly added edges reveal that, 
all of them -- when seen infinitely often -- eventually drag a play towards $N^i$. Thus, every play that sees a new cycle infinitely often sees $n$ infinitely often, and thus won by \Odd.

\smallskip

\noindent\textbf{Winning \Even Strategy Templates.} 
Here we show that $V \setminus T$ is an \Even-paradise in $\mathcal{G}^\ell$. 
We first define $\Xsr^i_\Even:=\SafeReach^f_\Even(X_\Odd^i, Z_\Even^i, \mathcal{G}^\ell)$ and denote by $sr^i$ the partial \Even strategy template defined on $\Xsr^i_\Even$. We further denote the winning \Even strategy on $Z_\Even^i$ in game $\mathcal{G}^\ell[Z^i]$ by $z^i$. 
We can now construct the \Even strategy template $\mathcal{S} = (V \setminus T, \e)$ where $\e$ is the combination of edges in $sr^i \cup z^i$ with $\{(v,w) \in E \mid v \in V_\Odd \cap (V \setminus T)\}$.
Since $V\setminus T$ is an \Odd-trap by observation~\ref{it:obs5}, the edge set $\e$ stays within $V \setminus T$, i.e. $\e \subseteq V\setminus T \times V \setminus T$. Then clearly, $\mathcal{S}$ is an \Even strategy template.
To see $\mathcal{S}$ is winning we first observe that each $v \in V \setminus T$ there exists a unique $i<k$ such that $v \in \Xsr^i_\Even$. Let $\pi = v_1 v_2 \ldots$ be a play compliant with $\mathcal{S}$ and let $s = \Xsr_1 \Xsr_2 \ldots$ be the sequence such that $v_i \in \Xsr$.
(1) If $v_j \in Z_\Even^i$, $v_{j+1} \in Z_\Even^i \cup \{\Xsr_\Even^r \mid r < i\}$. This follows from $Z_\Even^i$ being an \Odd-trap in $X_\Odd^i$.
(2) If $\pi$ visits $v \in \Xsr^i$ infinitely often, $\pi$ visits $Z_\Even^i$ infinitely often: This is because $\pi$ visits the $(v,w)$ in $\mathcal{S}$ that makes positive progress towards $Z_\Even^i$ infinitely often as well. 
Let $i$ be the minimum index such that $\Xsr_\Even^i$ is seen infinitely often in $s$. By (1), $\pi$ visits $Z_\Even^i$ infinitely often and by (1) and the minimality of $i$, it should eventually stay in $Z_\Even^i$.
Thus $\mathcal{S}$ eventually collapses to $z_\Even^i$ on $\pi$ and the play is won by \Even.

\vspace{-0.2cm}
%\smallskip
%\noindent \textbf{The Algorithm.} Observe that $\SOLVE_\Odd(n, \mathcal{G}^\ell)$ (Alg.~\ref{algo:fair-zielonka-odd}) calculates the sets as given in the construction(Fig.~\ref{fig:kuesters-figure-extended}) where $X$ holds the value of $X_\Odd^i$ at the end of the $i^{th}$ iteration of it's \emph{while} loop. $\{X_\Odd^i\}_{i \in \mathbb{N}}$ is a decreasing sequence which saturates at some $X_\Odd^k$ where $Z_\Even^k = \emptyset$.
%$\SOLVE_\Odd(n, \mathcal{G}^\ell)$ returns  $\SafeReach(V, X_\Odd^k, \mathcal{G}^\ell)$, which is shown to be equal to \Wo by Prop.~\ref{prop:n-odd}. 

%Similarly, $X$ in $\SOLVE_\Even(n, \mathcal{G}^\ell)$ (Alg.~\ref{algo:fair-zielonka-even}) holds the value of $X_\Even^i$ after the $i^{th}$ iteration. The constructive proof of Thm.~\ref{thm:solveeven} (App.~\ref{app:zielonka-proof}) shows that the saturation value $X^{k'}_\Even$ (where $Z^{k'}_\Odd = \emptyset$) is equal to \We, and this is exacly the value $\SOLVE_\Even(n, \mathcal{G}^\ell)$ returns.
%

\subsection{Experimental Results}
% \vspace{-1mm}

We conducted an experimental study to empirically validate the claim that our new \Odd-fair Zielonka's algorithm retains its efficiency in practice (see App.~\ref{app:experiments} for details). 

We generated \Odd-fair parity instances manipulating $286$ benchmark instances of PGAME$\_$ Synth$\_$2021 dataset of the SYNTCOMP benchmark suite~\cite{syntcomp} and $51$ instances of PGSolver dataset of Keiren's benchmark suite~\cite{keirens} by adding live edges to the given (normal) parity games.
%We generated $286$ benchmark instances from the PGAME$\_$Synth$\_$2021 dataset of the SYNTCOMP benchmark suite~\cite{syntcomp} and $51$ benchmark instances from the PGSolver dataset of Keiren's benchmark suite~\cite{keirens} by adding live edges to the given (normal) parity games. 
We empirically compared the (non-optimized\footnote{While optimized version of \texttt{N-ZL} and \texttt{N-FP} are available in \texttt{oink} \cite{oink} our goal is a conceptual comparison, which is better achieved by similar (non-optimized) implementations for all algorithms.}) C++-based implementations of 
% For this, we implemented the following algorithms  in C++:
\begin{inparaenum}[(i)]
 \item the \Odd-fair Zielonka's algorithm (\texttt{OF-ZL}) from Alg.~\ref{algo:fair-zielonka-bb},
 \item the \enquote{normal} Zielonka's algorithm (\texttt{N-ZL}) from~\cite{Zielonka98}, 
 \item the fixed-point algorithm for \Odd-fair parity games (\texttt{OF-FP}) from~\cite{banerjee2022fast} implementing~\eqref{eq:fp-odd}, and
 \item the \enquote{normal} fixed-point algorithm (\texttt{N-FP}) for \enquote{normal} parity games from~\cite{EJ91}. % \todo{IS: a reviewer correctly asked us to add \texttt{N-FP} to the implemented algorithms list as well, and I did. Here for normal parity fixed-point algorithm should I cite an old paper that gives the parity fixed-point, or is it enough to cite Banerjee et. al. agan (as I did here) since the fixed-point formulation for regular parity is given there as well.}
\end{inparaenum}
On the \emph{SYNTCOMP benchmarks}, the time-out rates are: $82$ instances for \texttt{OF-FP}, $58$ for \texttt{OF-ZL}; $73$ for \texttt{N-FP} and $47$ for \texttt{N-ZL}. On the 204 instances that neither of the algorithms time out the average computation times are: $122.7$ seconds for \texttt{OF-FP}, $4.6$ seconds for \texttt{OF-ZL}, $45.2$ seconds for \texttt{N-FP} and $3.6$ seconds for \texttt{N-ZL}.   
For all instances that did not time out for all four algorithms, Fig.~\ref{fig:logscale-main} shows scatter plots comparing the computation times of \texttt{OF-ZL} with \texttt{OF-FP} (left) and  \texttt{OF-ZL} with \texttt{N-ZL} (right) using logarithmic scaling. The diagonal shows instances with similar computation times. Points above the diagonal show superior performance of \texttt{OF-ZL}.
For the \emph{PGSolver dataset} \texttt{OF-FP} timed out on all generated instances, whereas \texttt{OF-ZL} took $24.9$ seconds on average to terminate.

\begin{figure}%LOGSCALE
% \centering
% \begin{subfigure}{.5\textwidth}
%      \centering
%      \includegraphics[scale=0.28]{figures/experiments/logplot_regular_parity_fp_vs_zl.png}
%      %\caption{A subfigure}
%      \label{fig:logscale_regular}
% \end{subfigure}%
% \begin{subfigure}{.5\textwidth}
%      \centering
\begin{minipage}{0.5\textwidth}
\begin{center}
  \includegraphics[width=0.8\textwidth]{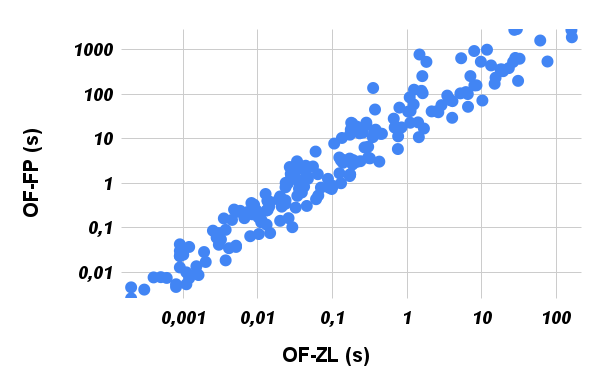}
\end{center}
\end{minipage}
\begin{minipage}{0.5\textwidth}
 \begin{center}
  \includegraphics[width=0.8\textwidth]{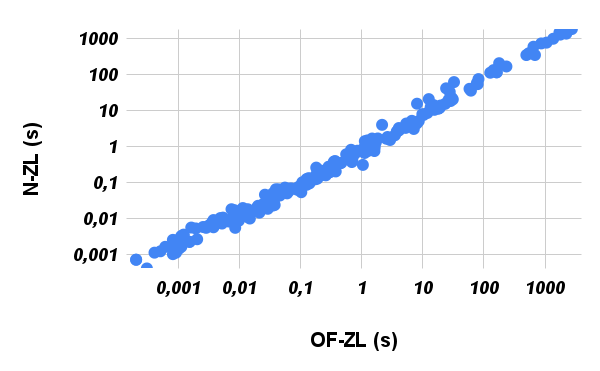}
 \end{center}
\end{minipage}
\vspace{-0.4cm}
\caption{Scatter plot for a comparative evaluation of \texttt{OF-ZL} vs. \texttt{OF-FP} (left), and  \texttt{OF-ZL} vs. \texttt{N-ZL} (right). Both plots show computation times in seconds using logarithmic scaling.}
\label{fig:logscale-main}
\end{figure}

We clearly see that  \texttt{OF-ZL} performs up to one order of magnitude better than \texttt{OF-FP} in many instances while \texttt{OF-ZL} and \texttt{N-ZL} perform very similar on the given benchmark instances. In addition, we observe that \texttt{OF-FP} starts timing out as soon as the examples became more complex. %, being especially sensitive to the increase in the number of priorities. 
%  On the other hand, \texttt{OF-ZL} preserves its performance considerably in the face of the increase in the same parameters. 
These outcomes match the known comparison results between the naive fixed-point calculation versus Zielonka's algorithm, on normal parity games.

\bibliography{references-wo-doi-url}

\begin{thebibliography}{10}

\bibitem{syntcomp}
The reactive synthesis competition.
\newblock URL: \url{http://www.syntcomp.org}.

\bibitem{AMT2013}
Rajeev Alur, Salar Moarref, and Ufuk Topcu.
\newblock Counter-strategy guided refinement of {GR(1)} temporal logic
  specifications.
\newblock In {\em Formal Methods in Computer-Aided Design, {FMCAD} 2013,
  Portland, OR, USA, October 20-23, 2013}, pages 26--33. {IEEE}, 2013.

\bibitem{AGR20}
Benjamin Aminof, Giuseppe~De Giacomo, and Sasha Rubin.
\newblock Stochastic fairness and language-theoretic fairness in planning in
  nondeterministic domains.
\newblock In J.~Christopher Beck, Olivier Buffet, J{\"{o}}rg Hoffmann, Erez
  Karpas, and Shirin Sohrabi, editors, {\em Proceedings of the Thirtieth
  International Conference on Automated Planning and Scheduling, Nancy, France,
  October 26-30, 2020}, pages 20--28. {AAAI} Press, 2020.

\bibitem{ANP21}
Andr{\'{e}} Arnold, Damian Niwiński, and Paweł Parys.
\newblock A quasi-polynomial black-box algorithm for fixed point evaluation.
\newblock In Christel Baier and Jean Goubault{-}Larrecq, editors, {\em 29th
  {EACSL} Annual Conference on Computer Science Logic, {CSL} 2021, January
  25-28, 2021, Ljubljana, Slovenia (Virtual Conference)}, volume 183 of {\em
  LIPIcs}, pages 9:1--9:23. Schloss Dagstuhl - Leibniz-Zentrum f{\"{u}}r
  Informatik, 2021.

\bibitem{baierbook}
Christel Baier and Joost{-}Pieter Katoen.
\newblock {\em Principles of model checking}.
\newblock {MIT} Press, 2008.

\bibitem{banerjee2022fast}
Tamajit Banerjee, Rupak Majumdar, Kaushik Mallik, Anne{-}Kathrin Schmuck, and
  Sadegh Soudjani.
\newblock Fast symbolic algorithms for omega-regular games under strong
  transition fairness.
\newblock {\em TheoretiCS}, 2, 2023.

\bibitem{belta2017formal}
Calin Belta, Boyan Yordanov, and Ebru~Aydin Gol.
\newblock {\em Formal methods for discrete-time dynamical systems}, volume~15.
\newblock Springer, 2017.

\bibitem{BJPPS2006}
Roderick Bloem, Barbara Jobstmann, Nir Piterman, Amir Pnueli, and Yaniv Sa'ar.
\newblock Synthesis of reactive(1) designs.
\newblock {\em J. Comput. Syst. Sci.}, 78(3):911--938, 2012.

\bibitem{CAFMR13}
Krishnendu Chatterjee, Luca de~Alfaro, Marco Faella, Rupak Majumdar, and
  Vishwanath Raman.
\newblock Code aware resource management.
\newblock {\em Formal Methods Syst. Des.}, 42(2):146--174, 2013.

\bibitem{CPRT03}
Alessandro Cimatti, Marco Pistore, Marco Roveri, and Paolo Traverso.
\newblock Weak, strong, and strong cyclic planning via symbolic model checking.
\newblock {\em Artif. Intell.}, 147(1-2):35--84, 2003.

\bibitem{DTV99}
Marco Daniele, Paolo Traverso, and Moshe~Y. Vardi.
\newblock Strong cyclic planning revisited.
\newblock In Susanne Biundo and Maria Fox, editors, {\em Recent Advances in
  {AI} Planning, 5th European Conference on Planning, ECP'99, Durham, UK,
  September 8-10, 1999, Proceedings}, volume 1809 of {\em Lecture Notes in
  Computer Science}, pages 35--48. Springer, 1999.

\bibitem{DIRS18}
Nicol{\'{a}}s D'Ippolito, Natalia Rodr{\'{\i}}guez, and Sebastian
  Sardi{\~{n}}a.
\newblock Fully observable non-deterministic planning as assumption-based
  reactive synthesis.
\newblock {\em J. Artif. Intell. Res.}, 61:593--621, 2018.

\bibitem{EJ89}
E.~Allen Emerson and Charanjit~S. Jutla.
\newblock On simultaneously determinizing and complementing omega-automata
  (extended abstract).
\newblock In {\em Proceedings of the Fourth Annual Symposium on Logic in
  Computer Science {(LICS} '89), Pacific Grove, California, USA, June 5-8,
  1989}, pages 333--342. {IEEE} Computer Society, 1989.

\bibitem{EJ91}
E.~Allen Emerson and Charanjit~S. Jutla.
\newblock Tree automata, mu-calculus and determinacy (extended abstract).
\newblock In {\em 32nd Annual Symposium on Foundations of Computer Science, San
  Juan, Puerto Rico, 1-4 October 1991}, pages 368--377. {IEEE} Computer
  Society, 1991.

\bibitem{EJ99}
E.~Allen Emerson and Charanjit~S. Jutla.
\newblock The complexity of tree automata and logics of programs.
\newblock {\em {SIAM} J. Comput.}, 29(1):132--158, 1999.

\bibitem{Francez}
Nissim Francez.
\newblock {\em Fairness}.
\newblock Springer-Verlag, Berlin, Heidelberg, 1986.

\bibitem{GH82}
Yuri Gurevich and Leo Harrington.
\newblock Trees, automata, and games.
\newblock In Harry~R. Lewis, Barbara~B. Simons, Walter~A. Burkhard, and
  Lawrence~H. Landweber, editors, {\em Proceedings of the 14th Annual {ACM}
  Symposium on Theory of Computing, May 5-7, 1982, San Francisco, California,
  {USA}}, pages 60--65. {ACM}, 1982.

\bibitem{HS21}
Daniel Hausmann and Lutz Schr{\"{o}}der.
\newblock Quasipolynomial computation of nested fixpoints.
\newblock In Jan~Friso Groote and Kim~Guldstrand Larsen, editors, {\em Tools
  and Algorithms for the Construction and Analysis of Systems - 27th
  International Conference, {TACAS} 2021, Held as Part of the European Joint
  Conferences on Theory and Practice of Software, {ETAPS} 2021, Luxembourg
  City, Luxembourg, March 27 - April 1, 2021, Proceedings, Part {I}}, volume
  12651 of {\em Lecture Notes in Computer Science}, pages 38--56. Springer,
  2021.

\bibitem{Jurdzinski00}
Marcin Jurdzinski.
\newblock Small progress measures for solving parity games.
\newblock In Horst Reichel and Sophie Tison, editors, {\em {STACS} 2000, 17th
  Annual Symposium on Theoretical Aspects of Computer Science, Lille, France,
  February 2000, Proceedings}, volume 1770 of {\em Lecture Notes in Computer
  Science}, pages 290--301. Springer, 2000.

\bibitem{JMT22}
Marcin Jurdzinski, R{\'{e}}mi Morvan, and K.~S. Thejaswini.
\newblock Universal algorithms for parity games and nested fixpoints.
\newblock In Jean{-}Fran{\c{c}}ois Raskin, Krishnendu Chatterjee, Laurent
  Doyen, and Rupak Majumdar, editors, {\em Principles of Systems Design -
  Essays Dedicated to Thomas A. Henzinger on the Occasion of His 60th
  Birthday}, volume 13660 of {\em Lecture Notes in Computer Science}, pages
  252--271. Springer, 2022.

\bibitem{keirens}
Jeroen J.~A. Keiren.
\newblock Benchmarks for parity games.
\newblock In Mehdi Dastani and Marjan Sirjani, editors, {\em Fundamentals of
  Software Engineering - 6th International Conference, {FSEN} 2015 Tehran,
  Iran, April 22-24, 2015, Revised Selected Papers}, volume 9392 of {\em
  Lecture Notes in Computer Science}, pages 127--142. Springer, 2015.

\bibitem{Klarlund90}
Nils Klarlund.
\newblock {\em Progress Measures and Finite Arguments for Infinite
  Computations}.
\newblock PhD thesis, Cornell University, {USA}, 1990.

\bibitem{Klarlund94}
Nils Klarlund.
\newblock Progress measures, immediate determinacy, and a subset construction
  for tree automata.
\newblock {\em Ann. Pure Appl. Log.}, 69(2-3):243--268, 1994.

\bibitem{KlarlundKozen91}
Nils Klarlund and Dexter Kozen.
\newblock Rabin measures and their applications to fairness and automata
  theory.
\newblock In {\em Proceedings of the Sixth Annual Symposium on Logic in
  Computer Science {(LICS} '91), Amsterdam, The Netherlands, July 15-18, 1991},
  pages 256--265. {IEEE} Computer Society, 1991.

\bibitem{Kozen:muCalculus}
Dexter Kozen.
\newblock Results on the propositional mu-calculus.
\newblock {\em Theor. Comput. Sci.}, 27:333--354, 1983.

\bibitem{KFP2007}
Hadas Kress{-}Gazit, Georgios~E. Fainekos, and George~J. Pappas.
\newblock Where's waldo? sensor-based temporal logic motion planning.
\newblock In {\em 2007 {IEEE} International Conference on Robotics and
  Automation, {ICRA} 2007, 10-14 April 2007, Roma, Italy}, pages 3116--3121.
  {IEEE}, 2007.

\bibitem{KFP2009}
Hadas Kress{-}Gazit, Georgios~E. Fainekos, and George~J. Pappas.
\newblock Temporal-logic-based reactive mission and motion planning.
\newblock {\em {IEEE} Trans. Robotics}, 25(6):1370--1381, 2009.

\bibitem{Kuesters2002}
Ralf K{\"{u}}sters.
\newblock Memoryless determinacy of parity games.
\newblock In Erich Gr{\"{a}}del, Wolfgang Thomas, and Thomas Wilke, editors,
  {\em Automata, Logics, and Infinite Games: {A} Guide to Current Research
  [outcome of a Dagstuhl seminar, February 2001]}, volume 2500 of {\em Lecture
  Notes in Computer Science}, pages 95--106. Springer, 2001.

\bibitem{MMSS2021}
Rupak Majumdar, Kaushik Mallik, Anne{-}Kathrin Schmuck, and Sadegh Soudjani.
\newblock Symbolic control for stochastic systems via parity games.
\newblock {\em CoRR}, abs/2101.00834, 2021.

\bibitem{MR2015}
Shahar Maoz and Jan~Oliver Ringert.
\newblock Synthesizing a lego forklift controller in {GR(1):} {A} case study.
\newblock In Pavol Cern{\'{y}}, Viktor Kuncak, and Parthasarathy Madhusudan,
  editors, {\em Proceedings Fourth Workshop on Synthesis, {SYNT} 2015, San
  Francisco, CA, USA, 18th July 2015}, volume 202 of {\em {EPTCS}}, pages
  58--72, 2015.

\bibitem{Martin75}
Donald~A. Martin.
\newblock Borel determinacy.
\newblock {\em Annals of Mathematics}, 102(2):363--371, 1975.

\bibitem{NOL17}
Petter Nilsson, Necmiye Ozay, and Jun Liu.
\newblock Augmented finite transition systems as abstractions for control
  synthesis.
\newblock {\em Discret. Event Dyn. Syst.}, 27(2):301--340, 2017.

\bibitem{Parys19}
Paweł Parys.
\newblock Parity games: Zielonka's algorithm in quasi-polynomial time.
\newblock In Peter Rossmanith, Pinar Heggernes, and Joost{-}Pieter Katoen,
  editors, {\em 44th International Symposium on Mathematical Foundations of
  Computer Science, {MFCS} 2019, August 26-30, 2019, Aachen, Germany}, volume
  138 of {\em LIPIcs}, pages 10:1--10:13. Schloss Dagstuhl - Leibniz-Zentrum
  f{\"{u}}r Informatik, 2019.

\bibitem{PT01}
Marco Pistore and Paolo Traverso.
\newblock Planning as model checking for extended goals in non-deterministic
  domains.
\newblock In Bernhard Nebel, editor, {\em Proceedings of the Seventeenth
  International Joint Conference on Artificial Intelligence, {IJCAI} 2001,
  Seattle, Washington, USA, August 4-10, 2001}, pages 479--486. Morgan
  Kaufmann, 2001.

\bibitem{QS83}
Jean{-}Pierre Queille and Joseph Sifakis.
\newblock Fairness and related properties in transition systems - {A} temporal
  logic to deal with fairness.
\newblock {\em Acta Informatica}, 19:195--220, 1983.

\bibitem{RS14}
Miquel Ram{\'{\i}}rez and Sebastian Sardi{\~{n}}a.
\newblock Directed fixed-point regression-based planning for non-deterministic
  domains.
\newblock In Steve~A. Chien, Minh~Binh Do, Alan Fern, and Wheeler Ruml,
  editors, {\em Proceedings of the Twenty-Fourth International Conference on
  Automated Planning and Scheduling, {ICAPS} 2014, Portsmouth, New Hampshire,
  USA, June 21-26, 2014}. {AAAI}, 2014.

\bibitem{extended}
Irmak Sağlam and Anne-Kathrin Schmuck.
\newblock Solving odd-fair parity games (extended abstract), 2023.
\newblock \href {http://arxiv.org/abs/2307.13396} {\path{arXiv:2307.13396}}.

\bibitem{Schewe-strategy-improvement}
Sven Schewe.
\newblock An optimal strategy improvement algorithm for solving parity and
  payoff games.
\newblock In Michael Kaminski and Simone Martini, editors, {\em Computer
  Science Logic, 22nd International Workshop, {CSL} 2008, 17th Annual
  Conference of the EACSL, Bertinoro, Italy, September 16-19, 2008.
  Proceedings}, volume 5213 of {\em Lecture Notes in Computer Science}, pages
  369--384. Springer, 2008.

\bibitem{SE84}
Robert~S. Streett and E.~Allen Emerson.
\newblock The propositional mu-calculus is elementary.
\newblock In Jan Paredaens, editor, {\em Automata, Languages and Programming,
  11th Colloquium, Antwerp, Belgium, July 16-20, 1984, Proceedings}, volume 172
  of {\em Lecture Notes in Computer Science}, pages 465--472. Springer, 1984.

\bibitem{SKCCB2015}
Mar{\'{\i}}a Svorenov{\'{a}}, Jan Kret{\'{\i}}nsk{\'{y}}, Martin Chmelik,
  Krishnendu Chatterjee, Ivana Cern{\'{a}}, and Calin Belta.
\newblock Temporal logic control for stochastic linear systems using
  abstraction refinement of probabilistic games.
\newblock In Antoine Girard and Sriram Sankaranarayanan, editors, {\em
  Proceedings of the 18th International Conference on Hybrid Systems:
  Computation and Control, HSCC'15, Seattle, WA, USA, April 14-16, 2015}, pages
  259--268. {ACM}, 2015.

\bibitem{Tabuada2009}
Paulo Tabuada.
\newblock {\em Verification and Control of Hybrid Systems - {A} Symbolic
  Approach}.
\newblock Springer, 2009.

\bibitem{thistle1998control}
John~G Thistle and RP~Malham{\'e}.
\newblock Control of $\omega$-automata under state fairness assumptions.
\newblock {\em Systems \& control letters}, 33(4):265--274, 1998.

\bibitem{van-Dijk-tangle-learning}
Tom van Dijk.
\newblock Attracting tangles to solve parity games.
\newblock In Hana Chockler and Georg Weissenbacher, editors, {\em Computer
  Aided Verification - 30th International Conference, {CAV} 2018, Held as Part
  of the Federated Logic Conference, FloC 2018, Oxford, UK, July 14-17, 2018,
  Proceedings, Part {II}}, volume 10982 of {\em Lecture Notes in Computer
  Science}, pages 198--215. Springer, 2018.

\bibitem{oink}
Tom van Dijk.
\newblock Oink: An implementation and evaluation of modern parity game solvers.
\newblock In Dirk Beyer and Marieke Huisman, editors, {\em Tools and Algorithms
  for the Construction and Analysis of Systems - 24th International Conference,
  {TACAS} 2018, Held as Part of the European Joint Conferences on Theory and
  Practice of Software, {ETAPS} 2018, Thessaloniki, Greece, April 14-20, 2018,
  Proceedings, Part {I}}, volume 10805 of {\em Lecture Notes in Computer
  Science}, pages 291--308. Springer, 2018.

\bibitem{WEK18}
Kai~Weng Wong, R{\"{u}}diger Ehlers, and Hadas Kress{-}Gazit.
\newblock Resilient, provably-correct, and high-level robot behaviors.
\newblock {\em {IEEE} Trans. Robotics}, 34(4):936--952, 2018.

\bibitem{Zielonka98}
Wiesław Zielonka.
\newblock Infinite games on finitely coloured graphs with applications to
  automata on infinite trees.
\newblock {\em Theor. Comput. Sci.}, 200(1-2):135--183, 1998.

\end{thebibliography}
%
%\begin{thebibliography}{8}
%\bibitem{ref_article1}
%Author, F.: Article title. Journal \textbf{2}(5), 99--110 (2016)

%\bibitem{ref_lncs1}
%Author, F., Author, S.: Title of a proceedings paper. In: Editor,
%F., Editor, S. (eds.) CONFERENCE 2016, LNCS, vol. 9999, pp. 1--13.
%Springer, Heidelberg (2016). \doi{10.10007/1234567890}

%\end{thebibliography}
\newpage 

\appendix
\section{Appendix}
\subsection{Proof of the Fixed-point Formula for \Wo}\label{app:fp-proof}
It was recently shown in \cite{banerjee2022fast} that the winning region $\We$ for \Even in an \Odd-fair parity game $\mathcal{G}^\ell$ with least even upper bound priority $l\geq 0$ can be computed by the fixed-point algorithm 
\begin{align}\label{eq:fp-even}
 \We = &\nu {Y_l}.~ \mu X_{l-1}.~ \ldots \nu{Y_2}.~ \mu{X_1}.~ \bigcup_{j \in \ev{2}{l}} \A_j \quad \\
 & \text{ where, } \quad \A_j := \left(C_j \cap Cpre_\Even(Y_j)\right) \cup \left(\left(\textstyle\bigcup_{i \in [1,j-1]}C_i\right) \cap \Apre(Y_j, X_{j-1})\right)\nonumber
\end{align}

As  \Odd-fair parity games are determined, we can simply compute the winning region for player $\Odd$ by negating \eqref{eq:fp-even}, which leads to Prop.~\ref{prop: W_Odd}. For the sake of self-containment, we restate Prop.~\ref{prop: W_Odd} here.

% restating prop 2
\begingroup
\def\theproposition{\ref{prop: W_Odd}}
\begin{proposition}
    Given an \Odd-fair parity game $\mathcal{G}^\ell = (\ltup{V, \Ve, \Vo, E, \chi}, E^\ell)$ with least even upper bound $l\geq 0$ and
\begin{align}\label{eq:fp-odd-app}
    Z := &\mu {Y_l}.~  \nu {X_{l-1}}.~  \ldots \mu{Y_2}.~  \nu{X_1}.~  \bigcap_{j \in \ev{2}{l}} \B_j, \\
    &\text{ where} \quad
    \B_j := \left(\textstyle\bigcup_{i \in [j+1,l]} C_i\right) \cup \left(\overline{C_j} \cap \Npre(Y_j, X_{j-1}) \right) \cup \left(C_j \cap \Cpre_\Odd(Y_j)\right)\nonumber
\end{align}
then $\Phi=\Wo$.
Further, it takes $\mathcal{O}(n^{l+1})$ symbolic steps to compute $\Wo$ via \eqref{eq:fp-odd}.
\end{proposition}
\addtocounter{proposition}{-1} % decrease the counter that holds proposition numbers, so that the previous restated proposition is not seen.
\endgroup

\begin{proof}
We use the negation rule of the $\mu$-calculus, i.e., $\neg (\mu X~.~F(X))=\nu X~.~\neg F(\neg X)$, to negate \eqref{eq:fp-even}. Using the equivalences in \eqref{equ:Preseq} and \eqref{equ:cpre_equal} and common De-Morgan laws, we get 
\begin{subequations}
\begin{align}
 \neg \A_j(\neg Y_j,\neg X_{j-1})=&\left(\overline{C_j} \cup \Cpre_\Odd(Y_j)\right) \cap \left(\left(\textstyle\bigcup_{i \in [j,l]} C_i\right) \cup \Npre(Y_j, X_{j-1})\right)\\
 =&\left(\textstyle\bigcup_{i \in [j+1, l]} C_i\right) \cup \left(\overline{C_j} \cap \Npre(Y_j, X_{j-1})\right) \nonumber\\
 &\cup \left(C_j \cap \Cpre_\Odd(Y_j)\right) \cup \left(\Cpre_\Odd(Y_j) \cap \Npre(Y_j, X_{j-1})\right)\label{eq2}\\
 =&\left(\textstyle\bigcup_{i \in [j+1,l]} C_i\right) \cup \left(\overline{C_j} \cap \Npre(Y_j, X_{j-1}) \right) \cup \left(C_j \cap \Cpre_\Odd(Y_j)\right)
\end{align}\end{subequations}
where the last equivalence follows from the observation that the last term of \eqref{eq2} is redundant since it is a subset of both $ \Npre(Y_j, X_{j-1})$ and $\Cpre_\Odd(Y_j)$: If a $v$ is in the last term, it either has priority $j$, in which case it is already in $C_j \cap \Cpre_\Odd(Y_j)$, or it has a different priority, in which case it is already in $\Npre(Y_j, X_{j-1})$. %As all formal variables can have arbitrary symbols, we just rename them to their non-overlined versions but swap the preceding $\mu$/$\nu$ operators. This yields \eqref{eq:fp-odd} from negating \eqref{eq:fp-even}.
\end{proof}

\subsection{Proof of Prop.~\ref{prop:mainresult}}\label{app:counter-strategy-templates}
    We will restate the fixed-point formula that calculates the \Odd winning region and the main proposition for the sake of self-containment.
    
% restating prop 2
\begingroup
\def\theproposition{\ref{eq:fp-odd}}
\begin{proposition}\label{app-eq:fp-odd}
    Given an \Odd-fair parity game $\mathcal{G}^\ell = (\ltup{V, \Ve, \Vo, E, \chi}, E^\ell)$ with least even upper bound $l\geq 0$ it holds that $Z=\Wo$, where
    \begin{align}
        Z &:=\textstyle \mu {Y_l}.~  \nu {X_{l-1}}.~  \ldots \mu{Y_2}.~  \nu{X_1}.~  \bigcap_{j \in \ev{2}{l}} \B_j[Y_j, X_{j-1}], \\
        &\text{ where} \quad
        \B_j[\mathbf{Y}, \mathbf{X}] := \left(\textstyle\bigcup_{i \in [j+1,l]} C_i\right) \cup \left(\overline{C_j} \cap \Npre(\mathbf{Y}, \mathbf{X}) \right) \cup \left(C_j \cap \Cpre_\Odd(\mathbf{Y})\right).\nonumber
    \end{align}
     then $Z=\Wo$.
     Further, it takes $\mathcal{O}(n^{l+1})$ symbolic steps to compute $Z$.
\end{proposition}
\addtocounter{proposition}{-1} % decrease the counter that holds proposition numbers, so that the previous restated proposition is not seen.
\endgroup
    
    \begingroup
    \def\theproposition{\ref{prop:mainresult}}
    \begin{proposition}\label{app-prop:mainresult}
        Every player \Odd strategy compliant with $\Sc^{\mathcal{G}^\ell}$ is winning for \Odd in $\mathcal{G}^\ell$.

    \end{proposition}
    \addtocounter{proposition}{-1} % decrease the counter that holds proposition numbers, so that the previous restated proposition is not seen.
    \endgroup

    The main observation behind the proof of Prop.~\ref{app-prop:mainresult} is similar to the main observation in Sec.~\ref{sec:zielonka}, leading to the proof of Alg.~\ref{algo:fair-zielonka-bb}.
    That is, there exists a core subset of the \Odd winning region $\Wo'\subseteq \Wo$, that is added to $Z$ in the first iteration of the 
    fixed-point calculation in ~\eqref{eq:fp-odd}, to which each $v \in \Wo$ can be made to reach by \Odd. 
    Here in particular, we show that any \Odd strategy compliant with $\Sc^{\mathcal{G}^\ell}$ reaches $\Wo'$ (infinitely often) while obeying the fairness condition, and is thus winning for \Odd.

    The proof of Prop.~\ref{prop:mainresult} consists of $3$ main propositions. Before we present them, we will gather some observations from the fixed-point formula ~\eqref{app-eq:fp-odd} and present them as lemmas.
    
    According to our previous definitions, $Y_m^{r_l, r_{l-1}, \ldots, r_m}$ denotes the value of $Y_m$ variable after the $r_{m}^{th}$ iteration on it, while $Y_i, X_i$ variables for $i>m$ are in their ${{r_i}+1}^{th}$ iterations.
    If we flatten this formula we get the following equality: $Y_m^{r_l, r_{l-1}, \ldots, r_m} = $
    $$\nu X_{m-1}\ldots \mu Y_2 \nu X_1. \bigcap_{j \in \ev{m+2}{l}} \B_j[Y_j^{r_j}, X_{j-1}^{r_{j-1}}] \cap \B_m[Y_m^{r_m-1}, X_{m-1}] \cap \bigcap_{j \in \ev{2}{m-2}} \B_j[Y_j, X_{j-1}]$$

    Observe that when the fixed-point above is calculated, all $X_{j}, Y_j$ values for $j < m$ will saturate at the same value,
    which is the final result of the computation. That is, 
    \begin{lemma}\label{app-obs:flat-Z}
    $$ Y_m^{r_l, \ldots, r_m} = \bigcap_{j \in \ev{m+2}{l}} \B_j[Y_j^{r_j}, X_{j-1}^{r_{j-1}}] \cap \B_m[Y_m^{r_m-1},  Y_m^{r_l, \ldots, r_m}] \cap \bigcap_{j \in \ev{2}{m-2}} \B_j[Y_m^{r_l, \ldots, r_m}, Y_m^{r_l, \ldots, r_m}]$$
    \end{lemma}

    \begin{lemma}\label{app-lem:intersection_of_Y_m}   
        For all $v \in \Wo$ with $\rank{v} = (r_l, 0, \ldots, r_2, 0)$. Then,
        $$v \in \bigcap_{j \in \ev{2}{l}} Y_j^{r_l-1, 0, r_{l-2}-1, 0, \ldots, r_{j-2}-1, 0, r_j}$$
    \end{lemma}
        This is similar to our previous observation. $\rank{v} = (r_l, 0, \ldots, r_2, 0)$ implies $v$ was added to the formula while
        $Y_j$ variable was on it's $r_j^{th}$ iteration for all $j \in \ev{2}{l}$. Since $X_{j-1}^0 = V$, the iteration values of $X$ variables can be safely ignored. 
    
    \begin{lemma}\label{app-obs:v-Even-Odd-inequalities} 
        $\quad \text{if } v \in V_\Even, \quad \quad \forall(v, w)\in E, \rank{v}\geq_{l+1-\chi(v)} \rank{w}$
        $$\text{if } v \in V_\Odd, \quad \quad \exists(v, w)\in E, \rank{v}\geq_{l+1-\chi(v)} \rank{w}$$
       % \begin{align*}
       %     \quad\quad\quad\quad\quad\quad\text{if } v \in V_\Even, \quad \quad \forall(v, w)\in E, \rank{v}\geq_{l+1-\chi(v)} \rank{w}\\
       %     \quad\quad\quad\quad\quad\quad\text{if } v \in V_\Odd, \quad \quad \exists(v, w)\in E, \rank{v}\geq_{l+1-\chi(v)} \rank{w}
       % \end{align*}
        where $\rank{v} \geq_b \rank{w}$ denotes the $\geq$ relation in the lexicographic ordering, restricted to the first b elements of the tuple. If $\chi(v)$ is even, the inequalities are strict. 
    \end{lemma}
    \begin{proof}
        Consider a $v$ with $\chi(v) \in \{m-1, m\}$ for some even $m$ and let $\rank{v} = (r_l, 0, \ldots, r_2, 0)$.
        By Lem.~\ref{app-lem:intersection_of_Y_m}, $v \in Y_m^{r_l-1, 0, \ldots, r_{m-2}-1, 0, r_m}$. If we look at the flattening of this formula in Lem.\ref{app-obs:flat-Z}, $v$ is in particular, inside the middle term of this formula. That is,
         \\$v \in \B_m[Y_m^{r_l-1, \hdots, r_m-1}, Y_m^{r_l-1,\hdots, r_m }]$. If we go through the definition of this term we get,
            $$(\bigcup_{i \in [m+1, l]} C_i) \cup (\overline{C_m} \cap \Npre(Y_m^{r_l-1, \hdots, r_m-1}, Y_{m}^{r_l-1,0, \hdots, r_m})) \cup (C_{m} \cap \Cpre_\Odd(Y_m^{r_l-1, 0, \hdots, r_m-1}))$$
        
            \vspace{-7mm}
            \begin{align*}
               \text{ That gives us, } \quad \quad &\text{if } \chi(v) = m, \quad \quad \quad\quad\quad \,\,\, v \in \Cpre_\Odd(Y_m^{r_l-1, 0, \ldots, r_m-1}) \\
                &\text{if } \chi(v) = m-1, \quad \quad \quad\quad v \in \Npre(Y_m^{r_l-1, 0, \ldots, r_m-1}, Y_m^{r_l-1, 0, \hdots, r_m}) \\
        \end{align*}

        By the definition of $\Npre$ we get, if $\chi(v) = m-1$ then $v \in \Cpre_\Odd(Y_m^{r_l-1, 0, \hdots, r_m})$.
       Since odd indices get $0$-ranks, the claim of the lemma follows from the definition of $\Cpre_\Odd$ together with the observation $\rank{v} \geq_{l+1-m} \rank{w} \Leftrightarrow \rank{v} \geq_{l+1-(m-1)} \rank{w}$.
            %from the fact that, if $w \in Y_m^{r_l-1, 0, \hdots, r_m}$, then then first $l+1-m$ of $\rank{w}$ is less than or equal to that of $\rank{v}$. Since odd indices are always $0$, $\rank{v} \geq_{l+1-m} \rank{w} \iff \rank{v} \geq_{l+1-(m-1)} \rank{w}$. 
    \end{proof}

    Now we are ready to introduce the first of our three main propositions:
    
    \begin{proposition}\label{app-prop:Mexists}
        If $\Wo \neq \emptyset$, there exists a non empty set $M := \{ v \in \Wo \mid \rank{v} = (1, 0, 1, 0, \ldots, 1, 0)\}$. Furthermore, for all $v\in M$, $\chi(v)$ is odd.
    \end{proposition}
    Observe that $(1,0,1,0,\ldots, 1,0)$ is the smallest rank possible. Therefore, $v\in M$ are the vertices that were added to $Z$ in ~\eqref{app-eq:fp-odd} in the first iteration of the fixed-point calculation and were never removed.
    The first part of the proposition follows from the monotonicity of fixed-point calculation. That is, if $M$ was empty $Z$ would be empty as well.
    
    For the second part, observe that in the first iteration of the formula, for all $j$, $Y_j = \emptyset$. Also, $\Cpre_\Odd(\emptyset) = \emptyset$. 
    Then from~\eqref{app-eq:fp-odd}, $Z$ does not contain any $v$ with even priority. 
    
    \begin{proposition}\label{app-prop:cycle-through-M}
        All cycles in $\Sc^{\mathcal{G}^\ell}$ that pass through a vertex in $M$ are \Odd winning.% (i.e. the largest priority in the cycle is odd).
    \end{proposition}

    To see why Prop.~\ref{app-prop:cycle-through-M} holds, we make an observation.
    For an even $m\leq l$, let $Y_m^\1$ denote the value of $Y_m$ after the first ever iteration over it is completed, during the computation of ~\ref{eq:fp-odd}.
    I.e. $Y_m^\1 = Y^{0,0,\ldots ,0, 1} $.
    %$$\nu X_{m-1}\ldots \mu Y_2 \nu X_1. \bigcap_{j \in \ev{m+2}{l}} \B_j[\emptyset, V] \cap \B_m[\emptyset, X_{m-1}] \cap \bigcap_{j \in \ev{2}{m-2}} \B_j[Y_j, X_{j-1}]$$
    %In the first term $\B_j$ takes $\emptyset$ and $V$ as arguments. This is due to all $Y_{j}, X_{j-1}$ variables for $j \leq m$ having the values they are initialized with. Observe that when the fixed-point above is calculated, all $X_{j-1}, Y_j$ values for $j < m$ will saturate at the same value,
    %which is the final result of the computation. That is, 
    Since for all $j$, $Y_j^0 = \emptyset$ and $ X^0_{j-1} = V$, Lem.~\ref{app-obs:flat-Z} gives,
    \begin{equation}\label{eq:Ym1}
    Y_m^\1 = \bigcap_{j \in \ev{m+2}{l}} \B_j[\emptyset, V] \cap \B_m[\emptyset, Y_m^\1] \cap \bigcap_{j \in \ev{2}{m-2}} \B_j[Y_m^\1 Y_m^\1]
    \end{equation}
    If we go through the definition of $\B_j$ we see that: the first term of this formula adds or deletes $v \in C_j$ with $j > m$. It adds all the ones with odd $j$  and removes all the ones with even $j$.%is equal to $\bigcup_{j \in \ev{m+2}{l}}C_{j-1} \cup \bigcup_{j \in [1, m+1]}C_j$. That is, the first term eliminates all $v \in C_j$ with even $j>m$ from $Y_m^\1$ and add all $C_j$ with odd $j>m$.
     The last term adds and removes $v \in C_j$ for $j \leq m-2$. It adds the ones in $\Cpre_\Odd(Y_m^\1)$ and removes the ones that are not. The middle term eliminates $C_m$ and all $v \in C_j \cap \neg \Npre(\emptyset, Y_m^\1)$ for $j < m$, and adds $v \in C_{m-1} \cap \Npre(\emptyset, Y_m^\1)$.
    If we go through the definition of $\Npre$, we see that $\Npre(\emptyset, Y_m^\1) = \Cpre_\Odd(Y_m^\1) \cap (V_\Even \cup \Lpre^\forall(Y_m^\1))$.
    This gives,
    \begin{equation}\label{app-eq:obs}
         v \in Y_m^\1 \iff \chi(v)>m\text{ and is odd, or } \chi(v)< m \text{ and } v\in \Npre(\emptyset, Y_m^\1)
    \end{equation}
    %$ Y_m^\1$ consists of $v$ with either odd $\chi(v)>m$, or in $\Npre(\emptyset, Y_m^\1)$.

    Then for all $v \in M$, $v \in Y_m^\1$ for each even $m \leq l$. In particular, $ v \in Y_n^\1$ where $n$ is such that $\chi(v) = n-1$.
    It follows that $ v \in \B_n[\emptyset, Y_n^\1]$. Then, $v \in \Cpre_\Odd(Y_n^\1) \cap (V_\Even \cup \Lpre^\forall(Y_n^\1))$.
    Since all live outgoing edges of $v$ are in $Y_n^\1$, for all $(v,w)$ in $\Sc^{\mathcal{G}^\ell}$, $w \in Y_n^\1$.

    By our previous observation $w$ either has an odd priority larger than $n$, or is in $ \Cpre_\Odd(Y_n^\1) \cap (V_\Even \cup \Lpre^\forall(Y_n^\1))$.
    If $\chi(w)>n$ is odd, then $w \in Y^\1_{\chi(w)+1}$, and we repeat the same argument to conclude the highest priority seen is always odd.

    %Before we present the third proposition, we need a lemma and some observations obtained from formula ~\eqref{eq:fp-odd}.
    %\begin{lemma}\label{obs:v-Even-Odd-inequalities} 
    %%    \begin{align*}OBS
     %       &\quad \quad \quad \quad  \quad \quad \text{if } v \in V_\Even, \quad \quad \forall(v, w)\in E, \rank{v}\geq_{l+1-\chi(v)} \rank{w}\\
     %       &\quad \quad \quad  \quad \quad \quad  \text{if } v \in V_\Odd, \quad \quad \exists(v, w)\in E, \rank{v}\geq_{l+1-\chi(v)} \rank{w}
     %   \end{align*}
     %   where $\rank{v} \geq_b \rank{w}$ denotes the $\geq$ relation in the lexicographic ordering, restricted to the first b elements of the tuples $\rank{v}$ and $\rank{w}$. If $\chi(v)$ is odd, the inequalities are strict. 
    %\end{lemma}
    \begin{definition}
        We call a play $\pi = v_1 v_2 \ldots$ in $\Sc^{\mathcal{G}^\ell}$ \emph{minimal} if for all $v_i \in V_\Odd$, $v_{i+1}$ is the minimum ranked successor of $v_i$. A minimal cycle is a section of a minimal play.
    \end{definition}
    \begin{lemma}\label{app-lem:minimalplayOddwinning}
        Every minimal play is \Odd winning.
    \end{lemma}
    A minimal play only sees minimal cycles. Let $\delta = w_1 w_2 \ldots w_1$ be such a cycle. 
    $\delta$ cannot be an \Even winning cycle: Assume $b := \max\{ \chi(w) \mid w \in \delta\} $ is even. Let $w_i\in \delta$ have priority $b$. By Obs. \ref{app-obs:v-Even-Odd-inequalities}, $\rank{w_i} >_{l+1-b} \rank{w_{i+1}} \geq_{l+1 - \chi(w_{i+1})} \ldots \geq_{l+1-\chi(w_{i-1})} \rank{w_i}$. Since for all $w_j \in \delta$, $\chi(w_{j})\leq b$, the inequality yields $\rank{w_i} >_{l+1-b} \rank{w_i}$, which is a contradiction.

    %The last proposition states that all $\pi$ that starts in \Wo and is compliant with $\Sc^{\mathcal{G}^\ell}$, visits $M$ infinitely often. 
    \begin{proposition}\label{app-prop:minimal-play-visits-M}
        Any minimal play compliant with $\Sc^{\mathcal{G}^\ell}$ visits $M$ infinitely often.
    \end{proposition}
    %Any minimal play sees minimal cycles infinitely often. 
    Let $\delta = w_1 w_2 \ldots w_1$ be a minimal cycle and $w_k$ its vertex with maximum priority. We will show that $w_k \in M$. Since $\pi = \delta \delta \ldots$ is a minimal play, by Lemma.~\ref{app-lem:minimalplayOddwinning} we know $\chi(w_k)$ is odd. Furthermore, we have observed in \ref{app-eq:obs} that $w_k \in Y_m^\1$ for all $m > \chi(w_k)$. 
    If we can show that $w_k \in Y_m^\1$ also for $m < \chi(w_k)$, then we have $w_k \in M$. We will now show this. 
    
    Assume to the contrary that $w_k \not \in M$ and let $j$ be the largest non-trivial index of $\rank{w_k}$. 
    That is $j < l$ is the largest even integer such that $w_k \not \in Y_j^\1$. Let $t$ be the value of this index, i.e. $w_k \in Y_j^{0,\ldots, 0,t} \setminus Y_j^{0,\ldots, 0,{t-1}}$. 
    Let us denote $Y_j^{0, \ldots, 0, t}$ by $Y_j^\te$ for short. 

    Since $\delta$ is minimal, Lem.~\ref{app-obs:v-Even-Odd-inequalities} gives $\rank{w_i} \geq_{l+1 - \chi(w_i)} \rank{w_{i+1}}$ for all $w_i \in \delta$. Since $\chi(w_i) \leq \chi(w_k)$ for all $i$ and $\chi(w_k) < j$; $\rank{w_i} \geq_{l+1-j} \rank{w_{i+1}}$ for all $w_i \in \delta$. 
     This implies $\rank{w} =_{l+1-j} \rank{w'}$ for all $w, w'\in \delta$. It follows that for all $w\in \delta$, $w \in Y_j^\te \setminus Y_j^{\mathbf{t-1}}$.
 
     %We can follow the same steps in equation~\eqref{eq:Ym1} to observe that 
   %$\forall (v,w)$ in $\Sc^{\mathcal{G}^\ell}$, $w \in \Npre(Y^{\mathbf{t-1}}_j, Y^\te_j) = \Cpre_\Odd(Y^\te_j) \cap (V_\Even \cup \Lpre^\forall(Y^\te_j) \cup Pre^\exists_\Odd(Y^\mathbf{t-1}_j))$.
   %If $w \in Pre^\exists_\Odd(Y_j^{\mathbf{t-1}})$, since $\delta$ is a minimal cycle, $\delta$ will have an element from $Y_j^{\mathbf{t-1}}$. However, this contradicts our observation that $\delta$ lies in $Y_j^\te \setminus Y_j^{\mathbf{t-1}}$.
   %On the other hand if non of the $w \in \delta $ lie in $ Pre^\exists_\Odd(Y_j^{\mathbf{t-1}})$ this implies that they all get into the formula due to reaching other nodes in $Y_j^\te \setminus Y_j^\mathbf{t-1}$. This is not possible since a node in $\delta$ has to be added to $Y_j^\te \setminus Y_j^\mathbf{t-1}$ as the first node and thus, have to have a successor in $Y_j^\mathbf{t-1}$. % while $Y_j^\te \setminus Y_j^{\mathbf{t-1}}$ is empty.
   %  Therefore, $w_k\in M$.

     Once more by Lem.~\ref{app-obs:flat-Z} we get that for all $w \in \delta$, 
     $$w \in \B_j[Y_j^{\mathbf{t-1}}, Y_j^\te] = (\bigcup_{i\in[j+1, l]} C_i) \cup ( \overline{C_j} \cap \Npre(Y_j^{\mathbf{t-1}}, Y_j^\te) \cup (C_j \cap \Cpre_\Odd(Y_j^{\mathbf{t-1}})))$$
     Since $\chi(w) < j$, this implies $$w \in \Npre(Y_j^{\mathbf{t-1}}, Y_j^\te) = \Cpre_\Odd(Y_j^\te) \cap (V_\Even \cup \Lpre^\forall(Y_j^\te) \cap \Pre^\exists_\Odd(Y_j^{\mathbf{t-1}}) )$$
     %Since we assumed all $w \in Y_j^\te \setminus Y_j^{\mathbf{t-1}}$, $w$ cannot be in the last term.
     Now consider the set $Y_j^\te \setminus Y_j^{\mathbf{t-1}}$, which is initially empty. Then the first term in $\delta$ that gets in $Y_j^\te \setminus Y_j^{\mathbf{t-1}}$ has to be in $\Pre^\exists_\Odd(Y_j^{\mathbf{t-1}})$. 
     This contradicts our assumption that all $w_i \in  Y_j^\te \setminus Y_j^{\mathbf{t-1}}$ and proves that $w_k \in M$.
     We are now ready to prove the main theorem.
     \begin{proof}[Proof of Thm. \ref{prop:mainresult}]
          Let $\pi= v_0v_1\ldots$ be a play compliant with $\Sc^{\mathcal{G}^\ell}$ with $v_0 \in \Wo$. %Since all plays compliant with \Odd strategy templates obey the fairness condition, we only need to show that the maximum priority seen infinitely  in $\pi$ is odd.
          Since $\pi$ is compliant with an \Odd strategy template, it is a fair play. 
          For a node $v \in \Wo$, let $v_{\min}$ be the minimum ranked successor of $v$.
          Since $\pi$ is fair, for all $v$ that is visited infinitely often in $\pi$, $v_{\min}$ is visited infinitely often as well. 
          This gives us an infinite subsequence of $\pi$ that is minimal. Since all minimal plays visit $M$ infinitely often (Prop.~\ref{app-prop:minimal-play-visits-M}), 
          $\pi$ visists $M$ infinitely often. Then there must exist an $x \in M$ that $\pi$ visits infinitely often. 
          Then a tail of $\pi$ is consisted of consecutive cycles over $x$. Since all cycles that pass through $M$ are \Odd winning (Prop.~\ref{app-prop:cycle-through-M}), $\pi$ is \Odd-winning.
   \end{proof}

\subsection{Zielonka's Algorithm for \Odd-Fair Parity Games}\label{app:zielonka-proof}

% This section is intended to provide a detailed proof of Thm.~\ref{thm:solvebb}.
% 
% We refer the reader to Sec.~\ref{sec:zielonka} for the introduction and motivation of the proofs. We will not repeat the definitions of safe reachability sets and partial strategy templates in this section. Once more we refer the reader to Sec.~\ref{sec:zielonka} for these preliminaries. 
% That said,  we will not follow the same lay-out on the proofs as in Sec.~\ref{sec:zielonka}. Moreover, we will present new definitions and lemmas liberally. So this section should be perceived as \emph{somehow} stand-alone, with the exception of the aforementioned dependencies. 
% The proof roughly follows the foot steps of ~\cite{Kuesters2002}.% We will cite some lemmas from this work and leave them unproven if the proof is unaffected by the fairness condition. 

This section provides a detailed proof of Thm.~\ref{thm:solvebb}.
However, we will not follow the lay-out given for this proof in Sec.~\ref{sec:zielonka} but rather follow the foot steps of the correctness proof of the \enquote{normal} Zielonka's algorithm from  ~\cite{Kuesters2002}. Hence, this section should be perceived as stand-alone, with the exception of the definitions of safe reachability sets and partial strategy templates, which can be found in Sec.~\ref{sec:zielonka}. While we do not follow the same lay-out, the motivation and intuition given for the proof in Sec.~\ref{sec:zielonka} still carries over to this section.

% \subsubsection{Preliminaries} 
% 
%We start by extending the strategy template definition given in section 4 to \Even strategies (just \Even and \Odd are swapped in the definition). Note that, since \Even nodes do not have live outgoing edges, for all $v \in V'_\Odd$, $E(v) \subseteq E'$ and for all $v \in V'_\Even$, $|E'(v)| = 1$. Consequently, a positional \Even strategy $\rho$ is equal to the unique strategy compliant with the strategy template $(V', V'_\Even, V'_\Odd, E')$ for which $E'(v) = \rho(v)$ for all $v \in V'_\Even$. We use a positional \Even strategy $\rho$ to define a strategy template for \Even.
%To distinguish strategy templates defining \Odd strategies from the ones defining \Even strategies, we denote \Odd strategy templates by $\mathcal{S}^\Odd$ and \Even ones with $\mathcal{S}^\Even$. 
% 
%Furthermore, we obtain a partial strategy template by taking a strategy template and a subset $V'$ of its vertex set, and removing all the outgoing edges of all vertices in $V'$. \IS{?}
%Let $a \in \{\Even, \Odd\}$ and $\mathcal{S}^a = (V', V'_\Even, V'_\Odd, E')$ be an $a-$strategy template. We obtain a partial strategy template as follows: Let $V'' \subseteq V'$ be such that, for all $v \in V''$, $v$ does not lay on a cycle in $E'$. Now for all $v \in V''$, remove all outgoing edges of $v$ from $E'$ to obtain $E''$. Then $(V', V'_\Even, V'_\Odd, E'')$ is a partial strategy template for $a$. We denote a partial strategy template of $a$ by $\mathcal{P}^a$.

\subsubsection{Preliminaries} 
% This section proves important observations w.r.t.\ the \bb paradises in \odd-fair parity games. In order to do so, we restate the definitions of an $\bb-$trap (Def.~\ref{def:atrap}) and a  \bb paradise (Def.~\ref{}) from Sec.~\ref{sec:zielonka}. We also estate three lemmas (Lem.~\ref{app-lem:Kuesters6.2}, \ref{app-lem:Kuesters6.3} and \ref{app-lem:Kuesters6.4}) exactly as they appear in \cite{Kuesters2002} and omit the proofs since the statements of these lemmas are only concerned with the properties of the subsets of $V$, and are therefore unaffected by the fairness condition.
We emphasize again that we assume the underlying game graph of the fair parity game $\mathcal{G}^\ell$ to be deadend-free.

\smallskip
\noindent\textbf{Subgames.}
For some $U \subseteq V$ we denote by $E \mid_U = \{(v, w) \in E \mid v, w \in U \}$ and by $\chi\mid_U$ we denote the restriction of the function $\chi$ to the domain $U$. %With this, we formally define subgames as follows.

%We call a $\mathcal{P} = (V', E')$ a \emph{partial strategy template in $\mathcal{G}^\ell$} if it is a strategy template, with the exception of some dead-ends. That is, $\mathcal{P}$ is a partial \Odd (\Even) strategy template if for each $v \in V'$, either $E'(v) = \emptyset$ or $v$ satisfies~\ref{item:Oddstrtemprules} (\ref{item:Evenstrtemprules}).

\begin{definition}[Subgames]
    Let $U \subseteq V$. The subgraph of $\mathcal{G}^\ell$ induced by $U$ is shown as $\mathcal{G}^\ell[U]$ and is the restriction of the game graph to $U$, i.e.
    $\mathcal{G}^\ell[U] = \ltup{ \langle U, \Ve \cap U, \Vo \cap U, E|_U, \chi|_U\rangle, E^\ell|_U} $. $\mathcal{G}^\ell[U]$ is a subgame of $\mathcal{G}^\ell$ if and only if $\mathcal{G}^\ell[U]$ is deadend-free.
\end{definition}

\begin{lemma}[\cite{Kuesters2002}, Lemma 6.2]\label{app-lem:Kuesters6.2}
    If $U, U' \subseteq V$ where $\mathcal{G}^\ell[U]$ is a subgame of $\mathcal{G}^\ell$ and $(\mathcal{G}^\ell[U])[U']$ is a subgame of $\mathcal{G}^\ell[U]$, then $\mathcal{G}^\ell[U']$ is a subgame of $\mathcal{G}^\ell$.
\end{lemma}

The above lemma (as well as the following two lemmas  \ref{app-lem:Kuesters6.3} and \ref{app-lem:Kuesters6.4}) are restated exactly as they appear in \cite{Kuesters2002}. We omit their proofs since the statements of these lemmas are only concerned with the properties of the subsets of $V$, and are therefore unaffected by the fairness condition.

\smallskip
\noindent\textbf{\bb-Trap.} We restate the definition of a $\bb$-trap from Sec.~\ref{sec:zielonka}. and subsequently show important observations w.r.t.\ $\bb$-traps in \Odd-fair parity games.
 
\begin{definition}[\bb-trap]\label{def:atrap}
A $\bb$-trap is a subset $T \subseteq V$ for $\bb \in \{\Even, \Odd\}$ such that,
\begin{align*} &\forall v \in T \cap V_{\nb}, \quad \exists (v, w)\in E \text{ with } w \in T, \\
   & \forall v \in T \cap V_{\bb}, \quad \, (v, w) \in E \implies w \in T.
 \end{align*}
\end{definition}

%So any strategy $\rho$ that is compliant with a strategy template $\mathcal{S}^\Odd$, for which 

% Now we will give some lemmas, the proofs of which can be found in \cite{Kuesters2002} and the fairness condition do not change the proofs. For the following, let $\mathcal{G}^\ell = \tup{\mathcal{G}, E^\ell}$ be an \Odd-Fair Parity game. 

\begin{lemma}[\cite{Kuesters2002} Lemma 6.3]\label{app-lem:Kuesters6.3}
    \begin{enumerate}
        \item For every $\bb$-trap $U$ in $\mathcal{G}^\ell$, $\mathcal{G}^\ell[U]$ is a subgame.
        \item If $X$ is a $\bb$-trap in $\mathcal{G}^\ell$ and $Y\subseteq X$ is a $\bb$-trap in  $\mathcal{G}^\ell[X]$, then $Y$ is a $\bb$-trap in $\mathcal{G}^\ell$.
    \end{enumerate}
\end{lemma}

%\noindent For the following, let $U \subseteq V$  such that $\mathcal{G}^\ell[U]$  is a subgame of $\mathcal{G}^\ell$.
\begin{lemma}[\cite{Kuesters2002}, Lemma 6.4 -- Sec.~\ref{sec:zielonka:correct} Obs.~\ref{it:obs5}]\label{app-lem:Kuesters6.4}
 The set $U \setminus \SafeReach^f_\bb(U, R, \mathcal{G}^\ell)$ is a $\bb$-trap in $U$.
\end{lemma}
\begin{lemma}\label{app-lem:safereacheven-noliveedges} 
    Let $W = U \setminus \SafeReach^f_\Even(U, R, \mathcal{G}^\ell)$. There exists no $(v,w) \in E^\ell$ with $v \in W$ and $w \in \SafeReach^f_\Even(U, R, \mathcal{G}^\ell)$.
\end{lemma}
\begin{proof}
A node $v \in U \setminus \SafeReach^f_\bb(U, R, \mathcal{G}^\ell) \cap V_\bb$ cannot have an edge that leads to $\SafeReach^f_\bb(U, R, \mathcal{G}^\ell)$, since then $v$ itself must be in this set.
Similarly a node  $v \in U \setminus \SafeReach^f_\bb(U, R, \mathcal{G}^\ell) \cap V_{\nb}$ must have an edge that leads to $ U \setminus \SafeReach^f_\bb(U, R, \mathcal{G}^\ell)$, or else $v$ would be in  $\SafeReach^f_\bb(U, R, \mathcal{G}^\ell)$.
\end{proof}

\begin{lemma}\label{app-lem:SafeReachOdd_of_an_even_trap_is_an_even_trap}
    If $R$ is an \Even-trap in $U$, then so is $\SafeReach^f_{\Odd}(U, R, \mathcal{G}^\ell)$.
\end{lemma}
\begin{proof}
    This is easy to observe from the definition of a partial strategy template $sr_\Odd$ on $\SafeReach^f_\Odd(U, R, \mathcal{G}^\ell)$.
    All $(v, w) \in E$ with $v \in V_\Even \cap \SafeReach^f_\Odd(U, R, \mathcal{G}^\ell)\setminus R$, are in $sr_\Odd$. That is, $w \in \SafeReach^f_\Odd(U, R, \mathcal{G}^\ell)$. For all $v \in V_\Even \cap R$, all $(v,w) \in E\subseteq U \times U$ are in $R$ since $R$ is an \Even-trap in $U$. 
    Thus for all \Even nodes in $\SafeReach^f_\Odd(U, R, \mathcal{G}^\ell)$, all their successors in $U$ are in the set again. 
    We can similarly observe that for all $v \in V_\Odd \cap \SafeReach^f_\Odd(U, R, \mathcal{G}^\ell)$ they have at least one successor in the set. 
    Thus this set is an \Even-trap in $U$. 
\end{proof}

\smallskip
\noindent\textbf{\bb-Paradise.} We restate the definition of a $\bb$-paradise from Sec.~\ref{sec:zielonka} and subsequently show important observations w.r.t.\ $\bb-$paradises in \Odd-fair parity games.

\begin{definition}[$\bb$-paradise]\
A $\bb$-paradise of an \Odd-fair parity game $\mathcal{G}^\ell$ is a region $P \subseteq V$ from which player $\nb$ cannot escape (i.e. $P$ is a $\nb$-trap) and player $\bb$ has a strategy to win from all $v\in P$. As we have proven in section 5, this implies that there exists a strategy template $\mathcal{S}^\bb$ with the vertex set $P$ such that all player $\bb$ strategies compliant with $\mathcal{S}^\bb$ are winning for player $\bb$.

Formally $P \subseteq V$  is a $\bb$-paradise if:
\begin{itemize}
\item $P$ is a $\neg \bb$-trap and, 
\item There exists a winning \bb strategy template $\mathcal{S}^\bb = \ltup{P, E'}$ on $\mathcal{G}^\ell$.
\end{itemize}
\end{definition}
Note that if $P$ is a $\bb$-paradise, and play $\pi$ starting in $P$ and is compliant with $\mathcal{S}^a$, stays in $P$ and is won by \bb.

\begin{lemma}[Sec.~\ref{sec:zielonka:correct} Obs.~\ref{it:obs4}]\label{app-lem:safe-reach-Odd-paradise}
    If $R\subseteq V$ is an \Odd-paradise in $\mathcal{G}^\ell$, then $\SafeReach_\Odd^f(V, R, \mathcal{G}^\ell)$ is also an \Odd-paradise in $\mathcal{G}^\ell$.
\end{lemma}
\begin{proof}
Due to Lem.~\ref{app-lem:SafeReachOdd_of_an_even_trap_is_an_even_trap}, $\SafeReach_\Odd^f(V, R, \mathcal{G}^\ell)$ is an \Even-trap in $V$.
The winning \Odd strategy template on it is just a combination of the winning \Odd strategy template $\mathcal{S}$ on $R$ and the partial \Odd strategy template $sr_\Odd$ on $\SafeReach^f_\Odd(V, R, \mathcal{G}^\ell)$, on which nodes in $R$ are dead-ends and  
all $v \in \SafeReach^f_\Odd(V, R, \mathcal{G}^\ell) \setminus R$ are guaranteed to reach $R$ in finitely many steps. %Remember that $sr_\Odd$ is acyclic. 
Let $\e$ be the combination of edges in $sr_\Odd$ and $\mathcal{S}$. 
Since $R$ is an \Even-trap in $V$, all outgoing edges of \Even nodes in $R$ stay in $R$. All outgoing edges of \Even nodes in 
$\SafeReach^f_\Odd(V, R, \mathcal{G}^\ell) \setminus R$ are in $sr_\Odd$. Therefore all outgoing edges of \Even nodes in $\SafeReach^f_\Odd(V, R, \mathcal{G}^\ell)$ are in $\e$.
It's easy to see $\e$ introduces no new cycles to $sr_\Odd \cup \mathcal{S}$. Therefore $\mathcal{S}' = 
(\SafeReach^f_\Odd(V, R, \mathcal{G}^\ell), \e)$ is an \Odd strategy template in $\mathcal{G}^\ell$.
$\mathcal{S}'$ is winning because any play starting in $\SafeReach^f_\Odd(V, R, \mathcal{G}^\ell)\setminus R$ reaches $R$ in finitely many steps and from there on stays in $R$. 
Since from that point on $\mathcal{S}'$ collapses to $\mathcal{S}$, the game is won by \Odd. 

\end{proof}
\begin{corollary}\label{app-cor:determinacy}
    For an \Odd-fair parity game $\mathcal{G}^\ell$, $V$ is partitoned into an \Even-paradise and an \Odd-paradise. 
\end{corollary}
The corollary follows from the fixed-point equations~\eqref{eq:fp-even} and~\eqref{eq:fp-odd}. Winning region of player $\bb$ is by definition a $\bb$-paradise. \We is the \Even-paradise with the strategy template defined by the positional strategy acquired from the fixed-point formula in~\eqref{eq:fp-even}. The calculation of the positional strategy is closely related to the ranking function and strategy template computation in Sec.~\ref{sec:strat-templates}, and a brief introduction of the calculation can be found in \cite{banerjee2022fast}.
$\Wo = V \setminus \We$ is the \Odd-paradise. The calculation of the strategy template for \Odd is given in Section 5.

\subsubsection{Computing Winning Regions $\mathcal{W}_\bb$}
Now we will give a construction to calculate $\Wo$ and $\We$ in $\mathcal{G}^\ell$. The construction corresponds to the \Odd-fair Zielonka's algorithm given in Alg.~\ref{algo:fair-zielonka-bb}.
We will give the construction in two parts. First we will take an \Odd-fair parity game $\mathcal{G}^\ell$ and an \emph{odd} integer $n$ where $n$ is an upper bound on the priorities seen in the vertex set of $\mathcal{G}^\ell$. Then we will show how to obtain $\Wo$ and $\We$ in $\mathcal{G}^\ell$ in the existence of a procedure
that can do the same on a subgame $\mathcal{G}^\ell[X]$ of $\mathcal{G}^\ell$ where $n-1$ is an upper bound of the priorities seen in $\mathcal{G}^\ell[X]$. 
In the second part we will show the same for $\mathcal{G}^\ell$ with an \emph{even} $n$. The combination of these two procedures with a base case, will give the recursive algorithm we need to solve \Odd-fair parity games. 
We will count on strategy templates in the proof of both parts. However, the second part of the algorithm follows roughly the same principles in Zielonka's original algorithm, whereas the
 the first part requires an essential change in reasoning, due to the adoption of $\SafeReach^f_\Even$. Even though the reasoning required to prove the first part is fairly different than Zielonka's original algorithm, 
a computationally cheap addition to the original algorithm is sufficient to get the correct computation for the \Odd-fair variant. Surprisingly, the trick is cheap enough not to alter the complexity of the original algorithm at all!

\smallskip
\noindent\textbf{Subsets and Sequences.}
Let $n$ be an upper bound on the priorities seen in $V$. If $n$ is \Even, set $\bb:=\Even$, otherwise $\bb:=\Odd$.
Further, we construct a decreasing series of subsets of $V$, $\{X_\bb^i\}_{i\in \mathbb{N}}$
by assigning the following sets (see Fig.~\ref{fig:kuesters-figure-extended} for an illustration): % V =: &X_\Odd^0 \supsetneq X_\Odd^1 \supsetneq \ldots \supsetneq  X_\Odd^k =  X_\Odd^{k+1}
\vspace{0.3cm}

\noindent Initially set $X^0_\nb = \emptyset$. For all $i \in \mathbb{N}$, set 
\begin{subequations}\label{equ:seriesZielonka}
    \begin{align*}
   &X_\bb^i := V \setminus X_\nb^i \quad \quad \quad &N^i:= \{v \in X^i_\bb \mid \chi(v) = n\}\\
       &Z^i:= X^i_\bb \setminus \SafeReach^f_\bb(X^i_\bb, N^i, \mathcal{G}^\ell) \quad &X^{i+1}_\nb :=  \SafeReach^f_\nb(V, X_\nb^{i} \cup Z_\nb^{i}, \mathcal{G}^\ell) % X_\Even^{i} \cup \SafeReach_\Even^f(X^{i}_\Odd, Z_\Even^{i}, \mathcal{G}^\ell) )%\text{\todo{IS: I know the equality is not completely justified. The first one is cheaper for an algorithm pov, whereas the second one is easier to justify that $X^i_\Odd$ is an \Even-trap.}} 
   \end{align*}
   \end{subequations}
where $Z_\nb^i$ is the \nb winning region in the subgame $\mathcal{G}^\ell[Z^i]$, assuming it is a subgame. 
First let's show that these sets are well-defined.
\begin{lemma}
The sets $X_\bb^i, X_\nb^i, N^i, Z^i, Z_\nb^i$ and $Z_\bb^i$ are well defined for all $i \in \mathbb{N}$.
\end{lemma}
\begin{proof} We will prove this by induction. For the base case $i = 0$, 
$X_\bb^0 = V$ is trivially an \nb-trap in $V$ and $\mathcal{G}^\ell[X^0_\bb]$ is trivially a subgame of $\mathcal{G}^\ell$. 
By Lem.~\ref{app-lem:Kuesters6.4}, $Z^0$ is an \bb-trap in $X^0_\bb$, and thus by Lem.~\ref{app-lem:Kuesters6.3}-1, $\mathcal{G}^\ell[Z^0]$ is a subgame of $\mathcal{G}^\ell$. 
Due to Corollary~\ref{app-cor:determinacy}, we know $\mathcal{G}^\ell[Z^0]$ is divided into an \bb-paradise and \nb-paradise. Therefore,  
$Z^0_\bb$ and $Z^0_\nb$ are also well-defined.  

By induction on $i$, we get by Lem.~\ref{app-lem:Kuesters6.4} that $X^i_\bb$ is an \nb-trap in $V$, and by Lem.~\ref{app-lem:Kuesters6.3}-1 $\mathcal{G}^\ell[X_\bb^i]$ is a subgame of $\mathcal{G}^\ell$. $Z^i$ is an \bb-trap in $\mathcal{G}^\ell[X^i_\bb]$, and thus by Lem.~\ref{app-lem:Kuesters6.2}, $\mathcal{G}^\ell[Z^i]$ is a subgame in $\mathcal{G}^\ell$.
Therefore $Z_\nb^i$ and $Z_\bb^i$ are well-defined.
\end{proof}
We also derived the following observations from the proof:

\begin{observation}[Sec.~\ref{sec:zielonka:correct} Obs.~\ref{it:obs1}]\label{app-obs:traps-subgames}
    $X^i_\nb$ is an \bb-trap, $X^i_\bb$, $Z^i$ and $Z_\bb^i$ are \nb-traps in $V$. $Z^i$ is in \nb-trap in $X_\bb$ and $Z_\nb^i, Z_\bb^i$ are \bb and \nb traps in $Z^i$, respectively.
    Therefore by Lem.~\ref{app-lem:Kuesters6.2}, $\mathcal{G}^\ell[Y]$ is a subgame of $\mathcal{G}^\ell$ with $Y$ being any of these sets. 
\end{observation}

\begin{lemma}[Sec.~\ref{sec:zielonka:correct} Obs.~\ref{it:obs2}]\label{app-lem:X_nb-equivalence}
    $X_\nb^{i} \cup \SafeReach_\nb^f(X^{i}_\bb, Z_\nb^{i}, \mathcal{G}^\ell) =  \SafeReach_\nb^f(V, X_\nb^{i} \cup Z_\nb^{i}, \mathcal{G}^\ell) $
\end{lemma}
\begin{proof}
   $ \mathbf{(\subseteq )}$ Trivially, $X_\nb^{i} \subseteq \SafeReach_\nb^f(V, X_\nb^{i} \cup Z_\nb^{i}, \mathcal{G}^\ell)$.
    Similarly a \\$v \in  \SafeReach_\nb^f(X^{i}_\bb, Z_\nb^{i}, \mathcal{G}^\ell)$, can be made by $\nb$ to reach $Z_\nb^i$ while staying in $X_\bb^i$. Then $v$ is trivially in the righthand side equation as well.
    
    \noindent $ \mathbf{(\supseteq )}$ 
    Let $v \in \SafeReach_\nb^f(V, X_\nb^{i} \cup Z_\nb^{i}, \mathcal{G}^\ell) \setminus X_\nb^i$.%, be in $X^j \setminus X^{j-1}$ where $X^j$ is the value of the $X$ variable after the $j^{th}$ iteration of the fixed-point computation from formula~\eqref{eq:\SafeReachEven}. that is $\rank{v} = j$ and $v \in X_\Even^{i} \cup Z_\Even^{i} \cup \Cpre_\Even(X^{j-1}) \cup \Lpre^\exists(X^{j-1})$.
     Since $v \in X_\bb^i$ and $X_\bb^i$ is an \nb-trap in $V$, if $v \in V_\bb$ it has one outgoing edge not leading to $X_\nb^i$ and 
    if $v \in V_\nb$, no outgoing edge of $v$ lead to $X_\nb^i$. That is, $v$ can either be made by \nb to reach $Z^i_\nb$ by staying in $X_\bb^i$ (i.e. it is in $\SafeReach^f_\nb(X_\bb^i, Z_\nb^i, \mathcal{G}^\ell)$),
    or $\bb = \Odd$ there exists a sequence of outgoing live edges that make $v$ reach $X_\nb^i$. This is not possible since there exists no live edges from $X_\Odd^i$ to $X_\Even^i$ due to Lem.~\ref{app-lem:safereacheven-noliveedges}.
\end{proof}
\begin{corollary}[Sec.~\ref{sec:zielonka:correct} Obs.~\ref{it:obs3}]\label{app-cor:increasing-decreasing-sequences}
    Due to Lem.~\ref{app-lem:X_nb-equivalence}, $\{X_\nb^{i}\}_{i\in \mathbb{N}}$ is an increasing sequence. Consequently, $\{X_\bb^{i}\}_{i\in \mathbb{N}}$ is a decreasing sequence. 
\end{corollary}
Since $V$ is finite, the corollary immediately implies that these sequences reach saturation value for some, and in fact the same, $k$. 

% \vspace{0.3cm}
\smallskip
\noindent\textbf{Part 1.}
We first assume an odd number $n$ is the maximum priority in $\mathcal{G}^\ell$.
Cor.~\ref{app-cor:increasing-decreasing-sequences} gives that $\{X_\Odd^i\}_{i\in \mathbb{N}}$ is an increasing sequence and saturates at some index $k$.
Observe that $X_\Odd^k$ is the saturation value if and only if $Z_\Even^k = \emptyset$.
The following proposition states that, \Odd safe reachability set of the saturation value $X_\Odd^k$ gives us \Wo.
\begin{proposition}\label{app-prop:n-odd}
    If $Z_\Even^k = \emptyset$, then $\SafeReach^f_\Odd(V, X^k_\Odd, \mathcal{G}^\ell)$ is an \Odd-paradise and $V \setminus \SafeReach^f_\Odd(V, X^k_\Odd, \mathcal{G}^\ell)$ is an \Even-paradise in $\mathcal{G}^\ell$.
\end{proposition}
We give the proof of Prop.~\ref{app-prop:n-odd} in three parts: First we prove $X^k_\Odd$ is an \Odd-paradise, then we show $\SafeReach^f_\Odd(V, X^k_\Odd, \mathcal{G}^\ell)$ is an \Odd-paradise, and lastly we prove that $V \setminus \SafeReach^f_\Odd(V, X^k_\Odd, \mathcal{G}^\ell)$ is an \Even-paradise. 

\begin{proof}\noindent \textbf{\textbf{ ($X^k_\Odd$ is an \Odd-paradise)}} 

\noindent Let $z$ be the winning \Odd strategy template on $Z^k = Z_\Odd^k$ in game $\mathcal{G}^\ell[Z^k]$. Any play $\pi$ that starts and stays in $Z^k$, and is compliant with $z$ is clearly \Odd winning.
However, $z$ is not necessarily an \Odd strategy template in $\mathcal{G}^\ell$ since there are possibly some $(v,w) \in E$ with $v \in Z^k \cap V_\Even$  and $w \not \in Z^k$.
For all such $(v,w)$, $w \in \SafeReach^f_\Odd(X^k_\Odd, N^k, \mathcal{G}^\ell)$ since $X^k_\Odd$ is an \Even-trap in $V$. Let $sr$ be the partial \Odd strategy template on $\SafeReach^f_\Odd(X^k_\Odd, N^k, \mathcal{G}^\ell)$, defined via the ranking function as presented during the introduction of safe reachability sets. 
Every (finite) play that starts in $\SafeReach^f_\Odd(X^k_\Odd, N^k, \mathcal{G}^\ell)$ compliant with $sr$ reaches $N^k$ in finitely many steps. The nodes in $N^k$ are dead ends in $sr$. 
Define an \Odd strategy template on $X^k_\Odd$ with the edge set $\e$ defined as follows:
$$ (v,w) \in \e \text{ if }\begin{cases} (v,w) \in z \cup sr,\\
    (v,w) \in E \text{ and } v \in V_\Even \cap X^k_\Odd,\\
    w = v_r \text{ if } v \in N^k \cap V_\Odd
\end{cases}
$$ where $v_r$ is a randomly chosen fixed successor for each $v\in  N^k \cap V_\Odd$, that is inside $X^k_\Odd$. Such a successor is guaranteed to exist since $X^k_\Odd$ is an \Even-trap.
Observe that all edges in $\e$ are in $X^k_\Odd \times X^k_\Odd$. However $(X^k_\Odd, \e)$ is not necessarily an \Odd strategy template in $\mathcal{G}^\ell$ since
there may be some $v \in V^\ell$ that lie on a cycle in $(X^k_\Odd, \e)$ but $\e$ does not contain their live outgoing edges. 
We will expand the edge set $\e$ to add the necessary live edges iteratively, like we did in ~\ref{const:S} (S3)-(S4).
$\overline{\e}$ is defined to be the saturation value of $\overline{e}^j$ such that:
$$\overline{e}^0 = \e, \quad  \quad \overline{e}^j = \overline{e}^{j-1} \cup \{(v, w) \in V^\ell \mid v \text{ lies on a cycle in } (X_\Odd^k, \overline{e}^{j-1})\}.$$

\vspace{0.1cm}
With this construction $\mathcal{S} = (X_\Odd^k, \overline{\e})$ is an \Odd strategy template in $\mathcal{G}^\ell$. We claim it is also a winning one. 

The underlying observation of the proof of the claim is that every play starting  $X_\Odd^k$ compliant with $\mathcal{S}$ that eventually stops seing a newly added cycle (one that is not in $z \cup sr$), stays in $Z^k$ and is won by \Odd obeying $z$; and every play that takes a newly added cycle infinitely often must see priority $n$ infinitely often, and is thus won by \Odd.

Let us look at a play $\pi$ compliant with $\mathcal{S}$. If $\pi$ eventually does not see a newly added cycle, it is clear that it wins by eventually obeying $z$ (since $sr$ does not contain any cycles).

Observe that for all newly added edges $(v,w)$ either (i)  $v \in V_\Even \cap Z^k$ and $w \in \SafeReach^f_\Odd(X^k_\Odd, N^k, \mathcal{G}^\ell)$, (ii) $v \in N^k$ or (iii) $(v, w) \in E^\ell $ where $v$ does not lie on a cycle in $z \cup sr$ and has a unique edge $(v,w') \in z \cup sr$, and this edge lies on a cycle in $\mathcal{S}$.

All the newly added cycles have to contain a newly added edge. 
If $\pi$ sees a new edge infinitely often, it visits $N^k$ infinitely often, and is thus won by \Odd. This is clear for edges of kind (ii).
Let $\pi$ see an edge of kind (iii) infinitely often. If $w\in V_\Even$, then all its outgoing edges achieves positive progress towards $N^k$, and if $w \in V_\Odd$, then it has an edge that achieves positive progress. Since $w$ is taken infinitely often, an edge that achieves positive progress towards $N^k$ will eventually be taken. Thus, $N^k$ will eventually be reached. That is, $\pi$ will visit $N^k$ infinitely often.
Finally let $\pi$ see an edge $(v,w)$ of kind (i) infinitely often. Then $(v,w')$ is also seen infinitely often. Let $C^1$ be the cycle that contains $(v,w')$. Since $C^1$ is also newly added, it contains a newly added edge $(v_1, w_1) \neq (v,w)$ since $C^1$ exists in $\overline{\e}$ before $(v,w)$ is added. If $(v_1, w_1)$ is of kind (i) or (ii), we are done. Assume the edge is of kind (iii)
and let $(v_1, w'_1)$ be the unique outgoing edge of $v_1$ in $z \cup sr$. $(v_1, w'_1)$ lies on a newly added cycle $C^2$. Let $(v_2, w_2) \not \in \{(v, w), (v_1, w_1)\} $ be the newly added edge in $C^2$. 
Carry on in this manner, assuming all newly added edges $(v_i, w_i)$ are of kind (iii). Since all $(v_i, w_i)$ are distinct and there are a finite number of live edges, for some $C^r$, $(v_r, w_r)$ should be of kind (i) or (ii).
Since $\pi$ sees $v$ infinitely often it should see all $C^i$ infinitely often, and since $C^r$ visits $N^k$, $\pi$ visists $N^k$ infinitely often. Thus, $\pi$ is won by \Odd.

\vspace{0.2cm}
\noindent \textbf{($\SafeReach^f_\Odd(V, X^k_\Odd, \mathcal{G}^\ell)$ is an \Odd-paradise)} 

\noindent Since $X_\Odd^k$ is an \Odd-paradise in $\mathcal{G}^\ell$, by Lem.~\ref{app-lem:safe-reach-Odd-paradise} we get that $\SafeReach^f_\Odd(V, X_\Odd^k, \mathcal{G}^\ell)$ is again an \Odd-paradise in $\mathcal{G}^\ell$.
\vspace{0.2cm}

\noindent\textbf{($V \setminus \SafeReach^f_\Odd(V, X^k_\Odd, \mathcal{G}^\ell)$ is an \Even-paradise)} 

\noindent Let $T:=\SafeReach^f_\Odd(V, X^k_\Odd, \mathcal{G}^\ell)$ and 
$\Xsr_\Even^i := \SafeReach^f_\Even(X_\Odd^i, Z_\Even^i, \mathcal{G}^\ell)$. Let the partial \Even strategy template on $\Xsr_\Even^i$ be denoted by $sr^i$ and the winning \Even strategy on $Z_\Even^i$ in game $\mathcal{G}^\ell[Z^i]$ be denoted by $z^i$. By Lem.~\ref{app-lem:Kuesters6.4}, $V \setminus T$ is an \Odd-trap.
Cor.~\ref{app-cor:increasing-decreasing-sequences} gives us that $\{X_\Even^i\}_{i \in \mathbb{N}}$ is an increasing sequence. 
Furthermore by Lem.~\ref{app-lem:X_nb-equivalence}, which gives an alternative definition for $X_\Even^{i+1}$, we observe that each $v \in X^k_\Even$ belongs to $\Xsr^j$ for some $j < k$.
Moreover, we can observe that $X^i_\Even$ and $\Xsr^i$ are disjoint sets, due to $X_\Even^i$ and $X_\Odd^i$ being disjoint. Therefore, we conclude that 
each $v \in X_\Even^k$ belongs to a unique $Xsr^j$. The same clearly holds for $v \in V\setminus T$, since $(V\setminus T) \subseteq X_\Even^k$.
Furthermore, since $V\setminus T$ is an \Odd-trap, for all $(v,w) \in E$ with $v \in V_\Odd \cap (V \setminus T)$, $w \in (V\setminus T)$.

We construct the \Even strategy template $\mathcal{S} = (X, \e)$ where $\e$ is defined as follows: $(v,w) \in E $ is in $\e$ if, 
$$\begin{cases}v \in V_\Odd\\
    v \in Z_\Even^i \text{ and } (v,w) \text{ is the unique outgoing edge of }v \text{ in } z^i\\
    v \in \Xsr_\Even^i \setminus Z_\Even^i \text{ and } (v,w) \text{ is the unique outgoing edge of }v \text{ in } sr^i\\
\end{cases}$$

It is clear that $\mathcal{S}$ is an \Even strategy template since it contains all outgoing edges of \Odd nodes in $V \setminus T$, and a unique outgoing edge for each \Even node in $V \setminus T$. We claim that $\mathcal{S}$ is also winning.
To prove this claim we will need the following two observations. 

Let $\pi = v_1 v_2 \ldots $ be a fair play that start in $V \setminus T$ and is compliant with $\mathcal{S}$. Let $\Xsr(\pi) = \Xsr_1\Xsr_2\Xsr_3\ldots$ be such that $\Xsr_i$ is the unique $\Xsr^j$, $v_i$ belongs to.

(1) If $v_t \in Z_\Even^i$, then $v_{t+1}$ is either in $Z_\Even^i$ or in $\Xsr^r$ for some $t < i$. This follows from $Z_\Even^i$ being an \Odd-trap in $X_\Odd^i$ (by Obs.~\ref{app-obs:traps-subgames}).

(2) If $\Xsr^i$ is seen infinitely often in $\Xsr(\pi)$, then $Z_\Even^i$ is seen infinitely often as well. Due to the pigeonhole principle, $\Xsr^i$ being visited infinitely often in $\Xsr(\pi)$ implies that some $v \in \Xsr^i$ is visited infinitely often.
If $v \not \in Z_\Even^i$, it is in $\Xsr^i \setminus Z_\Even^i$. Say $v \in V_\Even$, then the unique $(v,w) \in \e$ causes positive progress towards $Z_\Even^i$. If $v \in V_\Odd \setminus V^\ell$, then all of the outgoing edges of $v$ cause positive progress towards $Z_\Even^i$.
If $v \in V^\ell$, there is at least one $(v,w) \in E^\ell$ causing positive progress towards $Z_\Even^i$. Since $v$ is seen infinitely often in $\pi$, this edge is taken infinitely often as well. 
By induction, $\pi$ visits $Z_\Even^i$ infinitely often.

\vspace{0.2cm}
\noindent \emph{Claim:} Any fair play $\pi$ starting in $X$ and compliant with $\mathcal{S}$ eventually stays in $Z_\Even^i$ for some $i$.

\noindent \emph{Proof of Claim.}
Let $i$ be the minimum index for which $\Xsr^i$ appears infinitely often in $\Xsr(\pi)$. By observation (2), $\pi$ sees a set of nodes $P \subseteq Z_\Even^i$ infinitely often. Let $v_t \in P$. By observation (1), $v_{t+1}$ is either in $Z_\Even^i$ or in $\Xsr^r$ for some $r < i$.
Since $i$ is the minimum index for which $\Xsr^i$ is seen infinitely often in $\Xsr(\pi)$, after some $t' \in \mathbb{N}$, for all $v_{t'}\in P$, $v_{t'+1} \in Z_\Even^i$.

Since $\pi$ eventually stays in $Z_\Even^i$, the strategy $\mathcal{S}$ eventually collapses to $z_\Even^i$ and thus, \Even wins $\pi$.
\end{proof}

With this, we have proven Prop.~\ref{app-prop:n-odd}, and therefore have given an algorithm to calculate \We and \Wo on an \Odd-fair parity game with 
an odd upper bound $n$ on the priorities in the game graph. The algorithm however requires a sibling-algorithm that does the same for an \Odd-fair parity game with an upper bound $n-1$ on its priorities. In the second part that follows, we give this sibling-algorithm.

% \vspace{0.5cm}
\smallskip
\noindent\textbf{Part 2.}
We now assume an even number $n$ is the maximum priority in $\mathcal{G}^\ell$. %We construct a decreasing series of subsets of $V$, $\{X_\Odd^i\}_{i\in \mathbb{N}}$ by assigning the following sets (Fig.~\ref{fig:X_Even}):
% \vspace{0.3cm}
%Now we know how to get the \Odd and \Even winning regions in a game $\mathcal{G}^\ell$, given that the highest priority $n$ in the game is odd, by reducing it down to solving subgames where the highes priority is $n-1$.
%We should show how to get \Odd and \Even winning regions in a game where the highest priority is even, then we will be able to apply these two constructions recursively to get an variant of Zielonka's algorithm that solves \Odd-fair parity games. 
% 
%This part of the proof is almost identical to the proof of Zielonka's algorithm for regular parity games. %We use strategy te
%But we will give out a full proof for the sake of completeness. 
% 
We set the sets as before, and because $n$ is even, this time $\{X_\Odd^i\}_{i \in \mathbb{N}}$ is an increasing sequence and $\{X_\Even^i\}_{i \in \mathbb{N}}$ is a decreasing one (Fig.~\ref{fig:kuesters-figure-extended}).
Both sequences saturate at some index $k$, and for this $k$, $Z_\Odd^k = \emptyset$. Furthermore, $X_\Even^k$ and $X_\Odd^k$ are \We and \Wo, respectively.
%We again set the sets with the same names as before, the division can be seen in Fig...
%Let $W_O$ be an \Odd paradise in $\mathcal{G}^\ell$, and let $X_\Odd = \SafeReach_\Odd(V, W_E, \mathcal{G}^\ell)$. It is clear that $X_\Odd$ is an \Odd paradise in $\mathcal{G}^\ell$.
%Let $X_\Even = V \setminus X_\Odd$. $X_\Even$ is an \Odd-trap. $N$ is defined as before and $Z:= X_\Even \setminus \SafeReach_\N(X_\Even, N, \mathcal{G}^\ell)$. $Z$ is an \Even-trap in $X_\Even$, and thus $\mathcal{G}^\ell[Z]$ is a subgame. 
%Let $Z_\Even$ and $Z_\Odd$ be the winning regions of \Even and \Odd in $\mathcal{G}^\ell[Z]$ as before. 

\begin{proposition} For all $i$, $Z^i_\Odd \cup X^i_\Odd$ is an \Odd-paradise in $\mathcal{G}^\ell$. 
\end{proposition}
\begin{proof}
    The fact that $Z^i_\Odd \cup X^i_\Odd$ is an \Even-trap follows from the observations in \ref{app-obs:traps-subgames}.
%The union is an \Even trap: $X_\Odd$ is an \Even trap in $\mathcal{G}^\ell$ and $Z_\Odd$ is an \Even trap in $X_\Even$ since it is an \Even trap in  $Z$ and $Z$ is an \Even trap in $X_\Even$ (lemma..).
%Then, $v \in V_\Odd \cap (Z_\Odd \cup X_\Odd)$ has an outgoing edge to staying inside the set, and for  $v \in V_\Even \cap (Z_\Odd \cup X_\Odd)$ all the outgoing edges are either from $X_\Odd$ to $X_\Odd$, $Z_\Odd$ to $Z_\Odd$ or from $Z_\Odd$ to $X_\Odd$.
%Therefore,  $Z_\Odd \cup X_\Odd$ is an \Even trap in $\mathcal{G}^\ell$.

Let us denote the winning \Odd strategy template on $ Z^i_\Odd $ in $\mathcal{G}^\ell[Z^i]$ with $z$ 
and the strategy template on $X_\Odd^i$ in $\mathcal{G}^\ell$ by $x$. 
Let $\e$ be the edge set that contains all edges in $z \cup x$, together with all $\{(v,w) \in E \mid v \in V_\Even \cap (Z^i_\Odd \cup Z_\Odd^i) \}$.
Due to $X_\Odd^i$ being an \Even-trap in $V$, all outgoing edges of \Even nodes in $X_\Odd^i$, stay in $X_\Odd^i$. Then, $\e$ does not introduce any new cycles to $z \cup x$ since all the newly added edges are in one direction, from $Z^i_\Odd$ to $X^i_\Odd$. Thus, $\mathcal{S} = (X_\Odd^i \cup Z_\Odd^i, \e)$ is an \Odd strategy template in $\mathcal{G}^\ell$.
We claim it is also a winning one. 
A play $\pi$ starting in $X_\Odd^i$ and compliant with $\mathcal{S}$ stays in $X_\Odd^i$ and therefore wins by obeying $x$. 
If $\pi$ starts in $Z_\Odd^i$, it either eventually reaches $X_\Odd^i$ and therefore wins by the previous argument. Or, it stays in $Z_\Odd^i$ and wins by obeying $z$.
\end{proof}

\begin{proposition} If $Z^i_\Odd = \emptyset$, $X^i_\Even$ is an \Even-paradise in $\mathcal{G}^\ell$.
\end{proposition}

\begin{proof}
We know $X^i_\Even$ is an \Odd-trap~\ref{app-obs:traps-subgames}. Let $z$ be the winning \Even strategy on $Z^i_\Even$ in subgame $\mathcal{G}^\ell[Z^i]$ and $sr$ be the partial strategy template on $\SafeReach^f_\Even(X^i_\Even, N^i, \mathcal{G}^\ell)$ where all nodes in $\SafeReach^f_\Even(X^i_\Even, N^i, \mathcal{G}^\ell) \setminus N^i$ are forced to positive progress towards $N^i$ in the next step, and nodes in $N^i$ are dead-ends.

We construct an \Even strategy template $\mathcal{S} = (X_\Even^i, \e)$ where $\e$ is defined as follows: 
$$
(v,w) \in \e \text{ if }\begin{cases}
    (v, w) \in z \cup sr, \\
    (v, w)\in E \text{ and } v\in V_\Odd \cap X_\Even,\\
    w = v_r \text{ if } v \in N^i \cap V_\Even
\end{cases}
$$ where $v_r$ is a randomly chosen fixed successor for each $v \in N^i \cap V_\Even$, that is inside $X_\Even^i$. Such a successor is guaranteed to exist since $X_\Even^i$ is an \Odd-trap.

$\mathcal{S}$ is clearly an \Even strategy template in $\mathcal{G}^\ell$ since all \Odd nodes in $X_\Even^i$ have all their outgoing edges in $\mathcal{S}$ and all \Even nodes have a unique outgoing edge.
We claim it is also winning. 

Let $\pi$ be a play that starts in $X_\Even^i$ and is compliant with $\mathcal{S}$. We claim $\pi$ either (i) eventually stays in $Z_\Even^i$, and therefore eventually obeys $z$ or (ii) it sees $N^i$ infinitely often.
It is easy to see that in both of these cases $\pi$ is \Even winning. We will try to show that one of these cases must occur.
Assume $\pi$ does not eventually stay in $Z^i_\Even$. Then $\pi$ visits some $ v \in \SafeReach^f_\Even(X^i_\Even, N^i, \mathcal{G}^\ell)$ infinitely often. If $v \in V_\Odd$, all outgoing edges of $v$ are in $sr$ make positive progress towards $N^i$, and if $v \in V_\Even $ the unique successor of $v$ in $sr$ make positive progress towards $N^i$. 
Thus, $\pi$ visists $N^i$ after finitely many steps. Since $v$ is visited infinitely often by $\pi$, $N^i$ is also visited infinitely often.
\end{proof}

\subparagraph{Corrrectness of Alg.~\ref{algo:fair-zielonka-bb}.}
The $X$ set in $\SOLVE_\Odd(n, \mathcal{G}^\ell)$ holds the value of $X_\Odd^i$ and
the $X$ set in $\SOLVE_\Even(n, \mathcal{G}^\ell)$ holds the value of $X_\Even^i$ at the $i^{th}$ iteration of their respective \emph{while} loops. 
Note that both of these sequences are initialized at $V$ and are strictly decreasing, until they reach their saturation value $X_\Odd^k$ or $X_\Even^{k'}$. When these saturation values are reached $Z_\Even^k = \emptyset $ in the $\SOLVE_\Odd$ procedure and $Z_\Odd^{k'} = \emptyset $ in the $\SOLVE_\Even$ procedure. 
This is exactly when $\SOLVE_\Even$ returns $X_\Even^{k}$ and $\SOLVE_\Odd$ returns $\SafeReach^f_\Odd(V, X_\Odd^{k'}, \mathcal{G}^\ell)$; correctfully returning their respective winning regions according to the correctness proof of Thm.~\ref{thm:solvebb}.

\subsection{Details on Experimental Results}\label{app:experiments}

We conducted an experimental study to empirically validate the claim that our new \Odd-fair Zielonka's algorithm retains its efficiency in practice. For this, we implemented the following algorithms (non-optimized) in C++:
\begin{itemize}
 \item \texttt{OF-ZL}: \Odd-fair Zielonka's algorithm (Alg.~\ref{algo:fair-zielonka-bb}),
 \item \texttt{N-ZL}: \enquote{normal} Zielonka's algorithm from \cite{Zielonka98} (i.e., Alg.~\ref{algo:fair-zielonka-bb} with the simplifications described in Sec.~\ref{sec:zielonka:orig}),
 \item \texttt{OF-FP}: the fixed-point algorithm for \Odd-fair parity games implementing \eqref{eq:fp-odd} ,
 \item \texttt{N-FP}: the fixed-point algorithm for \enquote{normal} parity games from \cite{EJ91}.
\end{itemize}
Of course, for both \texttt{N-ZL} and \texttt{N-FP} there exist optimized implementations (e.g. \texttt{oink}~\cite{oink}). However, the goal of this section is to show a conceptual comparison, rather than evaluating best computation times. We believe this is better achieved using similar (non-optimized) implementations for all algorithms. In particular, by our experiments we show:
\begin{enumerate}
 \item \texttt{OF-ZL}: is largely insensitive to the number of priorities and number of fair edges (Fig.~\ref{fig:percentages-colours}),
 \item \texttt{OF-ZL}: significantly outperforms \texttt{OF-FP} on almost all benchmarks (Fig.~\ref{fig:zoomed_out} (right))
\item the performance of \texttt{OF-ZL} and \texttt{N-ZL} on the given benchmark set is very similar (Fig.~\ref{fig:zielonkas_comparison}),
 \item the comparative performance of \texttt{OF-ZL} and \texttt{N-ZL} w.r.t.\ their respective fixed-point versions \texttt{OF-FP} and \texttt{N-FP}), respectively, is very similar (see Fig.~\ref{fig:logscale}).
\end{enumerate}
All experiments where run on a large benchmark suite explained in Sec.~\ref{app:experiments:benchmarks}. To perform our experiments we used a machine equipped with Intel(R) Core(TM) i5-6600 CPU @ 3.30GHz and 8GB RAM. We declare a timeout when the calculation of an example exceeds 1 hour.

\subsubsection{Benchmark}\label{app:experiments:benchmarks}
We generated \Odd-fair parity game instances manipulating $286$ benchmark instances of PGAME$\_$Synth$\_$2021 dataset of the SYNTCOMP benchmark suite~\cite{syntcomp} and $51$ benchmark instances of the PGSolver dataset of Keiren's benchmark suite~\cite{keirens}. Within the latter, we restricted ourselfs to instances with $\leq 5000$ nodes.
Both datasets contain examples of normal parity games. For each selected
example, we generate \Odd-fair parity game instances for a particular liveness percentage $\alpha$. For a $\alpha\%$-liveness instant, we fix $\alpha\%$ of the \Odd nodes in the game, and turn $\alpha\%$ of each of their outgoing edges to live edges. %, where $\alpha$ is either $30$ or $50$. 
In addition, we also generated \Odd-fair parity game instances with varying number of priorities $p$ by partitioning the nodes of the games uniformly at random according to the number of priorities.

%We observed that \textsc{OF-FP} times out on \emph{all} instances of the PGSolver dataset, while \textsc{OF-ZL} terminated on all instances with $24,9$ seconds average computation time. %Further, \textsc{OF-FP} times out on $67$ out of $304$ games from the PGAME$\_$Synth$\_$2021 dataset, while \textsc{OF-Z} only times out on $48$ instances. For the instances where both solvers compute a solution, \textsc{OF-Z} is faster in all but one case.
%
Detailed run-times of all algorithms on a representative selection of examples from the instances fenerated from SYNTCOMP benchmark suite are listed in Table~\ref{table:FPvsOddfairzlk}. 
On the \Odd-fair instances with $50\%-$liveness generated from the SYNTCOMP benchmark suite, there are 204 instances where neither of the algorithms \texttt{OF-FP}, \texttt{OF-ZL}, \texttt{N-FP} or \texttt{N-ZL} timed out. On these instances, \texttt{OF-ZL} gives an average computation time of $4.6$ seconds while \texttt{OF-FP} took $122.7$ seconds on average. 
On the same examples, \texttt{N-ZL} takes on average $3.6$ seconds to compute while \texttt{N-FP} gives an average of $45.2$ seconds. 
For the PGSolver dataset \texttt{OF-FP} timed out on all generated instances, whereas \texttt{OF-ZL} took $24.9$ seconds on average to terminate.

% In order to compare the algorithmic advantage Zielonka's algorithm posesses over the fixed-point implementation, 
% we have implemented a naive fixed-point algorithm in the same fashion with \texttt{OF-FP} and our own version of Zielonka's algorithm for regular parity games in the same fashion with \texttt{OF-ZL}.
% We call these these algorithms \texttt{N-FP} and \texttt{N-ZL} respectively, where \texttt{R} stands for `regular'.  
% We do not use any efficient online parity solvers and instead use our 'non-optimal' code that is written similarly to their fair-parity counterparts in order for the comparison to be meaningful.

\subsubsection{Sensitivity}\label{app:experiments:sensitivity}
To monitor the sensitivity of \texttt{OF-ZL} to the change in number of priorities as well as the percentage of live edges in the game, we picked $12$ parity game instances from the SYNTCOMP dataset which did not timeout (after one hour). %\AKS{are those also depicted in the table? if so, could we mark the rows?}\IS{They are mostly not in the table.}
With priorities $3-4-5-6$ and liveness degrees 0$\%$\footnote{regular parity game}-30$\%$-50$\%$-80$\%$ we get 192 different \Odd-fair parity instances. Fig.~\ref{fig:percentages-colours} shows the runtime of \texttt{OF-ZL} on these instances.

We can see that the runtimes of instances with different priority and liveness percentages are distributed in a seemingly random manner.
This tells us that \Odd-fair Zielonka's algorithm is highly insensitive to a change in the percentage of live edges and the number of priorities. %\todo{IS:whereas fixed-point algorithm...?} 
This observation is inline with the known insensitivity of Zielonka's algorithm for the number of priorities.

\begin{figure}%PERCENTAGES-COLOURS
\begin{center}
 \includegraphics[width=8cm]{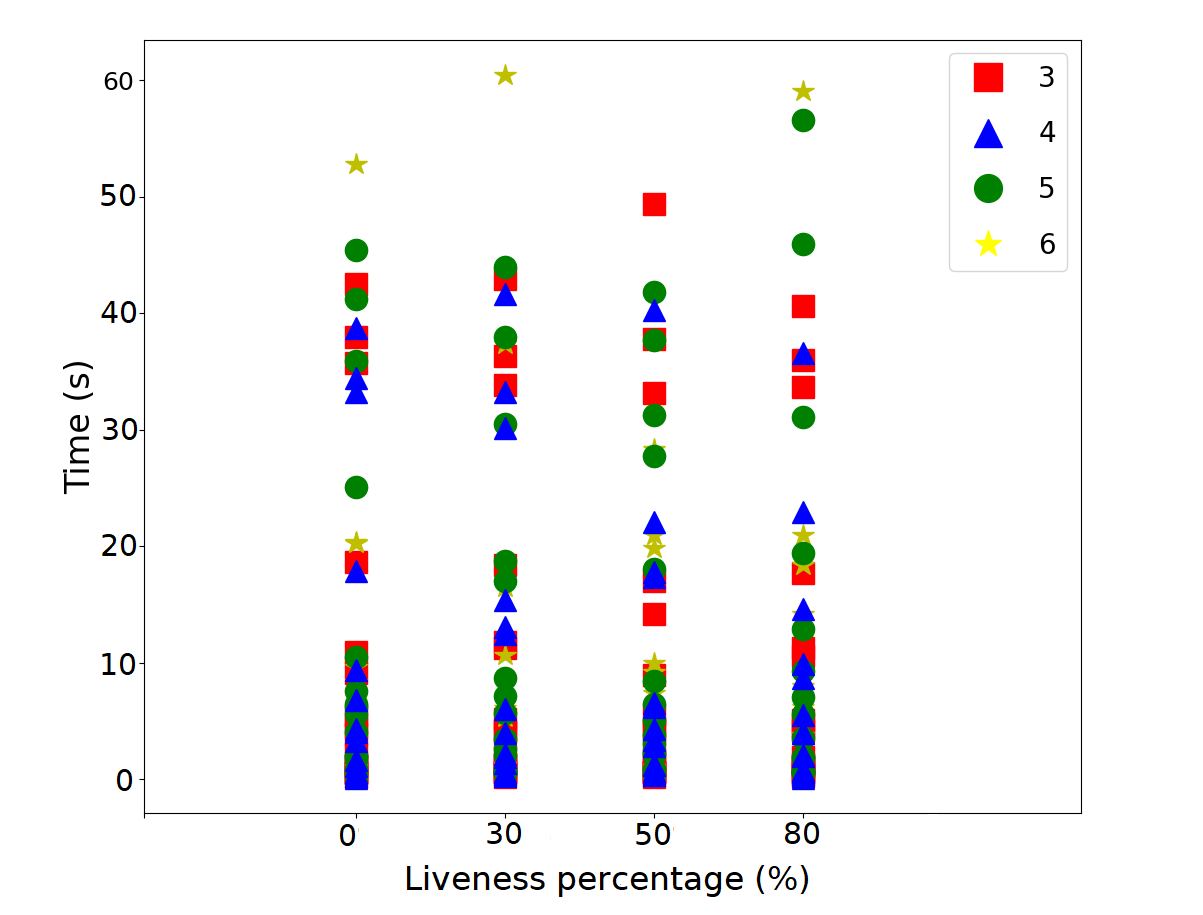}
\end{center}
\caption{Runtime of \texttt{OF-ZL} on the 192 \Odd-fair parity instances generated from 12 fixed parity examples through changing their priorities and liveness degrees. Different shapes indicate the number of prioirities an instance has, and the $x-$axis denotes their liveness percentages. At each coloumn we view 48 different instances of the 12 examples with varying colours.}
\label{fig:percentages-colours}
\end{figure}

\subsubsection{Comparative Evaluation}\label{app:experiments:comparison}
In order to validate the computational advantage of \texttt{OF-ZL} over \texttt{OF-FP}, we have run both algorithms on all 50$\%$-liveness instances generated from the SYNTCOMP benchmark dataset. On 58 of these instances, both algorithms time out. The run-times for all other instances are depicted in Fig.~\ref{fig:zoomed_out} (right),~\ref{fig:zoomed_in} (right) and~\ref{fig:logscale} (right). The left plots in Fig.~\ref{fig:zoomed_out}-\ref{fig:logscale} show the same comparison for the \enquote{normal} parity algorithms \texttt{N-ZL} and \texttt{N-FP}. 
In both cases, Fig.~\ref{fig:zoomed_in} shows the zoomed-in version of the respective plot in Fig.~\ref{fig:zoomed_out}. Fig.~\ref{fig:logscale} shows the data-points from the respective plot in Fig.~\ref{fig:zoomed_in} as a scatter plot in log-scale. 
The examples on which only \texttt{x-FP} times out, can be seen as the dots on the ceiling of the plots in Fig.~\ref{fig:zoomed_out}. In all plots, points above the diagonal correspond to instances where Zielonka's algorithm outperforms the fixed-point algorithm.

We clearly see in Fig.~\ref{fig:zoomed_out}-\ref{fig:logscale} that Zielonka's algorithm performs significantly better than the fixed-point version, both in the \Odd-fair (right) and in the normal (left) case. More importantly, the overall performance comparison between  \texttt{OF-ZL} over \texttt{OF-FP} (right plots) mimics the comparison between \texttt{N-ZL} over \texttt{N-FP}. This allows us to conclude that our new \Odd-fair Zielonka's algorithm retains the computational advantages of Zielonka's algorithm. % and handling fairness only introduces a mild overhead. 

In addition, Table~\ref{table:FPvsOddfairzlk} shows that \texttt{OF-ZL} results in almost the same run-time as \texttt{N-ZL}, showing that our changes in the algorithm incur almost no computational disadvantages over the original algorithm. This allows us to handle transition fairness for almost free in practice.

\begin{figure}%ZOOMED-OUT
\centering
\begin{subfigure}{.45\textwidth}
     \centering
     \includegraphics[scale=0.28]{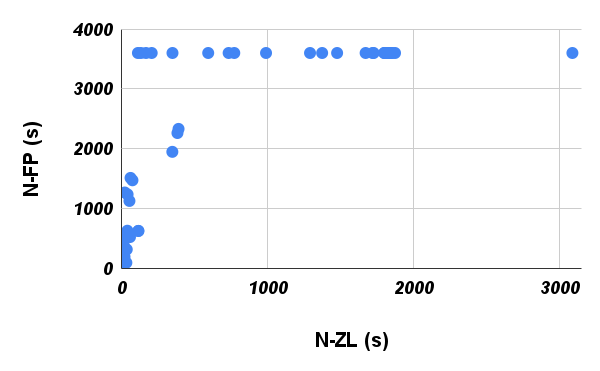}
     %\caption{A subfigure}
     \label{fig:zoomed_out_regular}
\end{subfigure}%
\begin{subfigure}{.45\textwidth}
     \centering
     \includegraphics[scale=0.28]{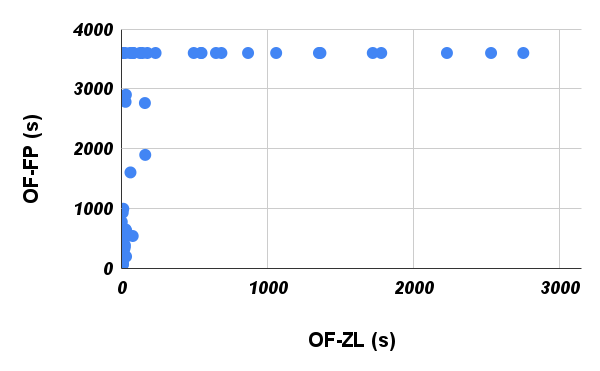}
     %\caption{A subfigure}
     \label{fig:zoomed_out_fair}
\end{subfigure}
\caption{(Zoomed out version) A comparison of \texttt{N-FP} vs. \texttt{N-ZL} in regular parity games (left), and \texttt{OF-FP} vs. \texttt{OF-ZL} on fair parity games (right)}
\label{fig:zoomed_out}
\end{figure}

\begin{figure}%ZOOMED-IN
\centering
\begin{subfigure}{.5\textwidth}
     \centering
     \includegraphics[scale=0.28]{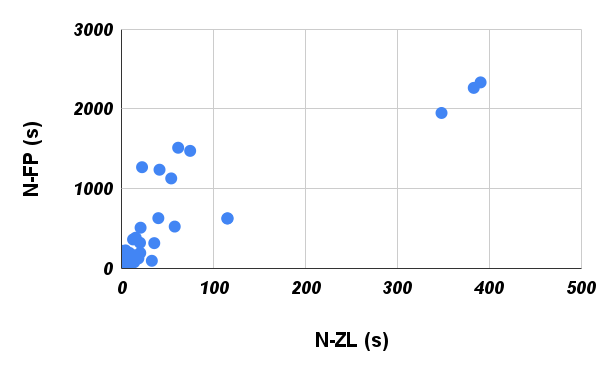}
     %\caption{A subfigure}
     \label{fig:zoomed_in_regular}
\end{subfigure}%
\begin{subfigure}{.5\textwidth}
     \centering
     \includegraphics[scale=0.28]{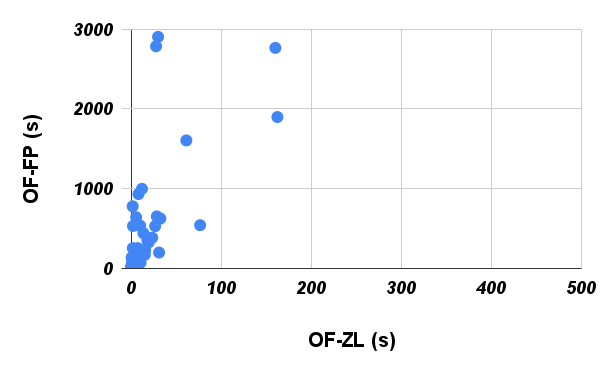}
     %\caption{A subfigure}
     \label{fig:zoomed_in_fair}
\end{subfigure}
\caption{(Zoomed in version) A comparison of \texttt{N-FP} vs. \texttt{N-ZL} in regular parity games (left), and \texttt{OF-FP} vs. \texttt{OF-ZL} on fair parity games (right)}
\label{fig:zoomed_in}
\end{figure}

\begin{figure}%LOGSCALE
\centering
\begin{subfigure}{.5\textwidth}
     \centering
     \includegraphics[scale=0.28]{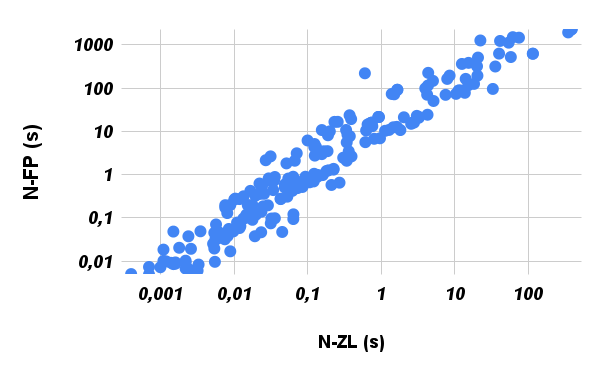}
     %\caption{A subfigure}
     \label{fig:logscale_regular}
\end{subfigure}%
\begin{subfigure}{.5\textwidth}
     \centering
     \includegraphics[scale=0.27]{figures/experiments/logplot_fair_parity_fp_vs_zl.png}
     %\caption{A subfigure}
     \label{fig:logscale_fair}
\end{subfigure}
\caption{A comparison of \texttt{N-FP} vs. \texttt{N-ZL} in regular parity games (left), and \texttt{OF-FP} vs. \texttt{OF-ZL} on fair parity games (right) in terms of log-scale plots where the timeouts are removed.}
\label{fig:logscale}
\end{figure}

\begin{figure}%ZIELONKAS COMPARISON
\centering
\begin{subfigure}{.5\textwidth}
     \centering
     \includegraphics[scale=0.28]{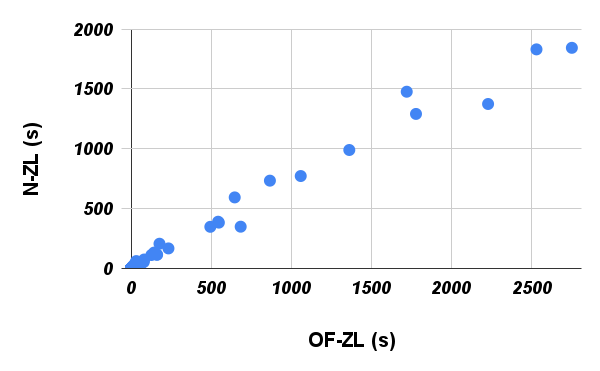}
     %\caption{A subfigure}
     \label{fig:logscale_regular}
\end{subfigure}%
\begin{subfigure}{.5\textwidth}
     \centering
     \includegraphics[scale=0.28]{figures/experiments/zielonkas_comparison_logscale.png}
     %\caption{A subfigure}
     \label{fig:logscale_fair}
\end{subfigure}
\caption{A comparison of \texttt{N-ZL} vs. \texttt{OF-ZL} over examples that do not timeout on both. Right hand side plot visualizes the same data in logscale.}
\label{fig:zielonkas_comparison}
\end{figure}

\smallskip
\textbf{Conclusion:}
The results show that Zielonka's algorithm is significantly faster in solving \Odd-fair parity games compared to the calculation performed by the fixed-point algorithm, as is the case in normal parity games. 
The fixed-point algorithm started timing out as soon as the examples became more complex, being especially sensitive to  
the increase in the number of priorities. Whereas, Zielonka's algorithm preserves its performance considerably in the face of the increase in the same parameters. These 
outcomes match the known comparison results between the naive fixed-point calculation versus Zielonka's algorithm, on normal parity games.

\newpage
% \begin{center}

     \begin{longtable}{||c |c |c |c |c |c ||} %begin{table}
     \caption{Detailed run-time comparison of \texttt{N-FP} and \texttt{N-ZL} on the original parity game instances (yellow rows) with \texttt{OF-FP} and \texttt{OF-ZL} on their respective $30\%$- and $50\%$-liveness \Odd-fair parity game instances (white rows). The instance name is taken from the original benchmark suite.
%      parity examples (the rows with the original example's names)
% and \Odd-fair instances of these examples with 30$\%$- and 50$\%$-liveness (the rows with the original example's names prefixed by \enquote{30$\%$-} or \enquote{50$\%$-})
% solved via a naive implementation of their corresponding fixed point and Zielonka's algorithms. For the regular parity examples \texttt{FP} and \texttt{ZL} correspond to \texttt{N-FP} and \texttt{N-ZL}, and for \Odd-fair parity instances they correspond to \texttt{OF-FP} and \texttt{OF-ZL}, respectively.
}\label{table:FPvsOddfairzlk}\\
     
    %      \begin{tabular}{||c |c |c |c |c |c ||}
 \hline
 Name &  $\#$  &  $\#$  &  $\#$  & \texttt{FP} & \texttt{ZL}  \\ [0.5ex] 
      &    nodes    &   edges     &   priorities &     (sec.)       & (sec.)  \\
 \hline  \hline \rowcolor{Highlight}
 EscalatorCountingInit & 99 & 148 & 3 & 0.064 & 0.012  \\ 
 \hline
 $30\%$-EscalatorCountingInit & 99 & 148 & 3 & 0.075 & 0.018  \\ 
 \hline
 $50\%$-EscalatorCountingInit & 99 & 148 & 3 & 0.072 & 0.02 \\
 \hline \rowcolor{Highlight}
 KitchenTimerV1 & 80 & 124 & 3 & 0.055 & 0.008 \\
 \hline
 $30\%$-KitchenTimerV1 & 80 & 124 & 3 & 0.068 & 0.012 \\
 \hline
 $50\%$-KitchenTimerV1 & 80 & 124 & 3 & 0.21 & 0.009  \\
 \hline \rowcolor{Highlight}
 KitchenTimerV6 & 4099 &  6560 & 3 & 87 & 11 \\
 \hline
 $30\%$-KitchenTimerV6 & 4099 &  6560 & 3 & 88 & 11 \\
 \hline
 $50\%$-KitchenTimerV6 & 4099 &  6560 & 3 & 352 & 18 \\
 \hline \rowcolor{Highlight}
 MusicAppSimple & 344 &  562 & 3 & 0.488 & 0.073 \\
 \hline
 $30\%$-MusicAppSimple & 344 &  562 & 3 & 0.496 & 0.082 \\
 \hline
 $50\%$-MusicAppSimple & 344 &  562 & 3 & 0.799 & 0.089 \\
 \hline \rowcolor{Highlight}
 TwoCountersRefinedRefined & 1933 & 3140 & 3 & 14.9 & 2.5 \\
 \hline
 $30\%$-TwoCountersRefinedRefined & 1933 & 3140 & 3 & 15 & 1.2 \\
 \hline 
 $50\%$-TwoCountersRefinedRefined & 1933 & 3140 & 3 & 74 & 3.72 \\
 \hline \rowcolor{Highlight}
 Zoo5 & 479 & 768 & 3 & 0.96 & 0.135 \\
 \hline
 $30\%$-Zoo5 & 479 & 768 & 3 & 0.981 & 0.152 \\
 \hline 
 $50\%$-Zoo5 & 479 & 768 & 3 & 1.57 & 0.172 \\
 \hline \rowcolor{Highlight}
 amba$\_$decomposed$\_$lock$\_$3 & 1558 & 2336 & 3 & 72 & 1.5  \\
 \hline
 $30\%$-amba$\_$decomposed$\_$lock$\_$3 & 1558 & 2336 & 3 & 73 & 1.5  \\
 \hline 
 $50\%$-amba$\_$decomposed$\_$lock$\_$3 & 1558 & 2336 & 3 & 56 & 2.9 \\
 \hline \rowcolor{Highlight}
 full$\_$arbiter$\_$2 & 204 & 324 & 3 & 0.59 & 0.049 \\
 \hline
 $30\%$-full$\_$arbiter$\_$2 & 204 & 324 & 3 & 0.602 & 0.047 \\
 \hline 
 $50\%$-full$\_$arbiter$\_$2 & 204 & 324 & 3 & 5 & 0.059 \\
 \hline \rowcolor{Highlight}
 full$\_$arbiter$\_$3 & 1403 & 2396 & 3 & 21.18 & 2 \\
 \hline
 $30\%$-full$\_$arbiter$\_$3 & 1403 & 2396 & 3 & 21.5 & 2 \\
 \hline 
 $50\%$-full$\_$arbiter$\_$3 & 1403 & 2396 & 3 & 93 & 3.46 \\
 \hline \rowcolor{Highlight}
 lilydemo06 & 369 & 548 & 3 & 8.1 & 0.18 \\
 \hline 
 $30\%$-lilydemo06 & 369 & 548 & 3 & 8.13 & 0.206 \\ 
 \hline 
 $50\%$-lilydemo06 & 369 & 548 & 3 & 18 & 0.212 \\
 \hline \rowcolor{Highlight}
 lilydemo07 & 78 & 108 & 3 & 0.27 & 0.01 \\ 
 \hline
 $30\%$-lilydemo07 & 78 & 108 & 3 & 0.284 & 0.017 \\ 
 \hline 
 $50\%$-lilydemo07 & 78 & 108 & 3 & 0.33 & 0.008 \\ 
 \hline \rowcolor{Highlight}
 simple$\_$arbiter$\_$unreal1 & 2178 & 3676 & 3 & 22.8 & 3 \\
 \hline
 $30\%$-simple$\_$arbiter$\_$unreal1 & 2178 & 3676 & 3 & 23 & 3 \\
 \hline 
 $50\%$-simple$\_$arbiter$\_$unreal1 & 2178 & 3676 & 3 & 254 & 7 \\ 
 \hline \rowcolor{Highlight}
 amba$\_$decomposed$\_$arbiter$\_$2 & 141 & 212 & 4 & 0.72 & 0.03  \\
 \hline
 $30\%$-amba$\_$decomposed$\_$arbiter$\_$2 & 141 & 212 & 4 & 0.73 & 0.06  \\
 \hline
 $50\%$-amba$\_$decomposed$\_$arbiter$\_$2 & 141 & 212 & 4 & 1 & 0.035  \\
 \hline \rowcolor{Highlight}
 loadfull3 & 1159 & 2030 & 4 & 5.62 & 0.609 \\
 \hline
 $30\%$-loadfull3 & 1159 & 2030 & 4 & 5 & 0.614 \\
 \hline 
 $50\%$-loadfull3 & 1159 & 2030 & 4 & 5 & 0.754 \\
 \hline \rowcolor{Highlight}
 ltl2dba01 & 101 & 152 & 4 & 0.074 & 0.031  \\
 \hline
 $30\%$-ltl2dba01 & 101 & 152 & 4 & 0.075 & 0.030  \\
 \hline 
 $50\%$-ltl2dba01 & 101 & 152 & 4 & 1.4 & 0.028 \\
 \hline \rowcolor{Highlight}
 ltl2dba14 & 97 & 144 & 4 & 0.18 & 0.016 \\
 \hline
 $30\%$-ltl2dba14 & 97 & 144 & 4 & 0.181 & 0.013 \\
 \hline 
 $50\%$-ltl2dba14 & 97 & 144 & 4 & 0.574 & 0.012 \\
 \hline \rowcolor{Highlight}
 ltl2dba22 & 21 & 30 & 4 & 0.037 & 0.002 \\
 \hline
 $30\%$-ltl2dba22 & 21 & 30 & 4 & 0.036 & 0.002 \\
 \hline 
 $50\%$-ltl2dba22 & 21 & 30 & 4 & 0.03 & 0.0009 \\
 \hline \rowcolor{Highlight}
 prioritized$\_$arbiter$\_$unreal2 & 851 & 1412 & 4 & 15.8 & 0.73 \\
 \hline
 $30\%$-prioritized$\_$arbiter$\_$unreal2 & 851 & 1412 & 4 & 16 & 0.759 \\
 \hline 
 $50\%$-prioritized$\_$arbiter$\_$unreal2 & 851 & 1412 & 4 & 126 & 1.2 \\
 \hline \rowcolor{Highlight}
 lilydemo17 & 3102 & 5334 & 7 & 1237 & 41 \\
 \hline
 $30\%$-lilydemo17 & 3102 & 5334 & 7 & Timeout & 41 \\
 \hline 
 $50\%$-lilydemo17 & 3102 & 5334 & 7 & Timeout & 24 \\
 \hline \rowcolor{Highlight}
 lilydemo18 & 449 & 728 & 9 & 220 & 0.6 \\
 \hline
 $30\%$-lilydemo18 & 449 & 728 & 9 & 224 & 0.621 \\
 \hline 
 $50\%$-lilydemo18 & 449 & 728 & 9 & Timeout & 0.552 \\[1ex] 
 \hline
%\end{tabular}
%\caption{Performance comparison between parity examples solved via \texttt{N-FP} and \texttt{N-ZL} (the rows with the original example's names),
%and \Odd-fair instances of the example with 30$\%$- and 50$\%$-liveness solved via \texttt{OF-FP} and \texttt{OF-ZL} (the rows with the original example's names prefixed by \enquote{30$\%$-} or \enquote{50$\%$-}).}

\end{longtable} %\end{table}
% \end{center}
     
% 2 figures side-byy-side using minipages:
%\begin{figure}
%\centering
%\begin{minipage}{.5\textwidth}
%     \centering
%     \includegraphics[width=.4\linewidth]{image1}
%     \captionof{figure}{A figure}
%     \label{fig:test1}
%\end{minipage}%
%\begin{minipage}{.5\textwidth}
%     \centering
%     \includegraphics[width=.4\linewidth]{image1}
%     \captionof{figure}{Another figure}
%     \label{fig:test2}
%\end{minipage}
%\end{figure}

\newpage
\subsection{Additional material for Ex. \ref{ex:1}}\label{app:example}
Below we present an extended version of the fixed-point calculation in \eqref{equ:fpexample}, 
\begin{align*}
&   Y_4^{0} = \emptyset \\
&    \quad X_3^{0, 0} = V \\
&    \quad \quad Y_2^{0,0,0} = \emptyset \\
&    \quad \quad \quad X_1^{0,0,0,0} = V \\
&    \quad \quad \quad X_1^{0,0,0,1} = \Phi^{Y_4^{0}, X_3^{0, 0}, Y_2^{0,0,0}, X_1^{0,0,0,0} } = C_3 \cup C_1 \\
&    \quad \quad \quad X_1^{0,0,0,2} = \Phi^{Y_4^{0}, X_3^{0, 0}, Y_2^{0,0,0}, X_1^{0,0,0,1} } = C_3 \cup (C_1 \cap \Npre(Y_2^{0,0,0}, X_1^{0,0,0,1})) = C_3\\
&    \quad \quad \quad X_1^{0,0,0,3} = \Phi^{Y_4^{0}, X_3^{0, 0}, Y_2^{0,0,0}, X_1^{0,0,0,1} } = C_3 \cup (C_1 \cap \Npre(Y_2^{0,0,0}, X_1^{0,0,0,2})) = C_3\\
&    \quad \quad Y_2^{0,0,1} = X_1^{0,0,0,\infty} = C_3\\
&    \quad \quad \quad X_1^{0,0,1,0} = V \\
&    \quad \quad \quad X_1^{0,0,1,1} = \Phi^{Y_4^{0}, X_3^{0, 0}, Y_2^{0,0,1}, X_1^{0,0,0,0} } = C_3 \cup C_1 \cup \{2b\}\\
&    \quad \quad \quad X_1^{0,0,1,2} = \Phi^{Y_4^{0}, X_3^{0, 0}, Y_2^{0,0,1}, X_1^{0,0,0,1} } = C_3 \cup \{2b\}\\
&    \quad \quad \quad X_1^{0,0,1,3} = \Phi^{Y_4^{0}, X_3^{0, 0}, Y_2^{0,0,1}, X_1^{0,0,0,2} } = C_3 \cup \{2b\}\\
&    \quad \quad Y_2^{0,0,2} = X_1^{0,0,1,\infty} = C_3 \cup \{2b\}\\
&    \quad \quad \quad X_1^{0,0,2,0} = V \\
&    \quad \quad \quad X_1^{0,0,2,1} = \Phi^{Y_4^{0}, X_3^{0, 0}, Y_2^{0,0,2}, X_1^{0,0,0,0} } = C_3 \cup C_1 \cup \{2b, 2c\}\\
&    \quad \quad \quad X_1^{0,0,2,2} = \Phi^{Y_4^{0}, X_3^{0, 0}, Y_2^{0,0,2}, X_1^{0,0,0,1} } = C_3 \cup \{2b,2c\}\\
&    \quad \quad \quad X_1^{0,0,2,3} = \Phi^{Y_4^{0}, X_3^{0, 0}, Y_2^{0,0,2}, X_1^{0,0,0,2} } = C_3 \cup \{2b,2c\}\\
&    \quad \quad Y_2^{0,0,3} = X_1^{0,0,2,\infty} = C_3 \cup \{2b, 2c\}\\
&    \quad \quad \ldots\\
&    \quad \quad Y_2^{0,0,4} = X_1^{0,0,3,\infty} = C_3 \cup \{2b, 2c\}\\
&    \quad X_3^{0, 1} = Y_2^{0,0,\infty} =  C_3 \cup \{2b, 2c\} \\
&    \quad \quad Y_2^{0,1,0} = \emptyset \quad  \\
&    \quad \quad Y_2^{0,1,1} = X_1^{0,1,0,\infty} = \{3b\} \quad \\
&    \quad \quad Y_2^{0,1,2} = X_1^{0,1,1,\infty} = \{2b,3b\}\ \quad \\
&    \quad \quad Y_2^{0,1,3} = Y_2^{0,1,4} = X_1^{0,1,2,\infty} = X_1^{0,1,3,\infty} = \{2b,2c,3b\}\\
&    \quad X_3^{0, 2} = Y_2^{0,1,\infty} = \{2b,2c,3b\} \\
&    \quad \ldots \\
&    \quad X_3^{0, 3} = Y_2^{0,2,\infty} = \{2b,2c,3b\} \\
&    Y_4^{1} = X_3^{0,\infty} = \{2b,2c,3b\}\\
&    \quad X_3^{1, 0} = V \\
&    \quad \quad Y_2^{1,0,0} = \emptyset \\
&    \quad \quad Y_2^{1,0,1} = X_1^{1,0,0,\infty} = C_3 \cup C_4\\
\end{align*}
\begin{align*}
&    \quad \quad Y_2^{1,0,2} = X_1^{1,0,1,\infty} = C_3 \cup C_4 \cup \{2b\}\\
&    \quad \quad Y_2^{1,0,4} =  Y_2^{1,0,3} = C_1 \cup C_3 \cup C_4 \cup \{2b, 2c\}\\
&    \quad X_3^{1, 1} =  Y_2^{1,0,\infty} = C_1 \cup C_3 \cup C_4 \cup \{2b, 2c\}\\
&    \quad \quad Y_2^{1,1,0} = \emptyset \\
&    \quad \quad Y_2^{1,1,1} = C_3 \cup C_4 \\
&    \quad \quad Y_2^{1,1,2} = C_3 \cup C_4 \cup \{2b\}\\
&    \quad \quad Y_2^{1,1,3} = Y_2^{1,1,4} = C_1 \cup C_3 \cup C_4 \cup \{2b, 2c\}\\
&    \quad X_3^{1, 2} =  Y_2^{1,1,\infty} = C_1 \cup C_3 \cup C_4 \cup \{2b, 2c\}\\
&   Y_4^{2} = X_3^{1,\infty} = C_1 \cup C_3 \cup C_4 \cup \{2b, 2c\}\\
&   \ldots \\
&   Y_4^{3} = C_1 \cup C_3 \cup C_4 \cup \{2b, 2c\}
\end{align*}
And finally,
\begin{equation*}
   \mathcal{W}_{Odd} = Y_4^\infty = C_1 \cup C_3 \cup C_4 \cup \{2b, 2c\} = V \setminus \{2a\}
\end{equation*}

\end{document}